\documentclass[a4paper,11pt]{article}

\usepackage{amsfonts,qtree,stmaryrd,textcomp}
\usepackage{pifont,amssymb,amsmath,amsthm,tabularx,graphicx}
\usepackage{tree-dvips,pictexwd,dcpic}
\usepackage{bbm}
\usepackage{bm}

\usepackage{stmaryrd}

\usepackage[T1]{fontenc}
\usepackage[latin1]{inputenc}
\usepackage[text={17cm,26cm},centering]{geometry}

\usepackage{etex}
\usepackage{fancyhdr}
\usepackage{graphicx}
\usepackage{enumitem}
\usepackage{amsmath}
\usepackage[bbgreekl]{mathbbol}% order sensitive; it is before amssymb
\usepackage{amssymb}
\usepackage{amsthm}
\usepackage[all]{xy}

\usepackage{epigraph}
\RequirePackage{bussproofs}

\usepackage{mathabx}

\usepackage{scalerel,stackengine}

\usepackage{framed}
\usepackage{mathtools}
\usepackage{url}
\usepackage{proof}
\usepackage{xspace}
\usepackage{listings}
\usepackage{wrapfig}
\usepackage{tikz}
\usepackage{tikz-cd}

\usepackage{relsize}

\usetikzlibrary{positioning}
\usetikzlibrary{shapes}

\usetikzlibrary{arrows}
\usetikzlibrary{calc}
\usepackage{tikz-3dplot}
\usetikzlibrary{decorations.pathmorphing}

\usepackage{amsfonts,amssymb,amsmath,qtree,stmaryrd,textcomp, amsthm}

\usetikzlibrary{decorations.pathreplacing}
\tikzset{axis/.style={&lt;-&gt;}}

\stackMath
\newcommand\reallywidehat[1]{%
\savestack{\tmpbox}{\stretchto{%
  \scaleto{%
    \scalerel*[\widthof{\ensuremath{#1}}]{\kern-.6pt\bigwedge\kern-.6pt}%
    {\rule[-\textheight/2]{1ex}{\textheight}}%WIDTH-LIMITED BIG WEDGE
  }{\textheight}% 
}{0.5ex}}%
\stackon[1pt]{#1}{\tmpbox}%
}
\usepackage{cancel}

 \definecolor{MyBlue}{rgb}{0.05, 0.25, 0.65}
 
 \definecolor{MyRed}{rgb}{0.90, 0.05, 0.05}
\definecolor{MyGreen}{rgb}{0.05, 0.90, 0.05}

\newcommand{\B}{\boldsymbol}
\newcommand{\C}[1]{\mathcal{#1}}
\newcommand{\D}[1]{\mathbb{#1}}
\newcommand{\M}[1]{\mathscr{#1}}

\usepackage{ mathrsfs }

\newtheorem{theorem}{Theorem}[section]
\newtheorem{proposition}[theorem]{Proposition}

\newtheorem{corollary}[theorem]{Corollary}
\newtheorem{remark}[theorem]{Remark}

\newtheorem{example}[theorem]{Example}

\newtheorem{definition}[theorem]{Definition}

\newcommand{\Nat}{{\mathbb N}}
\newcommand{\Real}{{\mathbb R}}

\newcommand{\id}{\mathrm{id}}

\newcommand{\AC}{\mathrm{AC}}

\newcommand{\BISH}{\mathrm{BISH}}

\newcommand{\CST}{\mathrm{CST}}

\newcommand{\dom}{\mathrm{dom}}

\newcommand{\PEM}{\mathrm{PEM}}

\newcommand{\CLASS}{\mathrm{CLASS}}

\newcommand{\CC}{\mathrm{CC}}

\newcommand{\TOT}{\Leftrightarrow}

\newcommand{\To}{\Rightarrow}

\newcommand{\sto}{\rightsquigarrow}

 \newcommand{\pto}{\rightharpoonup}

\newcommand{\CZF}{\mathrm{CZF}}

\newcommand{\pr}{\textnormal{\texttt{pr}}}

\newcommand{\BST}{\mathrm{BST}}

\newcommand{\Ker}{\textnormal{\texttt{Ker}}}

\newcommand{\Disj}{\mathbin{\B ) \B (}}

\newcommand{\Set}{\mathrm{\mathbf{Set}}}

\newcommand{\Ineq}{\textnormal{\texttt{Ineq}}}

\newcommand{\SetIneq}{\textnormal{\textbf{SetIneq}}}

\newcommand{\SetExtIneq}{\textnormal{\textbf{SetExtIneq}}}

\newcommand{\EFQ}{\mathrm{EFQ}}

\newcommand{\emptys}{\cancel{\mathlarger{\mathlarger{\mathlarger{\mathlarger{\square}}}}}}

\newcommand{\Fun}{\textnormal{\textbf{Fun}}}

\newcommand{\StrExtFun}{\textnormal{\textbf{StrExtFun}}}

\newcommand{\LinIneq}{\textnormal{\texttt{LinIneq}}}

\newcommand{\InnProd}{\textnormal{\texttt{InnProd}}}

\newcommand{\Norm}{\textnormal{\texttt{Norm}}}

\newcommand{\coperp}{\mathrel{\top}}

\newcommand{\ti}{\mathrm{I}}

\newcommand{\loc}{\textnormal{\texttt{loc}}}

\newcommand{\tota}{\textnormal{\texttt{tot}}}

\newcommand{\BC}{\mathrm{BC}} 

\newcommand{\CoQL}{\mathrm{IntQL}} 

\newcommand{\ClQL}{\mathrm{ClassQL}} 

\newcommand{\ComplQL}{\mathrm{ComplQL}}

\newcommand{\MANC}{\mathrm{AC!}} 

\newcommand{\Hilb}{\textnormal{\textbf{Hilb}}}

\newcommand{\LinIsom}{\textnormal{\textbf{LinIsom}}}

\newcommand{\LinContr}{\textnormal{\textbf{LinContr}}}

\newcommand{\Closure}{\textnormal{\texttt{Closure}}}
\newcommand{\Complete}{\textnormal{\texttt{Complete}}}
\newcommand{\Cauchy}{\textnormal{\texttt{Cauchy}}}

\newcommand{\PL}{\mathrm{PL}} 

\newcommand{\spart}{\textnormal{\texttt{spart}}}

\begin{document}

%\fonttable{MnSymbolA10}

\date{}

%\title{\textbf{A constructive least number principle}}

\title{\textbf{Orthocomplemented subspaces and partial projections on a Hilbert space}}

%\title{\textbf{The Ex falso Principle in Constructive Mathematics}}
%\titlecomment{{\lsuper*}This paper is a major extension of~\cite{Pe17}.}

\author{Iosif Petrakis\\	%required
Department of Computer Science, University of Verona\\
iosif.petrakis@univr.it}  %optional
%\thanks{thanks 1, optional.}	%optional

% \author[B.~Name2]{Bob Name2}	%optional
% \address{address2; addresses should initially be duplicated, even if
%   authors share an affiliation}	%optional
% \email{name2@email2; ditto for email addresses}  %optional
% \thanks{thanks 2, optional.}	%optional
% 
% \author[C.~Name3]{Carla Name3}	%optional
% \address{address 3}	%optional
% \urladdr{name3@url3\quad\rm{(optionally, a web-page can be specified)}}  %optional
% \thanks{thanks 3, optional.}	%optional

%% etc.

%% required for running head on odd and even pages, use suitable
%% abbreviations in case of long titles and many authors:

%%%%%%%%%%%%%%%%%%%%%%%%%%%%%%%%%%%%%%%%%%%%%%%%%%%%%%%%%%%%%%%%%%%%%%%%%%%

%% the abstract has to PRECEDE the command \maketitle:
%% be sure not to issue the \maketitle command twice!

\maketitle

\begin{abstract}
\noindent 
We introduce the notion of an orthocomplemented subspace of a Hilbert space $\C H$, that is, a pair of orthogonal closed subspaces of $\C H$, as a two-dimensional counterpart to the one-dimensional notion of a closed subspace of $\C H$. Orthocomplemented subspaces are the Hilbert space-analogue to Bishop's complemented subsets. To complemented subsets correspond their characteristic functions, which are partial, Boolean-valued functions. Similarly, to orthocomplemented subspaces of $\C H$ correspond partial projections on $\C H$. 
Previous work of Bridges and Svozil on constructive quantum logic is an one-dimensional approach to the subject. The lattice-properties of the orthocomplemented subspaces of a Hilbert space is a  two-dimensional approach to constructive quantum logic, that we call complemented quantum logic. Since the negation of an orthocomplemented subspace is formed by swapping its components, complemented quantum logic, although constructive, is closer to classical quantum logic than the constructive quantum logic of Bridges and Svozil. The introduction of orthocomplemented subspaces and their corresponding partial projections allows a new approach to the constructive theory of Hilbert spaces. For example, the partial projection operator of an orthocomplemented subspace and the construction of the quotient Hilbert space bypass the standard restrictive hypothesis of locatedness on a subspace. Located subspaces correspond to total orthocomplemented subspaces.
\end{abstract}


\section{Introduction}
\label{sec: intro}

If $A$ is a subset of a set $(X, =_X)$, its characteristic function $\chi_A \colon X \to \D 2$ is defined through the classical principle of the excluded middle $(\PEM)$ by the rule: 
%the rule
$$\chi_A(x) := \left\{ \begin{array}{ll}
	1   &\mbox{, $x \in A$}\\
	0             &\mbox{, $x \notin A$.}
\end{array}
\right. $$
Constructively, such a definition is acceptable if and only if $A$ is a \textit{decidable} subset of $X$, i.e., 
$\forall_{x \in X}\big(x \in A \vee x \notin A\big)$. As a result, we cannot recover constructively the classical bijection between the powerset $\C P(X)$ of $X$ and the set $\D F(X, \D 2)$ of total functions from $X$ to the set of booleans $\D 2$. A recovery of this bijection is possible constructively, if we work with complemented subsets, instead of subsets, and with their characteristic functions, which are \textit{partial}, Boolean-valued functions on $X$. A \textit{complemented subset} is a pair $\B A:= (A^1, A^0)$ of subsets $A^1, A^0$ of a set $(X, =_X)$ with an extensional inequality $\neq_X$ (see Definition~\ref{def: apartness}), a strong and positive form of negation of the equality $=_X$ on $X$, such that $A^1$ and $A^0$ are subsets of $X$ that are disjoint in a strong and positive way, i.e., 
$$A^1 \Disj A^0 :\TOT \forall_{x^1 \in A^1}\forall_{x^0 \in A^0}\big(x^1 \neq_X x^0\big).$$
The \textit{charactersitic function} of $\B A$ is a partial function $\chi_{\B A} \colon X  \pto \D 2$, i.e., it is a (total) function $\chi_{\B A} \colon \dom(\B A) \to \D 2$ on the subset $\dom(\B A)$ of $X$, where 
$$\dom(\B A) := A^1 \cup A^0,$$
and it is defined by the rule:
$$\chi_{\B A}(x) := \left\{ \begin{array}{ll}
	1   &\mbox{, $x \in A^1$}\\
	0             &\mbox{, $x \in A^0,$}
\end{array}
\right. $$
avoiding the use of $\PEM$. It is easy to show that $\chi_{\B A}$ is a strongly extensional function (see Definition~\ref{def: se}). The \textit{complemented powerset} $\B P(X)$ of $X$ is a poset with
$$\B A \subseteq \B B :\TOT A^1 \subseteq B^1 \ \& \ B^0 \subseteq A^0,$$
i.e., $\B B$ has more ``provers'' and less ``refuters'' than $\B A$. The order on the characteristic functions of complemented subsets is given by
$$\chi_{\B A} \leq \chi_{\B B} :\TOT \dom(\B A) \subseteq \dom(\B B) \ \& \ \forall_{x \in \dom(\B A)}(\chi_{\B A}(x) \leq_{\D 2} \chi_{\B B}(x)).$$
Conversely, every strongly extensional function $\chi$ from a subset $\dom(\chi)$ of $X$ to $\D 2$ induces a complemented subset of $X$, defined by $\B A_{\chi} := (A_{\chi}^1, A_{\chi}^0)$, where 
$$A_{\chi}^1 := \{x \in \dom(\chi) \mid \chi(x) =_{\D 2} 1\},$$
$$A_{\chi}^0 := \{x \in \dom(\chi) \mid \chi(x) =_{\D 2} 0\}.$$
One can then show, see~\cite{PW22}, that the two procedures are inverse to each other that respect the corresponding partial orders, and hence there is a bijection between the complemented powerset $\B {\C P}(X)$ of $X$ and the totality $\C F(X, \D 2)$ of Boolean-valued partial functions on $X$. This bijection captures the computational content of the aforementioned classical bijection between $\C P(X)$ and $\D F(X, \D 2)$. We call the study of subsets within constructive logic the \textit{one-dimensional} approach to the theory of subsets, and we call the study of complemented subsets within constructive logic the \textit{two-dimensional} approach to the theory of subsets. Clearly, within the two-dimensional framework of complemented subsets, partial, Boolean-valued functions become first-class citizens. These concepts were introduced by Bishop and Cheng in~\cite{BC72}, their constrructive development of the Daniell approach~\cite{Da18} to measure theory that was seriously extended later in~\cite{BB85}. In order to recover constructively the definition of the measure $\mu(A)$ of a measurable set $A$ through the integral 
$$\mu(A) := \int \chi_A,$$
Bishop and Cheng used instead complemented measurable sets and integrals on partial functions, i.e.,
$$\mu(\B A) := \int \chi_{\B A}.$$
Bishop-Cheng measure theory is the most developed constructive theory of measure and integration (see also~\cite{Br79, Ch21, PZ24, GP25}).

The classical correspondence between subsets of a set $(X, =_X)$ and characteristic functions on $X$ is in complete analogy to the classical  correpsondence between subspaces of a Hilbert space $\C H$ and projections on $\C H$. Classically, every (closed) subspace $L$ of $\C H$ induces a (total) projection $P_L \colon H \to H$, i.e., an idempotent and self-adjoint operator on $\C H$. Conversely, an idempotent and self-adjoint operator $P$ on $\C H$ induces a closed subspace $L_P$ of $\C H$, defined by
$$L_P := 
%\{P(x) \mid x \in H\} =_{S(\C H)} 
\{x \in H \mid P(x) =_H x\}.$$
The partial orders on the subspaces $S(\C H)$ of $\C H$ and the set of projections $\Pr(\C H)$ on $\C H$
are defined , respectively, by the rules:
$$L \leq M :\TOT L \subseteq M,$$
$$P_L \leq P_M :\TOT \forall_{x \in H}\big(\langle P_L(x), x\rangle \leq \langle P_M(x), x \rangle\big).$$
The two procedures are inverse to each other, and hence, classically, there is a bijection between $S(\C H)$, the set of closed subspaces of $\C H$, and $\Pr(\C H)$, the set of total projections on $\C H$, that preserves the corresponding partial orders.
Constructively though, such a bijection cannot be captured within the one-dimensional framework of subspaces and projections, since the existence of the projection $P_L$ is equivalent to the locatedness of $ L$, i.e., the existsence of the distance
$$\rho(x, L) := \inf\big\{||x -l|| \mid l \in L\big\}$$
of $x$ from $L$, for every $x \in H$. Constructively, we cannot accept that every subspace of $L$ is located (see Proposition~\ref{prp: Brouwerian}). 
Clearly, \textit{locatedness of a subspace of $\C H$ corresponds to the decidability of a subset $A$ of a set} $X$. 

Here we extend Bishop's two-dimensional approach from subsets of a set to subspaces of a Hilbert space. Namely, we introduce what we call the \textit{orthocomplemented subspaces} of $\C H$, i.e., pairs $\B L :=(L^1, L^0)$ of orthogonal subspaces of $\C H$. To every orthocomplemented subspace $\B L$ of $\C H$ corresponds a \textit{partial projection} $P_{L^1} \colon \dom(\B L) \to H$ on $\C H$, where
$$\dom(\B L) := L^1 \vee L^0 := \overline{L^1 + L^0}.$$
Actually, we have that $\dom(\B L) =_{\mathsmaller{S(\C H)}} L^1 + L^0$ (see Theorem~\ref{thm: crucial}(i)). Conversely, to every partial projection $P \colon \dom(P) \leq \C H \to \C H$, i.e., a self-adjoint, idempotent and bounded partial operator on $\C H$, corresponds an orthocomplemented subspace $\B L_P := (L_{P^1}, L_{P^0})$, where 
$$L_{P^1} := \{x \in \dom(P) \mid P(x) =_H x\},$$
$$L_{P^0} := \{x \in \dom(P) \mid P(x) =_H 0\}.$$
The two procedures are inverse to each other, i.e., there is a bijection between $\B S(\C H)$, the totality of orthocomplemented subspaces of $\C H$, and $\M P(\C H)$, the totality of partial projections on $\C H$ (see Theorem~\ref{thm: bijection}). The computational content of the aforementioned classical biljection between subspaces of $\C H$ and (total) projections on $\C H$ is captured by the constructive bijection between orthocomplemented subspaces of $\C H$ and partial projections on $\C H$. As in the constructive measure theory of Bishop and Cheng, within the two-dimensional approach to the theory of Hilbert spaces partial operators become first-class citizens.

\begin{table}
	\begin{center}
		\caption{Correspondence between subsets of a set and subspaces of a Hilbert space}
		\label{decloc}
		\vspace{2mm}
		\begin{tabular}{l@{\hspace{1.0cm}}l}
			\itshape Subsets of a set $X$  & \itshape Subspaces of a Hilbert space $\C H$\\
			\hline
			$A \subseteq X$ decidable &  $L \leq \C H$ located \\
			
			$\chi_A$ is total with $\PEM$  &    $P_L$ is total with $\PEM$ \\
			
			$P(X)$ is bijective to $\D F(X, \D 2)$ & $S(\C H)$ is bijective to $\Pr(\C H)$\\
			
			$\B A := (A^1, A^0) \subseteq X$  &    $\B L := (L^1, L^0) \leq \C H$  \\
			
			$\chi_{\B A} \colon \dom(\B A) \to \D 2$ partial function & $P_{\B L}^1 \colon \dom(\B L) \to \C H$ partial projection\\
			
			$\B P(X)$ is bijective to $\C F(X, \D 2)$ &  $\B S(\C H)$ is bijective to $\M P(\C H)$\\\hline
			
		\end{tabular}
	\end{center}
\end{table}
%The ``internal'' logic of a mathematical structure

Important case-studies within the two-dimensional  ``complemented approach'' to negation of mathematical concepts are found in many areas of mathematics, other than Bishop-Cheng measure theory. See, for example, the work of Coquand and Lombardi~\cite{CL17} in the theory of prime ideals,, the study of Borel and Baire sets in Bishop spaces~\cite{Pe19a}, the development of topologies of open complemented subsets in~\cite{Pe24a} in constructive topology, and the theory of swap algebras and swap rings in~\cite{MWP24, MWP25} that generalise the theory of Boolean algebras and Boolean rings.

%Formal analogy with complemented subsets
%
%Paper on orthopairs in the swap file
%

 We structure this paper as follows: 
 \vspace{-2mm}
 \begin{itemize}
 	
 \item In section~\ref{sec: basic} we include some fundamental notions and facts from the constructive theory of Hilbert spaces that are going to be used in the rest of the paper. Speciffically, we describe the basic operations on subspaces of a Hilbert space which, if treated classically, lead to the notion of a classical orhtomodular lattice, or of classical quantum logic $(\ClQL)$.
 
 \item In section~\ref{sec: sub}, and based on the previous work of Bridges and Svozil~\cite{BS00}, we present the basic lattice-properties of the subspaces of a Hilbert space $\C H$ from the constructive point of view. To $(\ClQL)$ Bridges and Svozil correspond a constructive orthomodular lattice, or a constructive quantum logic $(\CoQL)$. Crucial to the formulation of $\CoQL$ is the notion of a located subspace $L$ of $\C H$.
 %, i.e., of a subspace of $\C H$ for which one can compute the distance of any point in $\C H$ from $L$. 
 In Theorem~\ref{thm: loctot} we characterise a located subspace of $\C H$ as a total element of $S(\C H)$. This lattice-theoretic characterisation of locatedness allows a new formulation of quantum logic from a constructive pont of view (Definition~\ref{def: CoQL}).

\item In section~\ref{sec: orthosub} we introduce the notion of an orthocomplemented subspace  $\B L$ of a Hilbert space $\C H$.
%, i.e., of a pair $\B L := (L^1, L^0)$, where $L^1, L^0$ are orthogonal subspaces of $\C H$. 
Orthocomplemented subspaces are a two-dimensional counterpart to the one-dimensional subspaces of $\C H$, and the Hilbert space analogue to Bishop's complemented subsets. In Theorem~\ref{thm: crucial} we prove the existence of a partial projection operator $P_{\B L}^1$, defined on the domain
$L^1 \vee L^0$ of $\B L$, without restricting to located subspaces, as in the standard one-dimensional constructive approach of Bishop and Bridges. As in the case of one-dimensional subspaces of $\C H$, the locatedness of an orthocomplemented subspace $\B L$ is equivalent to its totality,  i.e., $\dom(\B L)$ is the whole Hilbert space $\C H$.  While the approach of Bishop and Bridges is a direct constructive translation of the classical proof of the existence of the projection acting on a Hilbert space that corresponds to an arbitrary closed subspace, our proof of Theorem~\ref{thm: crucial} avoids countable choice.
In Proposition~\ref{prp: quotient} we define the quotient Hilbert space $\dom(\B L)/\B L$ avoiding again the hypothesis of locatedness for $L$. If $\B L$ is total, and hence $L^1, L^0$ are located, we get as a limiting case the standard constructive construction of 
%Bishop and Bridges 
the quotient Hilbert space $\C H/L^1$.

\item In section~\ref{sec: bpo} we introduce bounded partial operators between normed spaces, i.e., bounded operators defined on a subspace of their domain. The basic operations on their totality (Definition~\ref{def: partialoperations}) motivate the definition of a \textit{partial linear space} (Definition~\ref{def: pls}), a generalisation of a linear space, to every element $x$ of which corresponds a zero-element $0_x$, such that $0 \cdot x =_X 0_x$. The total elements of a partial linear space, i.e., those that satisfy the equality $0_x =_X 0$, form a (total) linear space.

\item In section~\ref{sec: partialproj}  we study 
%bounded partial operators on a Hilbert space $\C H$. 
partial projections on a Hilbert space $\C H$, i.e., self-adjoint, idempotent and bounded partial operators on $\C H$. The projection operator $P_{\B L}^1$ shown to exist in Theorem~\ref{thm: crucial} is such a partial projection. In Theorem~\ref{thm: bijection} we show the existence of a bijection between the orthocomplemented subspaces of $\C H$ and the partial
projections on $\C H$ that preserves the corresponding partial orders. This result captures the computational content of the classical bijection between the subspaces of $\C H$ and the total projections on $\C H$, which cannot be fully grasped within the one-dimensional constructive approach of Bishop and Bridges, due to the hypothesis of locatedness.

 \item In section~\ref{sec: typeI} we define the basic operations 
 %of type $(\ti)$ 
 on $\B S(\C H)$, and we prove that $\B S(\C H)$ is a complemented quantum lattice (Proposition~\ref{prp: ComplQL}). As the negation of an orthocomplemeted subspace $(L^1, L^0)$ of $\C H$ is the swapped subspace $(L^0, L^1)$, complemented quantum logic $(\ComplQL)$ satisfies all properties of classical negation $L^{\perp}$ in $\ClQL$. Since many of these properties cannot be accepted in $\CoQL$, the constructive logical system $\ComplQL$ is closer to $\ClQL$ than $\CoQL$. Moreover, the proof of Proposition~\ref{prp: ComplQL} seems to be quite informative 
 (notice that the proof of the corresponding Proposition~\ref{prp: CoQL} is not given in~\cite{BS00}).

  \item In section~\ref{sec: commuting} we extend basic properties of commuting total projections to commuting partial ones (Proposition~\ref{prp: comm1}).
  
% \item In section~\ref{sec: typeII} we extend to $\B S(\C H)$ the second type of join and meet on complemented subsets introduced by Bishop and Cheng in~\cite{BC72}. Their motivation for introducing these alternative operations of join and meet was their correspondence to the join and meet of the corrsponding characteristic functions, something which was not the case for the join and meet of type $(\ti)$. This correspondence was crucial to the constructive counterpart to the Daniell approach to inegration and measure that was first developed in~\cite{BC72}. We ...........
%   

 \end{itemize}
 \vspace{-2mm}

 We work within Bishop Set Theory $(\BST)$, a semi-formal system for Bishop's informal system of constructive mathematics $\BISH$ that behaves as a high-level programming language. For all notions and results of $\BST$ that are used here without definition or proof we refer to~\cite{Pe20, Pe22a, Pe24}. The type-theoretic interpretation of Bishop's set theory into the theory of setoids (see especially the work of Palmgren~\cite{Pa12}--\cite{Pa19}) has become nowadays the standard way to understand Bishop sets. Other formal systems for BISH are Myhill's Constructive Set Theory $(\CST)$, introduced  in~\cite{My75}, and Aczel's system $\CZF$ (see~\cite{AR10}).
% For all notions and results from Bishop's theory of sets that are used here without explanation or proof, we refer to~\cite{Pe20, Pe22a, Pe24}.  
 For all notions and results from constructive analysis within $\BISH$ that are used here without explanation or proof, we refer to~\cite{Bi67, BB85, BR87, BV06}.

 \section{Subspaces of a Hilbert space and classical quantum logic}
 \label{sec: basic}

 In this section we include some fundamental notions and facts from the constructive theory of Hilbert spaces that are going to be used in the rest of the paper. Especially, we describe the basic operations on subspaces of a Hilbert space which, if treated classically, lead to the notion of a classical orhtomodular lattice, or of a classical quantum logic $(\ClQL)$.

 Bishop Set Theory (BST) is an informal, constructive theory of \emph{totalities} and \emph{assignment routines} between totalities that accommodates BISH and serves as an intermediate step between Bishop's original theory of sets and an adequate and faithful formalisation of BISH in Feferman's sense~\cite{Fe79}. Totalities, apart from the basic, undefined set of natural numbers $\mathbb{N}$, are defined through a membership-condition. The \emph{universe} $\mathbb{V}_0$ of (predicative) sets is an open-ended totality, which is not considered a set itself, and every totality the membership-condition of which involves quantification over the universe is not considered a set, but a proper class. \emph{Sets} are totalities the membership-condition of which does not involve quantification over $\mathbb{V}_0$, and are equipped with an equality relation i.e., an equivalence relation. An equality relation $x =_X x{'}$ on a defined set $X$ is defined through a formula $E_X(x, x{'})$ and a proof that $E_X(x, x{'})$ satisfies the properties of an equivalence relation.
 Assignment routines are of two kinds: \emph{non-dependent} ones and \emph{dependent} ones. 
 
 \begin{definition}\label{def: function}
 	If $(X, =_X)$ and $(Y, =_Y)$ are sets, a function $f \colon X \to Y$ is a non-dependent assignment routine $f \colon X \sto Y$ i.e., for every $x \in X$ we have that $f(x) \in Y$, such that $x =_X x{'} \To f(x) =_Y f(x{'})$, for every $x, x{'} \in X$. We denote by $\D F(X, Y)$ the set of functions from $X$ to $Y$, where\footnote{In this way  the type-theoretic function extensionality axiom is built in as the canonical equality of the function set.}
 	$$f =_{\D F(X,Y)} g :\TOT \forall_{x \in X}\big(f(x) =_Y g(x)\big).$$
 	Let $\id_X$ be the identity function on $X$, and let $(\Set, \Fun)$ be the category of sets and functions.
 	%$(\Set, \Emb)$ its subcategory with embeddings and $(\Set, \Surj)$ its subcategory with surjections. 
 	Two sets are equal in $\mathbb{V}_0$ if there is $g \colon Y \to X$, such that\footnote{With this equality the universe in $\BST$ can be called \emph{univalent} in the sense of Homotopy Type Theory~\cite{Ri22, Ho13}, as, by definition, an equivalence between sets is an equality.} $g \circ f =_{\D F(X, X)} \id_X$ and $f\circ g =_{\D F(Y, Y)} \id_Y$.
 \end{definition}

 \begin{definition}\label{def: apartness}
 	Let the following formulas with respect to a relation $x \neq_X y$ on the set $(X, =_X)${$:$}\\[1mm]
 	$(\Ineq_1) \  \forall_{x, y \in X}\big(x =_X y \ \& \ x \neq_X y \To \bot \big)$, where $\bot := 0 =_{\Nat} 1$.\\[1mm]
 	$(\Ineq_2) \ \forall_{x, x{'}, y, y{'} \in X}\big(x =_X x{'} \ \& \ 
 	y =_X y{'} \ \& \ x \neq_X y \To x{'} \neq_X y{'}\big)$,\\[1mm]
 	$(\Ineq_3) \ \forall_{x, y \in X}\big(\neg(x \neq_X y) \To x =_X y\big)$,\\[1mm]
 	$(\Ineq_4)  \forall_{x, y \in X}\big(x \neq_X y \To y \neq_X x\big)$,\\[1mm]
 	$(\Ineq_5) \ \forall_{x, y \in X}\big(x \neq_X y \To \forall_{z \in X}(z \neq_X x \ \vee \ z \neq_X y)\big)$,\\[1mm]
 	$(\Ineq_6) \ \forall_{x, y \in X}\big(x =_X y \vee x \neq_X y\big)$.\\[1mm]
 	If $\Ineq_1$ is satisfied, we call $\neq_X$ an 	inequality on $X$, and the 
 	structure $\C X := (X, =_X, \neq_X)$ a set with an inequality. 
 	If $(\Ineq_6)$ is satisfied, then $\C X$ is 
 	\textit{discrete}, if $(\Ineq_2)$ holds, $\neq_X$ is 
 	called \textit{extensional}, and  if $(\Ineq_3)$ holds, it is \textit{tight}. If $\neq_X$ satisfies $(\Ineq_4)$ and $(\Ineq_5)$, it is called an \textit{apartness relation} on $X$.
 \end{definition}

 The primitive inequality $\neq_{\Nat}$ is a discrete, tight apartness relation on $\Nat$. 

\begin{definition}\label{def: se} If $\C X := (X, =_X, \neq_X)$ and $\C Y := (Y, =_Y, \neq_Y)$ are sets with an inequality, a function $f \colon X \to Y$ is \textit{strongly extensional},
 if $f(x) \neq_Y f(x{'}) \To x \neq_X x{'}$, for every $x, x{'} \in X$, and it is an \textit{injection}, if the converse implication holds i.e., $x \neq_X x{'} \To f(x) \neq_Y f(x{'})$, for every $x, x{'} \in X$. 
 Let $(\SetIneq, \StrExtFun)$ be the category of sets with an inequality and strongly extensional functions, and let	$(\SetExtIneq, \StrExtFun)$ be its subcategory with sets with an extensional inequality.	
 \end{definition}

$(\SetExtIneq, \StrExtFun)$ is the main category of sets and functions in $\BST$.

\begin{definition}\label{def: canonicalineq}
	Let $\C X := (X, =_X, \neq_X)$ and $\C Y := (Y, =_Y, \neq_Y)$ be sets with inequalities.
	The inequalities on the product $X \times Y$ and the function space
	$\D F(X, Y)$ are given, respectively, by
	%$$(x, y) =_{X \times Y} (x{'}, y{'}) :\TOT x =_X x{'} \ \& \ y =_Y y{'},$$
	$$(x, y) \neq_{X \times Y} (x{'}, y{'}) :\TOT x \neq_X x{'} \vee y \neq_Y y{'},$$
%	$$f =_{\D F(X, Y)} g :\TOT \forall_{x \in X} \big[f(x) =_Y g(x)\big],$$
	$$f \neq_{\D F(X, Y)} g :\TOT \exists_{x \in X} \big[f(x) \neq_Y g(x)\big].$$
\end{definition}
The projections $\pr_X, \pr_Y$ associated to $X \times Y$ 
are strongly extensional functions.
%It is not easy to give examples of non strongly extensional functions. 
We cannot accept
in $\BISH$ that all functions are strongly extensional. E.g., the strong extensionality of 
all functions from a metric space to itself is equivalent to Markov's principle (see~\cite{Di20}, p.~40).
Notice that in order to show that a constant function between sets with an inequality is strongly extensional, one needs 
intuitionistic,  and not minimal, logic. An apartness relation $\neq_X$ on a set $(X, =_X)$ is always
extensional (see~\cite{Pe20}, Remark 2.2.6, p.~11). Here we work only with extensional subsets, as in Bridges-Richman Set Theory i.e., the set-theoretic framework of~\cite{BR87, MRR88, BV06}. The use of extensional subsets allows a smoother treatment of the strong empty subset of a set (its definition follows Definition~\ref{def: proper}), and hence a calculus of subsets closer to that of subsets within classical logic.

\begin{definition}\label{def: linearspace}
A linear space over the field $\D K \in \{\Real, \D C\}$ is a structure $\C X := (X, =_X, \neq_X; +, 0, \cdot)$, where $(X, =_X, \neq_X) \in \SetExtIneq$ and $(+, 0, \cdot)$ is the linear structure of $\C X$ that satisifies the standard axioms, together with the following compatibility axioms\footnote{These axioms are also found in~\cite{BV06}, pp.~47-48.} between $\neq_X$ and $(+, 0, \cdot):$\\[1mm]
$(\LinIneq_1)$ $\forall_{x, x{'} \in X}\big(x \neq_X x{'} \TOT x - x{'} \neq_X 0\big)$.\\[1mm]
$(\LinIneq_2)$ $\forall_{x, x{'} \in X}\big(x + x{'} \neq_X 0 \To x \neq_X 0 \vee x{'} \neq_X 0\big)$.\\[1mm]
$(\LinIneq_3)$ $\forall_{x\in X}\forall_{k \in \D K}\big(k \cdot x \neq_X 0 \To k \neq_{\D K} 0 \ \& \ x \neq_X 0\big)$.\\[1mm]
We call $\C X$ strict, if there is $x \in X$ with $x \neq_X 0$.  Let $X_0 := \{x \in X \mid x \neq_X 0\}$.
\end{definition}

%\textit{Throughout this section $\C X$ and $\C Y$ denote linear spaces}.
Clearly, if $\C X, \C Y$ are linear spaces, a linear map $f \colon X \to Y$ is strongly extensional if and only if 
$$\forall_{x \in X}\big(f(x) \neq_Y 0 \To x \neq_X 0\big).$$
It is easy to show that the operations $+, \cdot$ of a linear space $\C X$ are strongly extensional, such that the following property holds
$$(\LinIneq_4) \ \ \ \ \ \ \ \ \ \ \ \ \ \ \ \ \ \ \ \ \ \ \ \ \ \ \ \ \ \ \ \ \ \ \ \ x \neq_{X} x{'} \TOT \forall_{y \in X}(x + y \neq_X x{'} + y). \ \ \ \ \ \ \ \ \ \ \ \ \ \ \ \ \ \ \ \ \ \ \ \ \ \ \ \ \ \ \ \ \ \ \ \ \ \ \ \ \ \ \ \ \ \ \ $$

%\begin{corollary}\label{cor: lscor1}
%If $\C X$ is a linear space, the following hold:\\[1mm]
%\normalfont (i) 
%\itshape The operations $+, \cdot$ are strongly extensional functions.\\[1mm]
%\normalfont (ii) 
%\itshape $x \neq_{X} x{'} \TOT \forall_{y \in X}(x + y \neq_X x{'} + y)$.
%\end{corollary}
%
%\begin{proof}
%(i)
%\end{proof}

%\begin{definition}\label{def: linind}
%If $\C X$ is a linear space and $B \subseteq X$, then $B$ is linear independent, if ....
%\end{definition}

%If you add the compatibility axioms above, then you should consider the compatibility of linear maps with the inequality of the linear space. 
%
%\begin{definition}\label{def: slmap} 
%If $\C X, \C Y$ are linear spaces, a linear map $f \colon X \to Y$ is called strong, if it is a strongly extensional function i.e., 
%$$\forall_{x \in X}\big(f(x) \neq_Y 0 \To x \neq_X 0\big).$$
%The category of linear spaces with strong linear maps.
%\end{definition}
%
%
%\begin{definition}\label{def: base}
%A basis for a linear space $\C X$ is a linear independent subset $B$ of $X$, such that $X = <B>$.
%The pair $(\C X; B)$ is called a based linear space. If $(\C Y; C)$ is a based linear space, a basis-preserving strong linear map is a strong linear map $f \colon X \to Y$, such that $f(B) \subseteq C$. The category of based linear spaces and basis-preserving strong linear maps....
%\end{definition}
%

\begin{definition}\label{def: ips}
An inner product space is a pair $(\C X, \langle , \rangle)$, where $\C X$ is a linear space and $\langle , \rangle \colon X \times X \to \D K$ satisfies the following conditions:\\[1mm]
$(\InnProd_1)$ $\forall_{x \in X}\big(\langle x,x \rangle \in \Real \ \& \ \langle x,x \rangle \geq 0\big)$.\\[1mm]
$(\InnProd_2)$ $\forall_{x \in X}\big(<x,x> =_{\D K} 0 \TOT x =_X 0\big)$.\\[1mm]
$(\InnProd_3)$ $\forall_{x, x{'} \in X}\big(\langle x,x{'} \rangle  =_{\D K} \overline{\langle x{'}, x \rangle }\big)$.\\[1mm]
$(\InnProd_4)$ $\forall_{x, x{'}, y \in X}\forall_{k, k{'} \in \D K}\big(\langle k \cdot x + k{'} \cdot x{'}, y \rangle =_{\D K} k  \langle x,y \rangle + k{'} \langle x{'},y \rangle \big)$.\\[1mm]
$(\InnProd_5)$ $\forall_{x \in X}\big(x \neq_X 0 \TOT \langle x,x \rangle \neq_{\D K} 0\big)$.\\[1mm]
If $x, y \in X$, the relations of perpendicularity and coperpendicularity on $X$ are defined by
$$x \perp x{'} :\TOT \langle x, x{'} \rangle =_{\D K} 0,$$
$$x  \coperp x{'} :\TOT \langle x, x{'} \rangle \neq_{\D K} 0.$$
A Hilbert space is a separable and complete inner product space. Let $H_0 := \{x \in H \mid x \neq_H 0\}$.
A normed space is a pair
$(\C X, ||.||)$, where $\C X$ is a linear space and $||.|| \colon X \times X \to \Real$ satisfies the following conditions:\\[1mm] 
$(\Norm_1)$ $\forall_{x \in X}\big(||x|| \geq 0\big)$.\\[1mm]
$(\Norm_2)$ $\forall_{x \in X}\big(||x|| =_{\Real} 0 \To x =_X 0\big)$.\\[1mm]
$(\Norm_3)$ $\forall_{x \in X}\forall_{k \in \D K}\big(||k \cdot x|| =_{\Real} |k| ||x||\big)$.\\[1mm]
$(\Norm_4)$ $\forall_{x, x{'} \in X}\big(||x + x{'}|| \leq ||x|| + ||x{'}||\big)$.\\[1mm]
$(\Norm_5)$ $\forall_{x \in X}\big(x \neq_X 0 \TOT||x|| > 0\big)$.
\end{definition}

It is easy to show that the inner product function $\langle , \rangle$ is strongly extensional, and that the relations $\perp$ and $\coperp$on $\C X$ are extensional, i.e., they respect the corresponding equalities. Moreover, we have that $x \perp x \TOT x =_X 0$ and $x \coperp x \TOT x \neq_X 0$. \\[1mm]
\textit{Throughout this paper $\C H := (H, =_H, \neq_H; +, 0, \cdot; \langle , \rangle)$ is a Hilbert space}.

%The norm induced by $\langle , \rangle $ is defined by $||x|| := \sqrt{\langle x,x \rangle }$, for every $x \in X$.
%Cauchy-Schwarz...
%
%\begin{corollary}\label{cor: seip}
%
%\end{corollary}
%
%\begin{proof}
%
%\end{proof}

\begin{proposition}\label{prp: ineqhilbert}
If $x, y \in H$, then the following hold:\\[1mm]
\normalfont (i)
\itshape $x =_H 0 \TOT \forall_{z \in H}\big(\langle x, z \rangle =_{\D K} 0\big)$.\\[1mm]
\normalfont (ii)
\itshape $x \neq_H 0 \TOT \exists_{z \in H}\big(\langle x, z \rangle \neq_{\D K} 0\big)$.\\[1mm]
\normalfont (iii)
\itshape The inequality $x \neq_H y$ is a tight apartness relation.
\end{proposition}

\begin{proof}
(i) It follows trivially by $(\InnProd_2)$.\\
(ii) If $x \neq_H$, then by $(\InnProd_5)$ we have that $\langle x, x \rangle \neq_{\D K} 0$. Conversely, if $z\in H$ with $\langle x, z \rangle \neq_{\D K} 0$, then by the Cauchy-Schwartz inequality we get $0 < |\langle x, z \rangle| \leq ||x|| \cdot ||y||$, hence $||x|| \neq_{\Real} 0$.\\
(iii) We only show that $x \neq_H y$ is tight and cotransitive. For tightness, it suffices to show the implication $\neg(x \neq_H 0) \To x =_H 0$. By $(\InnProd_5)$ $\neg(x \neq_H 0) \TOT \neg(\langle x, x \rangle \neq_{\D K} 0)$. Since the inequality $\neq_{\D K}$ is tight, we get $\langle x, x \rangle =_{\D K} 0$, and hence $x =_H 0$. For cotransitivity, it suffices to show the implication
$$x \neq_H 0 \To \forall_{z \in H}\big(x \neq_H z \vee z \neq_H 0\big).$$
By the cotransitivity of $\neq_{\Real}$ we have that 
$$\langle x, x\rangle \neq_{\Real} 0 \To \forall_{z \in H}\big(\langle x, x\rangle   \neq_{\Real} \langle z, z\rangle  \vee \langle z, z\rangle  \neq_{\Real} 0\big).$$
If $\langle z, z\rangle  \neq_{\Real} 0$, then $z \neq_H 0$. If $\langle x, x\rangle   \neq_{\Real} \langle z, z\rangle $, then $||z|| \neq_{\Real} ||x||$, and by the inequality
$$0 < \big| ||x|| - ||z|| \big| \leq ||x - z||$$
we get $x -z \neq_H 0 \TOT x \neq_H z$.
\end{proof}

\begin{definition}\label{def: operator} If $(\C X, ||.||)$ and $(\C Y, ||.||)$ are normed spaces, a linear map $T \colon X \to Y$ is bounded, if there is $C > 0$, such that for every $x \in X$ we have that $||T(x)|| \leq C ||x||$. We denote by $\C B(\C X, \C Y)$ their set. A bounded linear map $T \colon X \to Y$ is called normed, or normable, if its norm
	$$||T|| := \sup \{||T(x)|| \mid x \in X \ \& \ ||x|| \leq 1\}$$
exists. If $\C H$ is a Hilbert space, a bounded linear map
$T \colon H \to H$ is called a bounded operator on $\C H$, and their set is denoted by $\C B(\C H)$. A projection on $\C H$ is an idempotent and self-adjoint bounded operator on $\C H$. Their set is denoted by $\Pr(\C H)$.
\end{definition}

If $T \in \C B(\C X, \C Y)$, then by the boundedness condition $T$ is strongly extensional. Not every bounded linear map is accepted to be normed constructively (see~\cite{BV06}, p.~53). If $\C X$ is strict, then $X$ has a unit vector, and consequently, if it is normed, then its norm is written as
%we also have that
$$||T|| =_{\Real} \sup \{||T(x)|| \mid x \in X \ \& \ ||x|| =_{\Real} 1\}.$$
The identity $I_H \colon H \to H$ and the zero map $0 \colon H \to H$ are projections on $\C H$. In~\cite{BB85} a subspace $L$ of a normed space $\C X$ is a located and closed linear subset of $X$. Locatedness is crucial to the standard constructive definition of the projection function $P_L$ that corresponds to $L$ (see~\cite{Bi67, BB85}). As already mentioned in the introduction, constructively we cannot accept that every subspace of a normed (or Hilbert) space is located.
%, as this would imply $\PEM$ (see~\cite{BV06}, p.~41).

\begin{definition}\label{def: subspace}
	A subspace $L$ of $H$, in symbols $L \leq \C H$, is a closed linear subset of $H$. Let $\C S(\C H)$ be the totality of subspaces of $\C H$ equipped with the equality inherited by the powerset of $H$. If $L, M \in S(\C H)$, let $L \leq M :\TOT L \subseteq M$, hence $L =_{\mathsmaller{\C S(\C H)}} M \TOT L \leq M \ \& \ M \leq L$. A subspace $L$ of $\C H$ is called located in $\C H$, or simply located, if for every $x \in H$ the distance $\rho(x, L) := \inf\{||x - l|| \mid l \in L\}$ exists in $\Real$. We denote by $S_{\loc}(\C H)$ the totality of located subspaces of $\C H$. If $A \subseteq H$, then $\overline{A}$ denotes the closure of $A$ in $\C H$. If $L \leq \C H$, we call $L$ strict, if there is $x \in L$ with $x \neq_H 0$.
\end{definition}

Classically, all subspaces of $\C H$ are located. This cannot be accepted constructively (see~\cite{BS00} and the subsequent Proposition~\ref{prp: Brouwerian}).
Although a finite dimensional linear subspace of a Hilbert space is always closed (see~\cite{Li81}, p.~197), there are infinite dimensional linear subspaces of a Hilbert spaces that are not. E.g., take the subspace of all sequences in $l^2$ with all but finitely many zeros (see~\cite{Co89}, p.~44).

\begin{definition}\label{def: orthocomplement}
	If $L \leq \C H$, the orthocomplement $L^{\perp}$ of $L$ is defined by
	$$L^{\perp} := \{x \in H \mid \forall_{y \in L}(x \perp y)\}.$$
	$L^{\perp \perp} := (L^{\perp})^{\perp}$ is the double orthocomplement of $L$.
\end{definition}

Clearly, $L^{\perp} \leq H$ and $L \cap L^{\perp} = \{0\}$. 
%It is immediate to show that $L^{\perp} \leq \C H$. 
If $L$ is located, then $L^{\perp}$ is also located (see~\cite{BB85}, p.~368). Due to Corollary~\ref{cor: corBBcrucial}(ii) the orthocoplement $L^{\perp}$ of $L$ is the inner product-analogue to the logical complement $A^{\neg}$ of a set $(X, =_X)$, or to the strong complement $A^{\neq}$ of a subset $A$ of a set $X \in \SetExtIneq$, where
$$A^{\neg} := \{x \in X \mid \forall_{a \in A}\big(\neg (x =_X a)\big)\},$$
$$A^{\neq} := \{x \in X \mid \forall_{a \in A}\big(x \neq_X a\big)\}.$$
Because of Corollary~\ref{cor: corBBcrucial}(ii) we define a proper subspace $L$ of $\C H$ not if $L^{\neq}$ is inhabited, 
%there is $x \in H$ with $x \in L^{\neq}$, 
but as follows.

\begin{definition}\label{def: proper} We call a subspace $L$ of $\C H$ proper, in symbols $L \lneq \C H$, if $\exists_{x \in H_0}\big(x \perp L\big)$.
\end{definition}

Since the zero subspace $\{0\}$, which is a simplified writing for $\{0_H\}$, can be written as
$$\{0_H\} := \{x \in H \mid x \perp x\},$$
it corresponds to the weak empty subset $\emptyset_X$ of a set $(X =_X)$ and to the strong empty subset $\emptys_X$ of $X \in \SetExtIneq$, where
$$\emptyset_X := \{x \in X \mid \neg(x =_X x)\},$$
$$\emptys_X := \{x \in X \mid x \neq_X x\}.$$

\begin{definition}\label{def: operations1} The following operations\footnote{As it is mentioned in~\cite{PM09}, p.~24, this implication does not satisfy the condition $a \rightarrow b = 1 \TOT a \leq b$, which holds in every Boolean algebra. The latter is satisfied by the Sasaki hook $ a \rightarrow_1 b := a' \vee (a \wedge b)$, which is equal to $a' \vee b$ when the orthomodular
lattice is distributive, i.e. when it is a Boolean algebra. In addition to $ a \rightarrow_1 b$, there are exactly four other ``quantum implications'' that satisfy the above
condition and all reduce to $a' \vee b$ in a Boolean algebra 
%that were introduced by Kalmbach 
(see~\cite{PM09}, p.~32).} are defined on {$S(\C H)$}$:$
	$$1 := H,$$
	$$0 := \{0\},$$
	$$L \wedge M := L \cap M,$$
	$$L + M := \{x + y \mid x \in L \ \& \  y \in M\},$$
	$$L \oplus M := L + M \ \mbox{when} \ L \cap M =_{\mathsmaller{S(\C H)}} 0,$$
	$$L \vee M :=\overline{L + M} =_{\mathsmaller{S(\C H)}} \overline{\langle L \cup M \rangle},$$
	$$-L := L^{\perp},$$
	$$L - M := L \wedge (-M) := L \cap M^{\perp},$$
	$$L \To M := (-L) \vee M,$$
	$$L \TOT M := (L \To M) \wedge (M \To L),$$ 
	$$\neg L := L \To 0.$$
	If $(L_i)_{i \in I}$ is a family of subspaces of $\C H$ over the index-set $I$, let
	$$\bigwedge_{i \in I} L_i := \bigcap_{i \in I}L_i,$$
	$$\bigvee_{i \in I} L_i := \overline{\bigg\langle\bigcup_{i \in I}L_i\bigg\rangle}.$$
	If $L \leq \C H$ is understood as a proposition, its derivability $\vdash L$ is defined as a proof of $L =_{\mathsmaller{S (\C H)}} 1$.
\end{definition}

Clearly, $\neg L := (- L) \vee 0 =_{\mathsmaller{S(\C H)}} -L$. The equality
$0 \To L := 0^{\perp} \vee L =_{\mathsmaller{S(\C H)}} 1 \vee L =_{\mathsmaller{S(\C H)}} 1$
expresses that within $S(\C H)$, seen as a logical system, the Ex Falso Quodlibet $(\EFQ)$ is derived.
In general, the implication in $S(\C H)$ does not satisfy the adjunction of a Heyting algebra
$$(K \wedge L) \leq M \ \mbox{if and only if}  \ K \leq (L \To M).$$
For example, in $\Real^2$ let $M := \Real(1, 0)$ be the $x$-axis and $L := \Real(\cos(\pi/4), \sin(\pi/4))$ be the diagonal. As $L^{\perp} \vee M =_{\mathsmaller{S(\C H)}} \Real^2$, we have that $L \leq L^{\perp} \vee M$, but $L \cap L =_{\mathsmaller{S(\C H)}}$ is not a subspace of $M$. For the converse implication,
if $K :=  \Real(1, 0)$, $L :=  \Real(\cos(\pi/4), \sin(\pi/4))$, and $M := 0$, then $K \wedge L \leq 0$, but $K$ is not a subspace of $L^{\perp}$.
Some of the following standard properties of $S(\C H)$ require classical logic and constitute the notion of a \textit{classical quantum lattice}, or an \textit{orthomodular lattice}.

\begin{proposition}[$\CLASS$]\label{prp: ClQL}
	If $L, M \in S(\C H)$, then the following properties hold$:$\\[1mm]
	$(\ClQL_0)$ If $H$ contains a non-zero element, then $\neg(0 =_{\mathsmaller{S(\C H)}} 1)$.\\[1mm]
	$(\ClQL_1)$ $L \wedge M$ $(L \vee M)$ is the gratest lower bound $($least upper bound$)$ of $L$ and $M$ relative to $\leq$.\\[1mm]
	$(\ClQL_{1{'}})$ $\bigwedge_{i \in I} L_i$ $\big(\bigvee_{i \in I} L_i\big)$ is the greatest lower bound $($least upper bound$)$ of $(L_i)_{i \in I}$ relative to $\leq$.\\[1mm]
	$(\ClQL_2)$ $0 \leq L \leq 1$.\\[1mm]
	$(\ClQL_3)$ If $L \leq M$, then $M^{\perp} \leq L^{\perp}$.\\[1mm]
	$(\ClQL_4)$ $L =_{\mathsmaller{S(\C H)}} L^{\perp \perp}$.\\[1mm]
	$(\ClQL_5)$ $L \wedge L^{\perp} =_{\mathsmaller{S(\C H)}} 0$.\\[1mm]
	$(\ClQL_6)$ $L \vee L^{\perp} =_{\mathsmaller{S(\C H)}} 1$.\\[1mm]
	%$(\ClQL_7)$ If $L \in S_{\loc}(\C H)$ and $L \leq M^{\perp}$, then $L \vee M \in S_{\loc}(\C H)$.\\[1mm]
	$(\ClQL_7)$ If $L \leq M$, then $M =_{\mathsmaller{S(\C H)}} L \vee (M-L)$.
\end{proposition}

\begin{definition}\label{def: ClQL}
	A classical quantum lattice, or an orthomodular lattice, is a structure $(\C S, \leq, -, 0, 1)$, such that $(\C S, \leq)$ is a poset, $0, 1 \in \C S$, and $- \colon \C S \to \C S$, such that the following conditions hold$:$\\[1mm]
	$(\ClQL_0$ $\neg(0 =_{\mathsmaller{\C S}} 1)$.\\[1mm]
	$(\ClQL_1)$ $p \wedge q$ $(p \vee q)$ is the greatest lower bound $($least upper bound$)$ of $p$ and $q$ relative to $\leq$.\\[1mm]
	$(\ClQL_2)$ $0 \leq p \leq 1$.\\[1mm]
	$(\ClQL_3)$ If $p \leq q$, then $(-q) \leq (-p)$.\\[1mm]
	$(\ClQL_4)$ $p =_{\mathsmaller{\C S}} -(-p)$.\\[1mm]
	$(\ClQL_5)$ $p \wedge (-p) =_{\mathsmaller{\C S}} 0$.\\[1mm]
	$(\ClQL_6)$ $p \vee (-p) =_{\mathsmaller{\C S}} 1$.\\[1mm]
	%$(\ClQL_7)$ If $L \in S_{\loc}(\C H)$ and $L \leq M^{\perp}$, then $L \vee M \in S_{\loc}(\C H)$.\\[1mm]
	$(\ClQL_7)$ If $p \leq q$, then $q =_{\mathsmaller{\C S}} p \vee (q-p)$, where $q-p := q \wedge (-p)$.\\[1mm]
	A complete classical quantum lattice is defined in the obvious way.
\end{definition}

\begin{example}\label{ex: BA}
\normalfont
Every Boolean algebra, which always satisfies distributivity, is a classical quantum lattice, and every complete Boolean algebra is a a classical complete quantum lattice.
\end{example}

By $(\ClQL_1)$, or directly by the definition of $\wedge, \vee$ in $S (\C H)$, we have trivially that $L \wedge M =_{\mathsmaller{S(\C H)}} M \wedge L$ and $L \vee M =_{\mathsmaller{S(\C H)}} M \vee L$.
By $(\ClQL_6)$ we get that $S(\C H)$, as a logical system, satisfies $\PEM$, i.e., $\vdash L \vee \neg L$. Since $L \leq L \vee M, L \leq L \vee N$, and $M \wedge N \leq M \leq L \vee M$, and $M \wedge N \leq N \leq L \vee N$, the following inequality holds constructively
$$L \vee (M \wedge N) \leq (L \vee M) \wedge (L \vee N).$$
Since $L \wedge M \leq L$, $L \wedge M \leq M \leq M \vee N$, and $L \wedge N \leq L$ and $L \wedge N \leq N \leq M \vee N$, the following inequality holds constructively
$$L \wedge (M \vee N) \geq (L \wedge M) \vee (L \wedge N).$$
The converse inequalities fail also classically. Just take $L :=  \Real(1, 0)$, $M := \Real(0,1)$, and $N :=  \Real(\cos(\pi/4), \sin(\pi/4))$ for both of them. Consequently, $(S(\C H), \wedge, \vee)$ does not satisfy distributivity (D)\footnote{Distributivity holds only for a Hilbert space $H$ with $\dim(H) \leq 1$ (see~\cite{Ha82}, Problem 14).}. Property $(\ClQL_7)$, which is the \textit{orthomodular law}, or the \textit{orthomodular identity}, is a special case of  modularity (M), which itself is a special case of (D), according to which if $L, M, N \leq \C H$, then 
$$N \leq L \To [L \wedge (M \vee N) =_{\mathsmaller{S(\C H)}} (L \wedge M) \vee N].$$ 
Since $L \wedge M \leq L$, $L \wedge M \leq M \leq M \vee N$, and $N \leq M \vee N$, the implication
$$N \leq L \To [L \wedge (M \vee N) \geq (L \wedge M) \vee N]$$
holds classically. In general, the implication
$$N \leq L \To [L \wedge (M \vee N) \leq (L \wedge M) \vee N]$$ 
does not hold classically (see, e.g.,~\cite{BvN36}, p.~832).
In~\cite{Ha82}, Problem 14, it is shown classically that (M) holds if and only if $H$ is finite-dimensional.
In~\cite{Ha82}, Problem 15, it is also shown that if $M, N \leq \C H$, then $M + N$ is closed if and only if (M) holds for $M, N$ and every $L \leq \C H$ with  $N \leq L$. One direction of this result is constructive. %Closed sums of subspaces of $\C H$ appear later in Theorem~\ref{thm: crucial}(i).

\begin{proposition}\label{prp: closedMod}
	If $M, N \leq \C H$, such that $\overline{M+N} =_{\mathsmaller{S(\C H)}} M + N$, then for every $L \leq \C H$ with $N \leq L$, we have that $L \wedge (M \vee N) =_{\mathsmaller{S(\C H)}} (L \wedge M) \vee N$.
\end{proposition}

\begin{proof}
	The constructive proof of this result is found in~\cite{Ha82}, p.~177.
\end{proof}

\begin{definition}\label{def: perpcoperp}
	If $L^1, L^0$ are subspaces of $\C H$, we define the following relations$:$
	$$L^1 \perp L^0 :\TOT \forall_{x^1 \in L^1}\forall_{x^0 \in L^0}(x^1 \perp x^0),$$
	$$L^1 \coperp L^0 :\TOT \exists_{x^1 \in L^1}\exists_{x^0 \in L^0}(x^1 \coperp x^0).$$
	If $L^1 \perp L^0$, we say that $L^1, L^0$ are orthocomplemented, or orthodisjoint, while if $L^1 \coperp L^0$, we say that $L^1, L^0$ are orthoovelapping, or they orthooverlap.
\end{definition}

The relations $L^1 \perp L^0$ and $L^1 \coperp L^0$ are the inner prooduct-analogue to the relations
of strong disjointness $A^1 \Disj A^0$ and strong overlap $A^1 \between A^0$ between subsets $A^1, A^0$ of $X \in \SetExtIneq$, where
% with an extensional inequality $\neq_X$, where
$$A^1 \Disj A^0 :\TOT \forall_{x^1 \in A^1}\forall_{x^0 \in A^0}(x^1 \neq_X x^0),$$
$$A^1 \between A^0 :\TOT \exists_{x^1 \in A^1}\exists_{x^0 \in A^0}(x^1 =_X x^0).$$
The following properties of orthocomplemented subspaces connect $S(\C H)$ with \textit{apart algebras} (see~\cite{Pe25}) and are straightforward to show.
%Proposition on the apart algebra properties of $\perp$ and the overlap algebra properties of $\coperp$.

\begin{remark}\label{rem: perp1}If $L^1, L^0, M^0 \leq S(\C H)$, and $(L_i)_{i \in I}$ is an $I$-family of subspaces of $\C H$, the following hold$:$\\[1mm]
	\normalfont (i)
	\itshape If $L^1 \perp L^0$, then $L^1 \wedge L^0 = 0$ and $L^0 \leq {(L^1)}^{\perp}$.\\[1mm]
	\normalfont (ii)
	\itshape If $L^1 \perp L^0$, then $L^0 \perp L^1$.\\[1mm]
	\normalfont (iii)
	\itshape If $L^1 \perp L^0$ and $M^0 \leq L^0$, then $L^1 \perp M^0$.\\[1mm]
	\normalfont (iv)
	\itshape If $S \subseteq H$ and $L^1 \perp S$, then $L^1 \perp \overline{\langle S \rangle}$.\\[1mm]
	\normalfont (v)
	\itshape $L^1 \perp \bigvee_{i \in I} L_i \TOT \forall_{i \in I}\big(L^1 \perp L_i)$.
\end{remark}

%\begin{proof}
%	(i) .\\
%	(ii) 
%\end{proof}

The following properties of orthooverlapping subspaces connect $S(\C H)$ with \textit{overlap algebras} (see~\cite{CS10}) and are straightforward to show.

\begin{remark}\label{rem: coperp1}If $L^1, L^0, M^0 \leq S(\C H)$, and $(L_i)_{i \in I}$ is an $I$-family of subspaces of $\C H$, the following hold$:$\\[1mm]
	\normalfont (i)
	\itshape If $L^1 \coperp L^0$, then $L^0 \coperp L^1$.\\[1mm]
	\normalfont (ii)
	\itshape If If $L^1 \coperp L^0$, and $L^0 \leq M^0$, then $L^1 \coperp M^0$.\\[1mm]
	\normalfont (iii)
	\itshape $L^1 \coperp \bigvee_{i \in I} L_i \TOT \exists_{i \in I}\big(L^1 \coperp L_i)$.
\end{remark}

If the inequality $\neq_X$ on a set is cotransitive, and $A \subseteq X$, then emptiness of $A$ is characterised through strong disjointness, i.e., 
%analogy to the following characterisations of emptiness and inhabitedness of subsets
$$A =_{\mathsmaller{P(X)}} \emptys_X \TOT A \Disj A.$$
Dually, the inhabitedness of $A$ is characterised through strong overlap, i.e., 
$$A \ \mbox{is inhabited} \ \TOT A \between A.$$
In analogy, we have the following characterisations of equality to $0$ and strictness of subspaces within the constructive theory of Hilbert spaces:
$$L =_{\mathsmaller{S(\C H)}} 0 \TOT L \perp L,$$
$$L \ \mbox{is strict} \ \TOT L \coperp L.$$

\begin{definition}\label{def: arrow}
	If $H, H{'}$ are Hilbert spaces, a map $T \colon H \to H{'}$ is an isometry, if $||T(x)|| =_{\Real} ||x||$, for every $x \in H$, and $T$ is a contraction, if $||T(x)|| \leq ||x||$, for every $x \in H$. Let $(\Hilb, \LinIsom)$ be the category of Hilbert spaces with linear isometries, a subcategory of the category $(\Hilb, \LinContr)$ of Hilbert spaces and linear contractions\footnote{The latter category is considered in~\cite{Le23}}. 	If $L \leq \C H$, a bounded linear map $T \colon L \to H$ is called \textit{positive}, if $\langle Tx, x\rangle \geq 0$, for every $x \in L$.
\end{definition}

%A bounded operator on $\C H$ is a strongly extensional function; if $x \in H$, then the inequalities $0 < ||T(x)|| \leq C ||x||$, imply that $T(x) \neq_H 0 \To x \neq_H 0$. 

\begin{proposition}\label{prp: inv}
	If $T \colon H \to H{'} \in \LinIsom$ and $M, N \leq \C H{'}$, then the following hold$:$\\[1mm]
	\normalfont (i)
	\itshape $T^{-1}(M) \leq \C H$.\\[1mm] 
	\normalfont (ii)
	\itshape $T^{-1}(0) =_{\mathsmaller{S(\C H)}} 0$ and $T^{-1}(1) =_{\mathsmaller{S(\C H)}} 1$.\\[1mm] 
	\normalfont (iii)
	\itshape $T^{-1}(M \wedge N) =_{\mathsmaller{S(\C H)}} T^{-1}(M) \wedge T^{-1}(N)$.\\[1mm] 
	\normalfont (iv)
	\itshape  $T^{-1}(M \vee N) \geq T^{-1}(M) \vee T^{-1}(N)$.
	%\[1mm]
%	\normalfont (v)
%	\itshape If $M \perp N$, then $T^{-1}(M) \perp T^{-1}(N)$, if $\D K := \Real$....
\end{proposition}

%\begin{proof}
%	We only show case (v). Let $x, x{'} \in X$, such that $T(x) \perp T(x{'})$. Since
%	\begin{align*}
	%	||x||^2 + ||x{'}||^2  & =_{\D K} ||T(x)||^2 + ||T(x{'})||^2\\
	%	&  =_{\D K} ||T(x) + T(x{'})||^2\\
	%	&  =_{\D K} ||T(x+x{'})||^2\\
	%	&  =_{\D K} ||x+x{'}||^2.
%	\end{align*}
%	Hence, by the converse of Pythagoras theorem in the real case we get $x \perp x{'}$.
%\end{proof}

The existence of a projection $P_L$ from a Hilbert space $\C H$ onto a \textit{located} subspace $L$ of $\C H$ is shown constructively  in~\cite{BB85}, pp.~366-367 (Theorem (8.7)). 

\begin{theorem}[Bishop-Bridges]\label{thm: BBcrucial}
If $L \leq \C H$ is located, then to each $x \in H$ corresponds a
unique closest point $P(x) \in L$. The vector $P(x)$ can also be characterized
as the unique vector $y \in L$ such that $\langle x - y, z\rangle =_H 0$ for all $z \in L$. The
function $P$ is an idempotent operator with range $L$. If $L$ is strict, then $P$ is normable, and $||P|| =_{\Real} 1$.
\end{theorem}

The proof of Theorem~\ref{thm: BBcrucial} is along the lines of the corresponding classical proof and requires the axiom of countable choice $(\CC)$, which is generally accepted in $\BISH$, although it is desirable to avoid it (see~\cite{Ri01}). $\CC$ is necessary to extract a sequence $(l_n)_{n \in \Nat} \subseteq L$, such that 
$$||l_n - x|| \stackrel{n} \longrightarrow \rho(x, L) := \inf\{||l - x|| \mid l \in L\},$$
for every $x \in X$. In section~\ref{sec: orthosub} we prove a two-dimensional version of this result (Theorem~\ref{thm: crucial}) that avoids $\CC$. If one wants to avoid $\CC$ in the proof of Bishop and Bridges, then the definition of locatedness of $L$ must be accommodated by a \textit{modulus of locatedness} $\loc_L \colon X \to \D F(\Nat, L)$, where $\loc_L(x)$ is a sequence in $L$ that satisfies the above convergence condition, for every $x \in X$. Of course, then one needs to reformulate the theory of located subsets in metric spaces taking these moduli of locatedness into account. Although we believe that this is possible, it has not been done yet. Only after revisiting locatedness in this way, the only remaining choice principle used in the aforementioned proof of Bishop and Bridges is Myhill's axiom of unique choice $\MANC$, introduced by Myhill in his seminal paper~\cite{My75}. This is exactly the form of (non) choice used in our proof of Theorem~\ref{thm: crucial} too. By inspection of the proof of Theorem~\ref{thm: BBcrucial} we get the following corollary, which, classically, holds for an arbitrary subspace and in case (ii) for an arbitrary $x \in L^{\neg}$ (see~\cite{Ha57}, pp.~23-24).

\begin{corollary}\label{cor: corBBcrucial}
Let $L \leq \C H$ be located and $x \in X$.\\[1mm]
\normalfont (i)
\itshape There is $l \in L$ with $\rho(x, L) =_{\Real} ||x - l||$.\\[1mm]
\normalfont (ii)
\itshape If $x \in L^{\neq}$, then there is $z \in H_0$ with $z \in L^{\perp}$.
\end{corollary}

\begin{proof}
To show (i), it suffices to to take $l := P(x)$ from Theorem~\ref{thm: BBcrucial}, while for (ii) let $z := x - l$. Since by hypothesis $\rho(x, L) =_{\Real} ||x - l|| > 0$, we get $z \in H_0$.
\end{proof}

Of course, if $x \in H_0$ with $x \in L^{\perp}$, then $x \in L^{\neq}$, since by Pythagoras theorem
$$||x-l||^2 =_{\Real} ||x||^2 + ||l||^2 \geq ||x||^2 > 0,$$
for every $l \in L$. The standard equivalence between the order of subspaces of $\C H$ and the order of projections on $\C H$ holds also constructively, only if we restrict to located subspaces.

\begin{proposition}\label{prp: equivorder}
If $L, M \leq \C H$ are located and $P_L, P_M$ are their corresponding projections, then the following are equivalent$:$\\[1mm]
\normalfont (i)
\itshape $L \leq M$.\\[1mm]
\normalfont (ii)
\itshape $P_M \circ P_L =_{\mathsmaller{\Pr(\C H)}} P_L$.\\[1mm]
\normalfont (iii)
\itshape $P_L \circ P_M =_{\mathsmaller{\Pr(\C H)}} P_L$.\\[1mm]
\normalfont (iv)
\itshape $\forall_{x \in H}\big(||P_L(x)|| \leq ||P_M(x)||\big)$.\\[1mm]
\normalfont (v)
\itshape $P_L \leq P_M :\TOT \forall_{x \in H}\big(\langle P_L(x), x \rangle \leq \langle P_M(x), x\rangle\big)$.
\end{proposition}

\begin{proof}
The constructive proof of these equivalences can be found in~\cite{KR83}, p.~110, where the hypothesis of locatedness is not mentioned, since it is classically void.
\end{proof}

\section{Intuitionistic quantum logic of Bridges and Svozil}
\label{sec: sub}

The relation between quantum logic and pure logic is already discussed\footnote{Exactly this extract is quoted also by Landsman in~\cite{La17}, p.~459.} in the seminal paper of Birkhoff and von Neumann~\cite{BvN36}, p.~837.

\begin{quote}
	The models for propositional calculi which
	have been considered in the preceding sections are also interesting from the
	standpoint of pure logic. Their nature is determined by quasi-physical and
	technical reasoning, different from the introspective and philosophical considerations which have had to guide logicians hitherto. Hence it is interesting to compare the modifications which they introduce into Boolean algebra, with those which logicians on ``intuitionist'' and related grounds have tried introducing.
	%The main difference seems to be that whereas logicians have usually assumed
	%that properties L71-L73 of negation were the ones least able to withstand a
	%critical analysis, the study of mechanics points to the distributive identities L6 as
	%the weakest link in the algebra of logic. Cf. the last two paragraphs of ?10.
	$\ldots$ Our conclusion agrees perhaps more with those critiques of logic, which find
	most objectionable the assumption 
	%that a' U b = 0 implies a C b (or dually,
	%the assumption that a n b' = ? implies b D a-the assumption 
	$\ldots$ that to deduce
	an absurdity from the conjunction of $a$ and not $b$, justifies one in inferring that
	$a$ implies $b$.
	%).36
	%18. Suggested questions. The same heuristic reasoning suggests the
\end{quote}

The constructive approach to the foundations of quantum mechanics has its origin to the work of 
Bridges~\cite{Br81}, where a constructive axiomatic theory of events\footnote{Events, or questions, are the simplest observables with a spectrum (the set of possible measured-values of the observable) included in $2 := \{0, 1\}$.} and states of a physical system, together with its relation to a constructive version of the Hilbert space-model of quantum mechanics, were presented for the first time.
 In~\cite{Br81} a constructive proof of Gleason's theorem was posed as an open problem, which was found later by Richman and Bridges in~\cite{RB99, Ri00}. Gleason's theorem, a cornerstone of the mathematical foundations of quantum mechanics, was assumed by von Neumann, conjectured by Mackey, and proved classically by Gleason in~\cite{Gl57}. In~\cite{BS00}, which is an extension and an improvement of~\cite{Br81}, a new axiomatization of quantum lattices is given, tailored to the role of locatedness that constructively has to be added to the notion of a closed linear subspace of a Hilbert space.
 A different constructive approach to quantum logic is developed by Landsman in~\cite{La17}. The study of bounded and unbounded linear operators on Hilbert spaces in~\cite{Ye00,Ye11} within Ye's strict finitism, based on ideas of Bishop and Bridges, and the results of Spitters on the constructive theory of von Neumann algebras in~\cite{Sp02,Sp05} are notable related constructive developments indicating that large parts of quantum mechanics can be formulated constructively.

In this section, and based on the work of Bridges and Svozil~\cite{BS00}, we present the basic lattice-properties of the subspaces of a Hilbert space $\C H$ from the constructive point of view. As expected, not all classical properties related to the negation $L^{\perp}$ of $L$ in $S(\C H)$ can be accepted constructively.

\begin{proposition}[Bridges-Svozil]\label{prp: Brouwerian}
If $L, M \in S(\C H)$, then the following cannot be accepted constructively$:$\\[1mm]
$(\BC_1)$ If $L, M \leq \C H$,
%are located,
then $L \wedge M$ is located.\\[1mm]
$(\BC_2)$ If $L, M \leq \C H$,
% are located, 
then $L \vee M$ is located.\\[1mm]
$(\BC_3)$ $(L \wedge M)^{\perp} =_{\mathsmaller{S(\C H)}} L^{\perp} \vee M^{\perp}$.\\[1mm]
$(\BC_4)$ $(L \vee M)^{\perp \perp} =_{\mathsmaller{S(\C H)}} L \vee M$.\\[1mm]
$(\BC_5)$ $L^{\perp \perp} =_{\mathsmaller{S(\C H)}} L$.\\[1mm]
$(\BC_6)$ $(L \vee M)^{\perp} =_{\mathsmaller{S(\C H)}} 0 \To L \vee M =_{\mathsmaller{S(\C H)}} H$.\\[1mm]
$(\BC_7)$ If $L$ is located and $L^{\perp}  =_{\mathsmaller{S(\C H)}} M^{\perp}$, then $L =_{\mathsmaller{S(\C H)}} M$.
\end{proposition}

\begin{proof}
For the proof of $(\BC_1)$-$(\BC_6)$, see~\cite{BS00}, pp.~505-506. For the proof of $(\BC_7)$, take $M := L^{\perp \perp}$ and use the equality $L^{\perp} =_{S(\C H)} L^{\perp \perp \perp}$ in Proposition~\ref{prp: corCoQL}(iv). 
\end{proof}

Proposition~\ref{prp: CoQL} is found without proof in~\cite{BS00}. 
According to it, the subspaces of a Hilbert space, together with its located subspaces, form an \textit{intuitionistic quantum lattice}. In Proposition~\ref{prp: ComplQL} we provide an alternative approach to the proof of the next proposition.
Here we only added condition $(\CoQL_0)$ as part of the definition of an intuitionisitic quantum lattice. By Theorem~\ref{thm: BBcrucial} the totality of located subspaces of $\C H$ is equipped with the following inequality\footnote{In~\cite{BS00}, p.~511, this inequality is defined by quantification over the whole space $H$. Here we quantify over $H_0$, in order to be compatible with the partial case (see Definition~\ref{def: partialproj}).}.

\begin{definition}\label{def: ineqSloc}If $L, M$ are located subspaces of $\C H$, let
$$L \neq_{\mathsmaller{S_{\loc}(\C H)}} M :\TOT P_L \neq_{\mathsmaller{\C B(\C H)}} P_M :\TOT \exists_{x \in H_0}\big(P_L(x) \neq_H P_M(x)\big).$$
In this case we write $x \colon L \neq_{\mathsmaller{S_{\loc}(\C H)}} M$.
\end{definition}

\begin{example}\label{ex: negmap}
\normalfont
The function $^{\perp} \colon S_{\loc}(\C H) \to S_{\loc}(\C H)$, defined by the rule $L \mapsto L^{\perp}$
is well-defined (see our remark right after Definition~\ref{def: orthocomplement}). Moreover, it is strongly extensional, since if $L^{\perp} \neq_{\mathsmaller{\C B(\C H)}} M^{\perp} :\TOT \exists_{x \in H_0}\big(P_{L^{\perp}}(x) \neq_H P_{M^{\perp}}(x)\big)$, and if by Theorem~\ref{thm: BBcrucial} $x =_H l + l^{\perp} = m + m^{\perp}$, with $l \in L, m \in M, l^{\perp} \in L^{\perp}$, and $m^{\perp} \in M^{\perp}$, then by $\LinIneq_4$ and the strong extensionality of the map $x \mapsto (-1) \cdot x$ we get
$$l^{\perp} \neq_H m^{\perp} \To (x - l) \neq_H (x - m) \To (-l) \neq_H (-m) \To l \neq_H m \To P_L(x) \neq_H P_M(x).$$ 
Hence, $x \colon L \neq_{\mathsmaller{S_{\loc}(\C H)}} M$.
\end{example}

\begin{proposition}[Bridges-Svozil]\label{prp: CoQL}
Let the structure $(S(\C H), S_{\loc}(\C H), \leq, ^{\perp}, 0, 1)$. If $\C H$ is strict, if $L, M \in S(\C H)$, and if $(L_i)_{i \in I}$ is a family of subspaces over the index-set $I$, then the following hold$:$\\[1mm]
$(\CoQL_0)$ $0, 1 \in S_{\loc}(\C H)$ and $0 \neq_{\mathsmaller{S_{\loc}(\C H)}} 1$.\\[1mm]
$(\CoQL_1)$ $L \wedge M$ $(L \vee M)$ is the greatest lower bound $($least upper bound$)$ of $L,M$ relative to $\leq$.\\[1mm]
$(\CoQL_{1{'}})$ $\bigwedge_{i \in I} L_i$ $\big(\bigvee_{i \in I} L_i\big)$ is the greatest lower bound $($least upper bound$)$ of $(L_i)_{i \in I}$ relative to $\leq$.\\[1mm]
$(\CoQL_2)$ $0 \leq L \leq 1$.\\[1mm]
$(\CoQL_2)$ $0 \leq L \leq 1$.\\[1mm]
$(\CoQL_3)$ If $L \leq M$, then $M^{\perp} \leq L^{\perp}$.\\[1mm]
$(\CoQL_4)$ $L \leq L^{\perp \perp}$.\\[1mm]
$(\CoQL_5)$ $L \wedge L^{\perp} =_{\mathsmaller{S(\C H)}} 0$.\\[1mm]
$(\CoQL_6)$ If $L \in S_{\loc}(\C H)$, then $L^{\perp} \in S_{\loc}(\C H)$.\\[1mm]
$(\CoQL_7)$ If $L \in S_{\loc}(\C H)$ and $L \leq M^{\perp}$, then $L \vee M \in S_{\loc}(\C H)$.\\[1mm]
$(\CoQL_8)$ If $L \in S_{\loc}(\C H)$ and $L \leq M$, then $M =_{\mathsmaller{S(\C H)}} L \vee (M-L)$.
\end{proposition}

Property $(\CoQL_3)$ corresponds to contraposition, property $(\CoQL_4)$ corresponds to double negation introduction, which hold constructively also for the standard negation of formulas.
Property $(\CoQL_8)$ is the constructive version of the orthomodular law given by Bridges and Svozil.
The following properties of an intuitionistic quantum lattice are derived by its axioms in~\cite{BS00}.
%In Proposition~\ref{prp: corComplQL} we provide an alternative approach to these proofs?????Maybe skip it.....

\begin{proposition}[Bridges-Svozil]\label{prp: corCoQL}
If $\C H$ is strict, and $L, M \in S(\C H)$, then the following hold$:$\\[1mm]
\normalfont (i)
\itshape If $L \in S_{\loc}(\C H)$, $L \leq M$, and  $M-L =_{\mathsmaller{S(\C H)}} 0$, then $L =_{\mathsmaller{S(\C H)}} M$.\\[1mm]
\normalfont (ii)
\itshape $S(\C H)$ satisfies $(\CoQL_0)$-$(\CoQL_7)$, but we cannot accept in $\BISH$ that it satisfies $(\CoQL_8)$.\\[1mm]
\normalfont (iii)
\itshape $0 =_{\mathsmaller{S(\C H)}} 1^{\perp}$ and $1 =_{\mathsmaller{S(\C H)}} 0^{\perp}$.\\[1mm] 
\normalfont (iv)
\itshape $L^{\perp \perp \perp} =_{\mathsmaller{S(\C H)}} L^{\perp}$.\\[1mm]
\normalfont (v)
\itshape $(L \vee M)^{\perp} =_{\mathsmaller{S(\C H)}} L^{\perp} \wedge M^{\perp}$.\\[1mm]
\normalfont (vi)
\itshape $(L \wedge M)^{\perp} \geq L^{\perp} \vee M^{\perp}$.\\[1mm]
\normalfont (vii)
\itshape If $L \in S_{\loc}(\C H)$, then $L =_{\mathsmaller{S(\C H)}} L^{\perp \perp}$.\\[1mm]
\normalfont (viii)
\itshape If $L, M \in S_{\loc}(\C H)$ and $L \leq M$, then $M - L \in S_{\loc}(\C H)$.\\[1mm]
\normalfont (ix)
\itshape If $L, M \in S_{\loc}(\C H)$ and $L \leq M^{\perp}$, then $M =_{\mathsmaller{S(\C H)}} (L \vee M) - L$.\\[1mm]
%\normalfont (x)
%\itshape $1, 0 \in S_{\loc}(\C H)$.\\[1mm]
\normalfont (x)
\itshape If $L, M \in S_{\loc}(\C H)$ and $L^{\perp} =_{\mathsmaller{S(\C H)}} M^{\perp}$, then $L =_{\mathsmaller{S(\C H)}} M$.
\end{proposition}

By Proposition~\ref{prp: corCoQL}(x) the strongly extensional function $^{\perp} \colon S_{\loc}(\C H) \to S_{\loc}(\C H)$ is an injection.
A feature missing in the presentation of intuitionistic quantum logic in Proposition~\ref{prp: CoQL} is a lattice-theoretic characterisation of locatedness. In the abstract formulation of $\CoQL$ in~\cite{BS00} a quantum lattice is a structure $(\C S, \C T, \leq, -)$, such that $(\C S, \leq)$ is an inhabited poset, $- \colon \C S \to \C S$, and  $\C T$ is an inhabited subset of $\C S$ satisfying the corresponding abstract versions of $(\CoQL_1)$-$(\CoQL_8)$. No internal, lattice-theoretic characterisation of $\C T$ is given in~\cite{BS00}. Next, we provide such a characterisation.

\begin{theorem}\label{thm: loctot}If $L \in S(\C H)$, then $L$ is located if and only if $L \vee L^{\perp} =_{\mathsmaller{S(\C H)}} 1$
\end{theorem}

\begin{proof}
If $L$ is located, then the equality $L \vee L^{\perp} =_{\mathsmaller{S(\C H)}} 1$ follows from Theoremm~\ref{thm: BBcrucial}. If $L \vee L^{\perp} =_{\mathsmaller{S(\C H)}} 1$, then
by Exercise 37 in~\cite{BB85}, p.~394, which is proved in~\cite{BV06}, p.~180, we have that $L$ and $L^{\perp}$ are located. This proof is a special case of our proof of Theorem~\ref{thm: crucial}(v) for an orhocomplemented subspace with domain the whole space $H$.  
\end{proof}

\begin{definition}\label{def: totalsub}
If $L \leq \C H$, we call $L$ total if $L \vee L^{\perp} =_{\mathsmaller{S (\C H)}} 1$, and we denote their set by $S_{\tota}(\C H)$.
\end{definition}

Through this lattice-theoretic characterisation of locatedness quantum logic can be formulated constructively as follows.

\begin{definition}\label{def: CoQL}
An intuitionistic quantum lattice is a structure $(\C S, \leq, -, 0, 1)$, such that 
%$(\C S, =_{\mathsmaller{\C S}}, \neq_{\mathsmaller{\C S}}, \leq)$ 
$(\C S, \leq)$ is a poset, $0, 1 \in \C S$, and $- \colon \C S \to \C S$. If $\C S_{\tota}$ is the subset of the total elements of $\C S$, i.e., $p \vee (-p) =_{\C S} 1$, for every $p \in \C S_{\tota}$, equipped with an inequality $t \neq_{\mathsmaller{\C S_{\tota}}} t{'}$, the following conditions hold$:$\\[1mm]
$(\CoQL_0)$ $0, 1 \in \C S_{\tota}$ and $0 \neq_{\mathsmaller{\C S_{\tota}}} 1$.\\[1mm]
 $(\CoQL_1)$ $p \wedge q$ $(p \vee q)$ is the least upper bound $($greatest lower bound$)$ of $p, q$ relative to $\leq$.\\[1mm]
 $(\CoQL_2)$ $0 \leq p \leq 1$.\\[1mm]
 $(\CoQL_3)$ If $p \leq q$, then $-q \leq -p$.\\[1mm]
 $(\CoQL_4)$ $p \leq -(-p)$.\\[1mm]
 $(\CoQL_5)$ $p \wedge (-p) =_{\mathsmaller{\C S}} 0$.\\[1mm]
 $(\CoQL_6)$ If $t \in \C S_{\tota}$, then $-t \in \C S_{\tota}$.\\[1mm]
 $(\CoQL_7)$ If $t \in \C S_{\tota}$ and $t \leq (-p)$, then $t \vee p \in \C S_{\tota}$.\\[1mm]
 $(\CoQL_8)$ If $t \in \C S_{\tota}$ and $t \leq p$, then $p =_{\mathsmaller{\C S}} t \vee (p-t)$, where $p-t := p \wedge (-t)$.\\[1mm]
 A complete intuitionistic quantum lattice is defined in the obvious way.
\end{definition}

%If we require that $\C S$ is equipped with a given inequality $\neq_{\mathsmaller{\C S}}$, then one can also demand that $0 \neq_{\mathsmaller{\C S}} 1$. 
Since $0, 1 \in \C S_{\tota}$,  we get immediately that 
$\C S_{\tota}$ is inhabited.

\section{Orthocomplemented subspaces of a Hilbert space}
\label{sec: orthosub}

In this section we introduce the notion of an orthocomplemented subspace  of a Hilbert space $\C H$. Orthocomplemented subspaces form a two-dimensional counterpart to the one-dimensional subspaces of $\C H$, and the Hilbert space analogue to Bishop's complemented subsets.

\begin{definition}\label{def: orthocs}
If $L^1, L^0 \in S(\C H)$ and $L^1 \perp L^0$, we call the pair $\B L := (L^1, L^0)$ an orthocomplemented subspace of $\C H$, in symbols $\B L \leq \C H$. If $\B L \leq H$, its domain is  
$\dom(\B L) := L^1 \vee L^0 \leq H$. We call $\B L$\\[1mm]
\normalfont (a)
\itshape total, if $\dom(\B L) =_{\mathsmaller{S(\C H)}} H$, \\[1mm]
\normalfont (b)
\itshape strict, if $L^1$ is strict,\\[1mm]
\normalfont (c)
\itshape located, if $L^1$ and $L^0$ are located.\\[1mm]
 Let $\B {\C S}(\C H)$ be the totality of orthocomplemented subspaces of $\C H$. If $\B L, \B M \in \B {\C S}(\C H)$, we define
$$\B L \leq \B M :\TOT L^1 \leq M^1 \ \& \ M^0 \leq L^0,$$
$$\B L =_{\mathsmaller{\B {\C S}(\C H)}} \B M :\TOT \B L \leq \B M \ \& \ \B M \leq \B L.$$
If $L^1 \coperp L^0$, we call the pair $\C L := (L^1, L^0)$ a coorthocomplemented subspace of $\C H$, in symbols $\C L \leq \C H$.  Let ${\C C}(\C H)$ be the totality of coorthocomplemented subspaces of $\C H$. If $\C L, \C M \in {\C C}(\C H)$, we define\footnote{Coorthocomplemented subspaces correspond to overlapping subsets, and their order corresponds to the order
$(A^1, A^0) \leq (B^1, B^0) :\TOT A^1 \subseteq B^1 \ \& \ A^0 \subseteq B^0$ between overlapping subsets. If $(A^1, A^0)$ is an overlapping subset, then it represents the intersection $A^1 \cap A^0$, and hence the inequality between overlapping subsets amounts to $A^1 \cap A^0 \subseteq B^1 \cap B^0$.}
$$\C L \leq \C M :\TOT L^1 \leq M^1 \ \& \ L^0 \leq M^0,$$
$$\C L =_{\mathsmaller{{\C S}(\C H)}} \C M :\TOT \C L \leq \C M \ \& \ \C M \leq \C L.$$
\end{definition}

\begin{remark}\label{rem: orthocs1}Let $L^1, L^0 \in S(\C H)$. The following hold$:$\\[1mm]
	\normalfont (i)
	\itshape If $L^1 \perp L^0$, then $L^1 \cap L^0 = \{0\}$ and $L^0 \subseteq {(L^1)}^{\perp}$.\\[1mm]
	\normalfont (ii)
	\itshape If $L$ is located, then $\B L := (L, L^{\perp})$ is a total orthocomplemented subspace of $\C H$ and $H = L \oplus L^{\perp}$.\\[1mm]
	\normalfont (iii)
	\itshape If $\B L$ is total, then $\B L$ is located.\\[1mm]
	\normalfont (iv)
	\itshape $L^1$ and $L^0$ are located in the Hilbert subspace $\dom(\B L)$ of $\C H$.
\end{remark}

\begin{proof}
	Case (i) is trivial to show, and cases (ii)-(iv) follow 
	immediately 
	by Theorem~\ref{thm: loctot}.
	\end{proof}

Case (i) of the following theorem is known classically. It is shown differently (see~\cite{KR83}, p.~112, and our comment after the proof of Corollary~\ref{cor: cor7}), but it is shown constructively, exactly as we prove it here, in~\cite{Ha57}, p.~25. 
%Our proof of Theorem~\ref{thm: crucial} is simple and constructive. 
The proof of Theorem~\ref{thm: crucial}(ii) rests on Myhill's axiom unique choice $(\MANC)$. Theorem~\ref{thm: crucial} is our two-dimensional version of Theorem (8.7) in~\cite{BB85}, pp.~366-367, since we replace the locatedness of a closed subspace of $\C H$ by the totality of an orthocomplemented subspace of $\C H$. Totality of complemented objects is a two-dimensional property that appears in all two-dimensional representations of some sort of complementation (see~\cite{PW22, MWP24, MWP25, Pe24a}), and expresses a limiting classical behaviour. While the approach of Bishop and Bridges is a direct constructive translation of the classical proof of the existence of the projection acting on a Hilbert space that corresponds to an arbitrary closed subspace, our proof takes a new path, that avoids countable choice.
%and it is (hopefully) interesting also from a classical point of view.
%his seminal formal system $\CST$ for $\BISH$ in~\cite{My75}.

\begin{theorem}[$\MANC$]\label{thm: crucial}
If $\B L := (L^1, L^0) \in \B S(\C H)$, the following hold$:$\\[1mm]
\normalfont (i)
\itshape $L^1 \vee L^0 =_{\mathsmaller{S(\C H)}} L^1 \oplus L^0$.\\[1mm]
\normalfont (ii)
\itshape There is a function $P_{\B L}^1 \colon \dom(\B L) \to L^1$, which is a linear, self-adjoint,  and bounded map with $||P_{\B L}^1|| =_{\Real} 1$, if $L^1$ is strict, which is also idempotent, i.e., $P_{\B L}^1 \big(P_{\B L}^1 (x)\big) =_H x$, for every $x \in \dom(\B L)$.\\[1mm]
\normalfont (iii)
\itshape The function $P_{\B L}^0 \colon \dom(\B L) \to L^0$, where $P_{\B L}^0 := \id_{\dom(\B L)} - P_{\B L}^1$ satisfies the same properties as $P_{\B L}^1$.\\[1mm]
\normalfont (iv)
\itshape If $x \in \dom(\B L)$, then $P_{\B L}^1 (x)$ is the unique element of $L^1$ satisfying the following property$:$
$$(*) \ \ \ \ \ \ \ \ \ \ \ \ \ \ \ \ \ \ \ \ \ \ \ \ \ \ \ \ \ \ \ \ \ \ \ \ \forall_{m^1 \in L^1}\big(\langle x - P_{\B L}^1 (x), m^1\rangle =_{\D K} 0\big). \ \ \ \ \ \ \ \ \ \ \ \ \ \ \ \ \ \ \ \ \ \ \ \ \ \ \ \ \ \ \ \ \ \ \ \ \ \ \ \ \ \ \ \ \ \ \ $$
\normalfont (v)
\itshape $L^1$ and $L^0$ are located in $\dom(\B L)$, where $\rho(x, L^1) =_{\Real} ||P_{\B L}^0(x)||$, and $\rho(x, L^0) =_{\Real} ||P_{\B L}^1(x)||$, for every $x \in \dom(\B L)$.\\[1mm]
%is also a linear, self-adjoint, bounded map that satisfies the condition $P_{\B L}^0 \big(P_{\B L}^0 (x)\big) =_H x$, for every $x \in H$.
\normalfont (vi)
\itshape If $x \in \dom(\B L) \cap (L^1)^{\neq}$, then there is $z \in L^0 \cap H_0$, hence, $z \perp L^1$ and $z \neq_H 0$.
%\\[1mm]
\end{theorem}

\begin{proof}
(i)  Clearly, $L^1 \oplus L^0 \leq L^1 \vee L^0$. To show the converse inequality, let $x \in L^1 \vee L^0$. By definition there is a sequence $(x_n)_{n \in \Nat}$ in $L^1 + L^0$, such that $(x_n)_n \stackrel{n} \longrightarrow x$. For every $n \in \Nat$ there are (unique) $l^1_n \in L^1$ and $l_n^0 \in L^0$, such that $x_n =_{H} l_n^1 + l_n^0$. Since $(x_n)_{n \in \Nat}$ is a Cauchy sequence, we have that
$$||(l_n^1 + l_n^0) - (l^1_m + l^0_m)|| \stackrel{n, m} \longrightarrow 0 \To ||(l_n^1 + l_n^0) - (l^1_m + l^0_m)||^2 \stackrel{n, m} \longrightarrow 0.$$
Since $l_n^1 - l^1_m \in L^1$, $l_n^0 - l^0_m \in L^0$, and $L^1 \perp L^0$, by Pythagoras theorem we get
$$||(l_n^1 - l^1_m) + (l_n^0 - l^0_m)||^2 =_{\Real} ||l_n^1 - l^1_m||^2 + ||l_n^0 - l^0_m||^2.$$
Consequently, we get
$$||l_n^1 - l_m^1|| \stackrel{n, m} \longrightarrow 0 \ \ \& \ \ ||l_n^0 - l_m^|| \stackrel{n, m} \longrightarrow 0,$$
i.e., the extracted by $\AC!$ sequences $(l_n^1)_{n \in \Nat}$ in $L^1$ and $(l_n^0)_{n \in \Nat}$ in $L^0$ are Cauchy sequences. As $L^1$ and $L^0$ are closed in $H$, they are complete (see~\cite{BB85}, p.~91). Hence, there are $l^1 \in L^1$, such that $(l_n^1)_n \stackrel{n} \longrightarrow l^1$ and $l^0 \in L^0$, such that $(l_n^0)_n \stackrel{n} \longrightarrow l^0$. By the continuity of $+$ we get 
$(l_n^1 + l_n^0)_n \stackrel{n} \longrightarrow l^1 + l^0$, and by the uniqueness of the limit we get $x =_H l^1 + l^0$.
If $m^1 \in L^1$ and $m^0 \in L^0$, such that $x =_H m^1 + m^0$, then 
$L^1 \ni l^1 - m^1 =_H l^0 - m^0 \in L^0$, hence $l^1 - m^1 =_H 0$ and $l^0 - m^0 =_H 0$, i.e., $l^1 =_H m^1$ and $l^0 =_H m^0$.\\
(ii) By $\MANC$ there is a function $P_{\B L}^1 \colon \dom(\B L) \to L^1$, where $x \mapsto l^1$, and $x =_H l^1 + l^0$, as explained above. If $y \in \dom(\B L)$, such that $y =_H m^1 + m^0$, then $x + y \mapsto l^1 + m^1$, and if $k \in \D K$, then $k x \mapsto k l^1$, i.e., $P_{\B L}^1$ is linear. Since $x \mapsto l^1 \mapsto l^1$, $P_{\B L}^1$ is idempotent. Next we show that if $x, y \in \dom(\B L)$, then
$$\langle P_{\B L}^1(x), y \rangle =_{\D K} \langle x, P_{\B L}^1(y) \rangle.$$
If $x =_H l^1 + l^0$ and $y =_H m^1 + m^0$, then the required equality $\langle l^1, m^1 + m^0 \rangle =_{\D K} \langle l^1 + l^0, m^1 \rangle$ follows from the equalities $\langle l^1, m^0 \rangle =_{\D K} 0 =_{\D K} \langle l^0, m^1 \rangle$. Hence,
$$0 \leq ||P_{\B L}^1(x)||^2 =_{\D K} \langle P_{\B L}^1(x), P_{\B L}^1(x)\rangle =_{\D K} \langle x, P_{\B L}^1 P_{\B L}^1(x)\rangle =_{\D K} \langle x, P_{\B L}^1(x)\rangle =_{\D K} \langle P_{\B L}^1(x), x\rangle.$$
By the inequality Cauchy-Schwartz we get $\big|\langle x, P_{\B L}^1(x)\rangle\big| \leq ||x|| \cdot || P_{\B L}^1(x)||$, hence $||P_{\B L}^1(x)|| \leq ||x||$, and $||P_{\B L}^1|| \leq 1$. If $l^1 \in L^1$ and $l^1 \neq_H 0$, then $P_{\B L}^1(l^1) := l^1$ and $||P_{\B L}^1|| =_{\D K} 1$.\\
(iii) The proof is straightforward.\\
(iv) For simplicity, we give the proof if $\D K := \Real$. Since $x - l^1 =_H l^0 \perp L^1$, we get that $l^1$ trivially satisfies $(*)$. Let $s^1 \in L^1$ satisfying also $(*)$.  If $m^1 := l^1$ in $(*)$ for $s^1$ we get
$$\langle x-s^1, l^1\rangle =_{\Real} 0 \To \langle x, l^1\rangle =_{\Real} \langle s^1, l^1\rangle \To \langle l^1, l^1\rangle =_{\Real} \langle s^1, l^1\rangle.$$
If $m^1 := l^1 -s^1$ in $(*)$ for $s^1$, and using the first equality right above we get
$$\langle x-s^1, l^1 -s^1\rangle =_{\Real} 0 \To \langle x - s^1, l^1\rangle + \langle x - s^1, - s^1\rangle =_{\Real} 0 \To \langle x, s^1\rangle =_{\Real} \langle s^1, s^1\rangle.$$
Hence, 
$$\langle l^1-s^1, l^1-s^1\rangle =_{\Real} \big(\langle l^1, l^1\rangle - \langle l^1, s^1\rangle\big) + \big(\langle s^1, s^1\rangle - \langle s^1, l^1\rangle\big) =_{\Real} 0 + \langle x, s^1\rangle - \langle s^1, l^1\rangle.$$
If $m^1 := s^1$ in $(*)$ for $l^1$, we get $ \langle x - l^1, s^1\rangle =_{\Real} 0$, hence $\langle l^1-s^1, l^1-s^1\rangle =_{\Real} 0$ and $l^1 =_H s^1$.\\
(v) First, we show that $L^1$ is located in $\dom(\B L)$. If $x \in \dom(\B L) := L^1 \oplus L^0$, and $x =_H l^1 + l^0$, with $l^1 \in L^1$ and $l^0 \in L^0$, then, for every $m^1 \in L^1$, by Pythagoras theorem we have that
\begin{align*}
||x-m^1||^2 & =_{\Real} ||l^1 + l^0 -m^1||^2\\
& =_{\Real} ||(l^1 - m^1) + l^0||^2\\
& =_{\Real} ||(l^1 - m^1)||^2 + ||l^0||^2\\
& \geq ||l^0||^2\\
& =_{\Real} ||x - l^1||^2. 
\end{align*}
Hence, $\rho(x, L^1) =_{\Real} ||l^0|| =_{\Real} ||x - l^1||$. Similarly, $\rho(x, L^0) =_{\Real} ||l^1||$, and $L^1, L^0$ are located in $\dom(\B L)$.\\
(vi) We work as in the proof of Corollary~\ref{cor: corBBcrucial}(ii)\footnote{Cases (v) and (vi) of Theorem~\ref{thm: crucial} are the orthocomplemented version of Corrollary~\ref{cor: corBBcrucial}(i) and (ii), respectively.}, without using the locatedness hypothesis for $L^1$, since $x \in \dom(\B L)$.
\end{proof}

To avoid $\MANC$ in the above proof, one needs to accommodate the related notions with additional witnessing data. For example, to every $L \leq \C H$ we can add a map $\Closure_L \colon \overline{L} \to \Cauchy(L)$, where $x \mapsto (x_n)_{n \in \Nat}$, a given Cauchy sequence converging to $x$. Moreover, we could add to $\C H$ a map $\Complete_{\C H} \colon \Cauchy(\C H) \to H$, where each Cauchy sequence in $H$ is mapped to its limit. Consequently, in the proof we could use the map $\Closure_{L^1 \vee L^0}$ and $\Complete_{\C H}$ appropriately, in order to avoid $\MANC$. 
%..., but then it is surely impredicative...as the totalities $S(\C H)$ and $\B S(\C H)$.
If $L^1$ and $L^0$ are not orthogonal, then case (i) of the previous proposition may not hold (see Problems 52-55 in~\cite{Ha82}). By the proof of Theorem~\ref{thm: crucial}(ii) we also get that $P_{\B L}^1$ and $P_{\B L}^0$ are positive operators.

\begin{corollary}\label{cor: corthm}
If $\B L := (L^1, L^0), \B M := (M^1, M^0) \in \B S (\C H)$, and $x \in H$, then the following hold$:$\\[1mm]
\normalfont(i)
\itshape If $\B L$ is total, then $L^0 =_{\mathsmaller{S(\C H)}} (L^1)^{\perp}$ and 
$L^1 =_{\mathsmaller{S(\C H)}} (L^1)^{\perp \perp}$.\\[1mm]
\normalfont(ii)
\itshape If $\B L$ is total, then $\B L =_{\mathsmaller{\B S (\C H)}} (L, L^{\perp})$, where $L := L^1$.\\[1mm]
\normalfont(iii)
\itshape If $M \leq \C H$, such that $L^1 \leq M$, then 
$$M \wedge (L^{0} \vee L^1) =_{\mathsmaller{S(\C H)}} (M \wedge L^{0}) \vee L^1,$$
and if $L^{0} \leq M$, then 
$$M \wedge (L^1 \vee L^{0}) =_{\mathsmaller{S(\C H)}} (M \wedge L^1) \vee L^{0}.$$
\normalfont(iv)
\itshape If $\B L$ is total and $M \leq \C H$, such that $L \leq M$, then 
$$M \wedge (L^{\perp} \vee L) =_{\mathsmaller{S(\C H)}} (M \wedge L^{\perp}) \vee L,$$
and if $L^{\perp} \leq M$, then 
$$M \wedge (L \vee L^{\perp}) =_{\mathsmaller{S(\C H)}} (M \wedge L) \vee L^{\perp}.$$
\normalfont(v)
\itshape $x \in \dom(\B L) \TOT \exists_{l^1 \in L^1}\exists_{l^0 \in L^0}\big(x =_H l^1 + l^0\big)$.\\[1mm]
\normalfont(vi)
\itshape $x \in \dom(\B L)^{\neq} \TOT \forall_{l^1 \in L^1}\forall_{l^0 \in L^0}\big(x \neq_H l^1 + l^0\big)$
\end{corollary}

\begin{proof}
(i) By Theorem~\ref{thm: crucial}(i) we have that $H = L^1 \oplus L^0$, hence $L^0 \leq (L^1)^{\perp}$. To show the converse inequality, let $x \in (L^1)^{\perp}$, and $l^1 \in L^1, l^0 \in L^0$, such that $x =_H l^1 + l^0$. Hence, $x - l^0 =_H l^1$. Since $x - l^0 \in (L^1)^{\perp}$, we get $l^1 =_H 0$ and $x =_H l^0$ i.e., $x \in L^0$. Since $(L^0, L^1)$ is also total, hence using the previous twice we get $L^1 =_{\mathsmaller{S(\C H)}} (L^0)^{\perp} =_{\mathsmaller{S(\C H)}} (L^1)^{\perp \perp}$.
Case (ii) follows immediately from case (i).\\
(iii) Both equalities follow from case (i), Theorem~\ref{thm: crucial}(i), which expresses that $L^1 + L^0$ is closed in $\C H$, and Proposition~\ref{prp: closedMod}. Case (iv) is a special case of (iii).
%, and (v)-(vii) follow immediately by Theorem~\ref{thm: crucial}(i).
Cases (v) and (vi) follow by Theorem~\ref{thm: crucial}(i).
\end{proof}

Since by Corollary~\ref{cor: corthm}(i) every total element of $\B S(\C H)$ is of the form $(L, L^{\perp})$, for some located (by Remark~\ref{rem: orthocs1}(iii)) $L \leq \C H$, then by Remark~\ref{rem: orthocs1}(ii) we conclude that the \textit{totality of $(L, L^{\perp})$ is equivalent to the locatedness of} $L$. Hence, 
 \textit{the located subspaces of $\C H$ correspond exactly to the total orthocomplemented subspaces of} $\C H$. 
 %In other words, locatedness is the one-dimensional version of the two-dimensional notion of totality, and vice versa. 
 Hence, it is not an accident that subsequent Proposition~\ref{prp: corComplQL} replaces the hypothesis of locatedness found in Proposition~\ref{prp: corCoQL} by that of totality. Consequently, our proof of Proposition~\ref{prp: corComplQL} is also a proof of Proposition~\ref{prp: corCoQL}. Our use of the two-dimensional notion of totality instead of locatedness provides a new approach to the definition of the quotient Hilbert space. In~\cite{BV06}, p.~55, the quotient normed space $X/Y$ of a normed space $(\C X, ||.||)$ over a subspace $Y$ of $X$ is defined only if $Y$ is located. The quotient set $X/Y$ is $X$ again equipped with the following equality, inequality, and norm:
$$x =_{\mathsmaller{X/Y}} x{'} :\TOT \rho(x-x{'}, Y) =_{\Real} 0,$$
$$x \neq_{\mathsmaller{X/Y}} x{'} :\TOT \rho(x-x{'}, Y) \neq_{\Real} 0 \TOT \rho(x-x{'}, Y) > 0,$$
$$||x||_{\mathsmaller{X/Y}} := \rho(x, Y).$$
In the case of a Hilbert space we can define the quotient Hilbert space $\dom(\B L)/\B L$ of $\dom(\B L)$ over $\B L \in \B S(\C H)$ and get the above construction as the limiting total case. 

\begin{proposition}\label{prp: quotient}
If $\B L \in \B S(\C H)$, the quotient Hilbert space $\dom(\B L)/\B L$ of $\dom(\B L)$ over $\B L$ is the
structure 
$$\big(\dom(\B L), =_{\mathsmaller{\dom(\B L)/\B L}}, \neq_{\mathsmaller{\dom(\B L)/\B L}}; +, 0, \cdot; \langle, \rangle_{\mathsmaller{\dom(\B L)/\B L}}\big),$$
where if $x, y \in \dom(\B L)$ with $x =_H P_{\B L}l^1(x) + P_{\B L}^0(x)$ and $y =_H P_{\B L}^1(y) + P_{\B L}^0(y)$, then
$$x =_{\mathsmaller{\dom(\B L)/\B L}} y :\TOT P_{\B L}^0(x) =_H P_{\B L}^0(y),$$
$$x \neq_{\mathsmaller{\dom(\B L)/\B L}} y :\TOT P_{\B L}^0(x) \neq_H P_{\B L}^0(y),$$
$$\langle x, y \rangle_{\mathsmaller{\dom(\B L)/\B L}} := \langle P_{\B L}^0(x), P_{\B L}^0(y) \rangle_{\mathsmaller{H}},$$
$$||x||_{\mathsmaller{\dom(\B L)/\B L}} := ||P_{\B L}^0(x)||_{\mathsmaller{H}}.$$
The projection function $\pi_{\B L} \colon \dom(\B L) \to \dom(\B L)$ is defined by the rule $x \mapsto x$, where the domain of $\pi_{\B L}$ is understood with the relative equality, inequality and inner product of $H$, while its codomain is understood with the above defined equality, inequality, and inner product. The following hold$:$\\[1mm]
\normalfont (i) 
\itshape $\dom(\B L)/\B L$ is a Hilbert space.\\[1mm]
\normalfont (ii) 
\itshape $\pi_{\B L}$ is a linear contraction, and if $L^0$ is strict, then $||\pi_{\B L}|| =_{\Real} 1$.\\[1mm]
\normalfont (iii) 
\itshape The map $T_{\B L} \colon \dom(\B L) \to L^0$, where $x \mapsto  P_{\B L}^0(x)$, for every $x \in \dom(\B L)$, is a linear isometry.
\end{proposition}

\begin{proof}
(i) It is straightforward to show that $x =_{\mathsmaller{\dom(\B L)/\B L}} y$ is an equality on $\dom(\B L)$ and $x \neq_{\mathsmaller{\dom(\B L)/\B L}} y$ is an inequality on $\dom(\B L)$, and that properties $(\LinIneq_1)$-$(\LinIneq_3)$ of $\C H$ imply properties $(\LinIneq_1)$-$(\LinIneq_3)$ for $\dom(\B L)$. The inner product-properties for $\langle x, y \rangle_{\mathsmaller{\dom(\B L)/\B L}}$ follow immediately from the inner product properties of $\langle , \rangle_{\mathsmaller{H}}$. Moreover, $\dom(\B L)$ is complete, since it is closed in $\C H$.\\
(ii) If $x, y \in \dom(\B L)$ with $x =_H y$, then $P_{\B L}l^1(x) =_H P_{\B L}l^1(y)$ and $P_{\B L}l^0(x) =_H P_{\B L}l^0(y)$, and hence $\pi_{\B L}$ is a function. The linearity of $\pi_{\B L}$ follows immediately by its definition. If $x \in \dom(\B L)$, we have that 
$$||\pi_{\B L}(x)||_{\mathsmaller{\dom(\B L)/\B L}} := ||x||_{\mathsmaller{\dom(\B L)/\B L}} :=  ||P_{\B L}^0(x)||_{\mathsmaller{H}} \leq ||x||_{\mathsmaller{H}},$$
i.e., $\pi_{\B L}$ is a contraction. Since $||P_{\B L}^0|| =_{\Real} 1$, we get that $\pi_{\B L}$ is also normed with norm $1$, since 
%by definition 
\begin{align*}
	1 & =_{\Real}  ||P_{\B L}^0||\\
	& := \sup \big\{||P_{\B L}^0(x)||_{\mathsmaller{H}} \mid x \in \dom(\B L) \ \& \ ||x||_{\mathsmaller{H}} \leq 1\big\}\\
	& := \sup \big\{||\pi_{\B L}(x)||_{\mathsmaller{\dom(\B L)/\B L}} \mid x \in \dom(\B L) \ \& \ ||x||_{\mathsmaller{H}} \leq 1\big\}\\
	& =: ||\pi_{\B L}||.
	\end{align*}
(iii) By the definition of $x =_{\mathsmaller{\dom(\B L)/\B L}} y$ we have that $T_{\B L}$ is a function, which is linear, since $P_{\B L}^0$ is linear. If $x \in \dom(\B L)$, then 
$||T_{\B L}(x)||_{\mathsmaller{H}} := ||P_{\B L}^0(x)||_{\mathsmaller{H}} =: ||x||_{\mathsmaller{\dom(\B L)/\B L}}$, hence $T_{\B L}$ is an isometry.
\end{proof}

Clearly, if $\B L =_{\mathsmaller{\B S(\C H)}} (L, L^{\perp})$ is total, then the quotient Hilbert space $\C H/\B L$ coincides with the quotient Hilbert space $H/L$ defined in~\cite{BV06}, as we described above.

\section{Bounded partial operators between normed spaces}
\label{sec: bpo}

In this section we introduce bounded partial operators between normed spaces, which motivate the introduction of the concept of a partial linear space. The bounded linear maps that are extended by the Hahn-Banach theorem are such bounded partial operators.

\begin{definition}\label{def: partialoperator} If $(X, =_X)$ and $(Y, =_Y)$ are sets, a partial function, or a partial map, $T$ from $X$ to $Y$, in symbols $T \colon X \pto Y$, is a function $T \colon \dom(T) \to Y$, where $\dom(T) \subseteq X$. We denote by $\C F(X, Y)$ the totality\footnote{As we quantify over $\C P(\C H)$ to define the membership condition for $\C F(X, Y)$, it is a proper class (see~\cite{Pe20}, section 2.7).} of partial functions from $X$ to $Y$. If $S, T \in \C F(X, Y)$, we define
$$S =_{\mathsmaller{\C F(X, Y)}} T :\TOT \dom(S) =_{\C P(X)} \dom(T) \ \& \ \forall_{x \in \dom(S)}\big(S(x) =_Y T(x)\big).$$
If $(Z, =_Z)$ is a set and $U \colon Y \pto Z$, the composite function $U \circ T \colon X \pto Z$ is the partial function 
\begin{center}
	\begin{tikzpicture}
		
		\node (E) at (0,0) {$X$};
		\node[right=of E] (F) {$Y$};
		\node[right=of F] (A) {$Z$};

		\draw[MyBlue,-left to] (E) to node [midway,above] {$T$} (F) ;
		\draw[MyBlue,-left to] (F) to node [midway,above] {$U$} (A) ;
		\draw[MyBlue,-left to, bend right=40] (E) to node [midway,below] {$U \circ T$} (A) ;
		
	\end{tikzpicture}
\end{center}
$U \circ T \colon \dom(U \circ T) \to Z$, where 
$$\dom(U \circ T) := \big\{x \in \dom(T) \mid T(x) \in \dom(U)\big\},$$
and $(U \circ T)(x) := U(T(x))$, for every $x \in \dom(U)$. If $(\C X, ||.||)$ and $(\C Y, ||.||)$ are normed spaces, a linear partial map $T \colon X \pto Y$ is a partial function $T \colon \dom(T) \to Y$, where $\dom(T) \leq X$, which is also linear. $T$ is a bounded partial operator, if it is bounded as a linear map from $\dom(T)$ to $Y$, and we denote by $\M B(\C X, \C Y)$ their totality. We write $L \leq \C X$ to denote that $L$ is a closed subspace of $\C X$, and let $S(\C X)$ be their totality. In case $\neq_X$ and $\neq_Y$ are extensional inequalities on $X$ and $Y$, respectively, 
%and $X_0 := \{x \in X \mid x \neq_X 0\}$, 
we define the following equality and inequality on $\M B(\C X, \C Y)$
\begin{align*}
	T =_{\mathsmaller{\M B(\C X, \C Y)}} U & :\TOT \forall_{x \in X}\bigg[\big(x \in \dom(T) \To x \in \dom(U)\big) \wedge \\
	& \ \ \ \ \ \ \ \ \ \ \ \ \ \big(x \in \dom(U) \To x \in \dom(U)\big) \wedge \\
	& \ \ \ \ \ \ \ \ \ \ \ \ \ \big(x \in \dom(T) \wedge \dom(U) \To T(x) =_{Y} U(x)\big)\bigg],
\end{align*}
\vspace{-6mm}
\begin{align*}
	T \neq_{\mathsmaller{\M B(\C H, \C Y)}} U & :\TOT \exists_{x \in X_0}\bigg[\big(x \in \dom(T) \ \& \ x \in \dom(U)^{\neq}\big) \vee \\
	& \ \ \ \ \ \ \ \ \ \ \ \ \ \big(x \in \dom(U) \ \& \ x  \in \dom(T)^{\neq}\big) \vee \\
	& \ \ \ \ \ \ \ \ \ \ \ \ \ \big(x \in \dom(T) \wedge \dom(U) \ \& \ T(x) \neq_{Y} U(x)\big)\bigg].
\end{align*} 
If $x \in X_0$ witnesses the inequality $T \neq_{\mathsmaller{\M B(\C H, \C Y)}} U$, we write $x \colon T \neq_{\mathsmaller{\M B(\C H, \C Y)}} U$.
\end{definition}

The use of quantification over $X_0$ in the definition of $T \neq_{\mathsmaller{\M B(\C H, \C Y)}} U$, instead of quantification over $X$, is chosen so that the definition of the inequality between bounded partial operators on a Hilbert space (see Definition~\ref{def: partialproj}) is of the same form, although there the orthocomplement of a subspace is used instead of its strong complement (see footnote 13). 

\begin{proposition}\label{prp: ext}
The inequality relation $T \neq_{\mathsmaller{\M B(\C H, \C Y)}} U$ is extensional.
\end{proposition}

\begin{proof}
We work as in the proof of Proposition~\ref{prp: ext1}. 
\end{proof}

\begin{remark}\label{rem: domcomp}
If $(\C X, ||.||)$, $(\C Y, ||.||)$, and $(\C Z, ||.||)$ are normed spaces, $T \colon X \pto Y$ and $U \colon Y \pto Z$ are partial bounded linear maps, then $\dom(U \circ T) \leq \C X$ and $U \circ T \colon X \pto Z$ is a partial bounded linear map.
\end{remark}

\begin{proof}
By definition $\dom(T) \leq \C X$ and $\dom(U) \leq \C Y$. Clearly, $0 \in \dom(U \circ T)$. If $x, y \in \dom(U \circ T)$, then $T(x), T(y) \in \dom(U)$, hence $T(x) + T(y) =_Y T(x + y) \in \dom(U)$, therefore $x + y \in \dom(U \circ T)$. Similarly, we show that if $x \in \dom(U \circ T)$ and $k \in \D K$, then $k \cdot x \in \dom(U \circ T)$. If $(x_n)_{n \in \Nat} \subseteq \dom(U \circ T)$ and $x \in X$ with $(x_n) \stackrel{n} \longrightarrow x$, then $x \in \dom(T)$ and $(T(x_n))_{n \in \Nat} \subseteq \dom(U)$. By the continuity of $T$ we also get $T(x) \in \dom(U)$, hence $x \in \dom(U \circ T)$. Trivially, $U \circ T$ is a bounded linear map.
\end{proof}

\begin{definition}\label{def: partialoperations}
Let $(\C X, ||.||)$, $(\C Y, ||.||)$ be normed spaces, and let $T \colon \dom(T) \to Y$, $U \colon \dom(U) \to Y$ be in $\M B(\C X, \C Y)$. Their addition $T + U$ has domain $\dom(T) \wedge \dom(U)$, and if $k \in \D K$, the product $k \cdot T$ has domain $\dom(T)$. Moreover, let $0_T \colon \dom(T) \to Y$, defined by $0_T(x) := 0$, for every $x \in \dom(T)$.
\end{definition}

The structure of $\M B(\C X, \C Y)$, equipped with the above operations, motivates the introduction of the notion of a partial linear space. In~\cite{MWP24} the notion of a swap ring was introduced as an abstract version of the algebraic structure of partial, Boolean-valued functions. Richman in~\cite{Ri12}, pp.~2682-2683, has also formulated an abstract
structure motivated by the real-valued partial functions on $\Real$, and especially the (explicit) algebraic functions, without developing their theory. 
Every swap ring is a Richman ring, in the sense of~\cite{Ri12} (see also~\cite{MWP25b, MWP25c}).

\begin{definition}\label{def: pls}
A partial linear space over the field $\D K \in \{\Real, \D C\}$ is a structure $\C X := (X, =_X, \neq_X; +, 0, (0_{x})_{x \in X}, \cdot)$, where $(X, =_X, \neq_X)$ is a totality with an extensional inequality and $(+, 0, (0_{x})_{x \in X}, \cdot)$ is the partial linear structure of $\C X$ with $0_{x} \in X$, for every $x \in X$,
%a strongly extensional function, $x \mapsto 0_x$, 
that satisifies the axioms of a linear space in Definition~\ref{def: linearspace} except from the existence of the additive inverse, for every $x \in X$, and the property $0 \cdot x =_X 0$, for every $x \in X$. Instead, the following generalisations of these properties hold{}$:$\\[1mm]
$(\PL_1)$ $0 \cdot x =_X 0_x$, for every $x \in X$.\\[1mm]
$(\PL_2)$ For every $x \in X$ there is an element $(-x) \in X$, the partial additive inverse of $x$, such that $x + (-x) =_X 0_x$ and $0_{-x} =_X 0_x$.\\[1mm]
$(\PL_3)$ $0_0 =_X 0$.\\[1mm]
$(\PL_4)$ $0_{0_x} =_X 0_x$.\\[1mm]
Moreover, $\C X$ satisfies the axioms $(\LinIneq_1), (\LinIneq_2)$ from Definition~\ref{def: linearspace} and the following generalisation of $(\LinIneq_3)${}$:$\\[1mm]
$(\PL_5)$ $k \cdot x \neq_X 0 \To \big[0_x \neq_X 0 \vee (k \neq_{\D K} 0 \ \& \ x \neq_X 0) \big]$.\\[1mm]
We call $x \in X$ total, if $0_x =_X 0$, and we denote by $X_{\tota}$ their totality. We call $x \in X$ strictly partial, if $0_x \neq_X 0$, and we denote by $X_{\spart}$ their totality.
%We call $\C X$ strict, if there is $x \in X$ with $x \neq_X 0$.  Let $X_0 := \{x \in X \mid x \neq_X 0\}$.
\end{definition}

Clearly, every linear space is a partial linear space which is equal to its total elements. 
%Notice that i
If $\C X$ is a linear space, then $0_x \neq_X 0$ is impossible, and hence $(\PL_5)$ is reduced to $(\LinIneq_3)$. If $T \in \M B(\C X, \C Y)$ has domain a proper subspace of $\C X$, then $T$ is strictly partial. The converse also holds, since 
$$0_T \neq_{\mathsmaller{ \M B(\C X, \C Y)}} 0 :\TOT \exists_{x \in X_0}\big[x \in \dom(0_T)^{\neq} \vee \big(x \in \dom(0_T) \ \& \ 0_T(x) \neq_Y 0(x)\big)\big].$$

\begin{corollary}\label{cor: pls1} If $\C X$ is a partial linear space and $x \in X$, then the following hold:\\[1mm]
\normalfont (i)
\itshape $x + 0_x =_X x$.\\[1mm]
\normalfont (ii)
\itshape The partial additive inverse $(-x)$ of $x$ is unique.\\[1mm]
\normalfont (iii)
\itshape $(-1) \cdot x =_X (-x)$.\\[1mm]
\normalfont (iv)
\itshape $0_{k \cdot x} =_X 0_{x} =_X k \cdot 0_x$, for every $k \in \D K$, and every $x \in X$.\\[1mm]
\normalfont (v)
\itshape $0_x + 0_y =_X 0_{x+y}$, for every $x, y \in X$.\\[1mm]
\normalfont (vi)
\itshape $X_{\tota}$, equipped with the restricted partial linear structure, is a linear space.
\end{corollary}

\begin{proof}
(i) $x + 0_x =_X 1 \cdot x + 0 \cdot x =_X (1 + 0) \cdot x =_X 1 \cdot x =_X x$.\\
(iii) Let $y, z \in X$, such that $x + y =_X 0_x =_X x + z$ and $0_y =_X 0_x =_X 0_z$.  By case (i) we have that
$$y =_X y + 0_y =_X y + 0_x =_X y + x + z =_X 0_x + z =_X 0_z + z =_X z.$$
(iii) $(-1) \cdot x + 1 \cdot x = (-1 + 1)\cdot x =_X 0 \cdot x = 0_x$ and $0_{(-1)\cdot x} =_X 0 \cdot [(-1)\cdot x] =_X [0 \cdot (-1)] \cdot x =_X 0 \cdot x =_X 0_x$. Hence, by the uniqueness of the partial additive of $x$ we get $(-1) \cdot x =_X (-x)$.\\
(iv) By $(\PL_1)$ we have that $0_{kx} =_X 0 \cdot (k\cdot x) =_X k \cdot (0 \cdot x) =_X  (0 \cdot k) \cdot x =_X 0 \cdot x =_X 0_x$.\\
(v) $0_x + 0_y =_X 0 \cdot x + 0 \cdot y =_X 0 \cdot (x+y) =_X 0_{x+y}$.\\
(vi) By $(\PL_3)$ we have that $0 \in X_{\tota}$. If $x, y \in X_{\tota}$, then by case (v) we get $x + y \in X_{\tota}$, and if $k \in \D K$, then by case (iv) we get $k \cdot x \in X_{\tota}$. If $x \in X_{\tota}$, then $0 \cdot x =_X 0_x =_X 0$ and $x +(-x) =_X 0_x =_X 0$.
\end{proof}

\begin{proposition}\label{prp: ispls} 
If $(\C X, ||.||)$ and $(\C Y, ||.||)$ are normed spaces, then  $\M B(\C X, \C Y)$, equipped with the operations of Definition~\ref{def: partialoperations}, is a partial linear space.
\end{proposition}

\begin{proof}
Properties $(\PL_1)$-$(\PL_4)$ are straightforward to show.
%for  $\M B(\C X, \C Y)$. 
To show $(\PL_5)$, by Definition~\ref{def: partialoperator} we 
get
%have that
$$k \cdot T \neq_{\mathsmaller{ \M B(\C X, \C Y)}} 0 :\TOT \exists_{x \in X_0}\big[x \in \dom(T)^{\neq} \vee \big(x \in \dom(T) \ \& \ k \cdot T(x) \neq_Y 0\big)\big].$$
If $x \in X_0$ and $x \in \dom(T)^{\neq}$, then $x \colon 0_T \neq_{\mathsmaller{ \M B(\C X, \C Y)}}  0$, since $\dom(0_T) =_{\mathsmaller{S(\C X)}} \dom(T)$. If $x \in X_0$, $x \in \dom(T)$ and $k \cdot T(x) \neq_Y 0$, then by $(\LinIneq_3)$ on $\C Y$
 %we get 
 $k \neq_{\D K} 0$ and $T(x) \neq_Y 0$. Hence $x \colon T \neq_{\mathsmaller{ \M B(\C X, \C Y)}} 0$.
\end{proof}

\begin{definition}\label{def: plmap}
If $\C X$ and $\C Y$ are partial linear spaces, a $p$-linear map is a total function $T \colon X \to Y$ that satisfies the following properties$:$ $T(0) =_Y 0$, $T(x + x{'}) =_Y T(x) + T(x{'})$ and $T(k\cdot x) =_Y k \cdot T(x)$, for every $x, x{'} \in X$ and every $k \in \D K$. Let $\M L(\C X, \C Y)$ be the set of $p$-linear maps from $\C X$ to $\C Y$, equipped with the equality and inequality in Definitions~\ref{def: function} and~\ref{def: canonicalineq}, respectively.
\end{definition}

Notice, that we cannot prove $T(0) =_Y 0$ by the other two properties of a $p$-linear map, as in the case of a standard linear map. What we only get from each one of them is that $T(0_x) =_Y 0_{T(x)}$.

\begin{proposition}\label{prp: pls1}
If $\C X$ and $\C Y$ are partial linear spaces, then the following hold:\\[1mm]
\normalfont (i)
\itshape $\C Z(\C X) := \{0_x \mid x \in X\}$ is a partial linear subspace of $\C X$.\\[1mm]
 \normalfont (ii)
 \itshape The function $\zeta \colon X \to \C Z(\C X)$, defined by the rule $x \mapsto 0_x$, is in $\M L(\C X, \C Z(\C X))$.\\[1mm]
 \normalfont (iii)
 \itshape $\M L(\C X, \C Y)$ is a partial linear space.\\[1mm]
  \normalfont (iv)
 \itshape If $T \in \M L(\C X, \C Y)$ is total, then $\forall_{x \in X}\big(T(x) \ \mbox{is total} \ \hspace{-1mm}\big)$.
 %the total elements of which are the linear maps from $\C X$ to $\C Y$.
	
\end{proposition}

\begin{proof}
(i) $0 \in \C Z(\C X)$, since $0 =_X 0_0$. If $0_x, 0_y \in \C Z(\C X)$, then $0_x + 0_y =_X 0_{x + y} \in \C Z(\C X)$, and if $k \in \D K$, then $k \cdot 0_x =_X 0_x \in \C Z(\C X)$.  Moreover, $-0_x =_X 0_x$, since $0_x + 0_x =_X 0_{2 \cdot x}  =_X 0_x =_X 0_{0_x}$.\\
(ii) Clearly, $\zeta(0) := 0_0 =_X 0$, and $\zeta(x +y ) := 0_{x+y} =_X 0_x + 0_y =: \zeta(x) + \zeta(y)$. If $k \in \D K$, then $\zeta(k \cdot x) := 0_{k \cdot x} =_X k \cdot 0_x =: k \cdot \zeta(x)$.\\
(iii) If $T, U \in \M L(\C X, \C Y)$, let $(T + U)(x) := T(x) + U(x)$ and $(k \cdot T)(x) := k \cdot T(x)$, for every $x \in X$. Let also $0_{T}(x) := 0_{T(x)}$, for every $x \in X$. Clearly, $T + U, k \cdot T$, and $0_T$ are in $\M L(\C X, \C Y)$. By definition $(0 \cdot T)(x) := 0 \cdot T(x) =_Y 0_{T(x)} =: 0_{T}(x)$, for every $x \in X$. Similarly, $[T + (-T)](x) := T(x) + (-T)(x) =_Y T(x) + (-T(x)) = 0_{T(x)} =: 0_{T}(x)$, for every $x \in X$. Moreover, $[0_{-T}](x) := 0_{(-T)(x)} =_Y 0_{-T(x)} =_Y 0_{T(x)} =: [0_T](x)$, for every $x \in X$. If $x \in X$, then $[0_0](x) := 0_{0(x)} := 0_0 =_Y 0 =: 0(x)$. If $x \in X$, then $[0_{0_T}](x) := 0_{(0_T)(x)} =_Y 0_{T(x)}  =: [0_T](x)$. 
%Let $T \in \M L(\C X, \C Y)$ with $0_T =_{\mathsmaller{\M L(\C X, \C Y)}} 0$. 
To show $(\PL_5)$, let 
$k \cdot T \neq_{\mathsmaller{\M L(\C X, \C Y)}} 0 :\TOT \exists_{x \in X_0}\big(k \cdot T(x) \neq_Y 0\big)$, hence $\exists_{x \in X_0}\big[0_{T(x)} \neq_Y 0 \vee \big(k \neq_{\D K} 0 \ \& \ T(x) \neq_Y 0\big)\big]$.
%\begin{align*}
%k \cdot T \neq_{\mathsmaller{\M L(\C X, \C Y)}} 0 & :\TOT \exists_{x \in X_0}\big(k \cdot T(x) \neq_Y 0\big)\\
%& \To \exists_{x \in X_0}\big[0_{T(x)} \neq_Y 0 \vee \big(k \neq_{\D K} 0 \ \& \ T(x) \neq_Y 0\big)\big].
%\end{align*}
If $x \in X_0$ with $0_{T(x)} \neq_Y 0$, then $x \colon 0_T \neq_{\mathsmaller{\M L(\C X, \C Y)}} 0$. If $x \in X_0$ with $k \neq_{\D K} 0$ and $T(x) \neq_Y 0$, then $x \colon T \neq_{\mathsmaller{\M L(\C X, \C Y)}} 0$.\\
%By the definition of this equality we get $ \dom(T) =_{\mathsmaller{S(\C X)}} \dom(0_T) =_{\mathsmaller{S(\C X)}} \dom(0) =_{\mathsmaller{S(\C X)}} X$. If $T$ is a linear map from $\C X$ to $\C Y$
Case (iv) follows immediately from the above definition of $0_T$ in $\M L(\C X, \C Y)$.
\end{proof}

\section{Partial projections on a Hilbert space}
\label{sec: partialproj}

In this section we study bounded partial operators on a Hilbert space $\C H$, and especially partial projections on $\C H$. Classically, all bounded operators $T$ on $\C H$, their domain $\dom(T)$ of which is a proper subspace of $\C H$, can be extended to a total bounded operator $\widehat{T}$ on the whole space $\C H$ by defining $\widehat{T}(x) := 0$, for every $x \in \dom(T)^{\perp}$. Since, classically, $\dom(T) \vee \dom(T)^{\perp} =_{\mathsmaller{S(\C H)}} H$, we may consider that all bounded operators on $\C H$ are total. Hence, classically, it is only when an operator is not bounded, or unbounded, that the domain of definition of $T$ plays a role (see~\cite{Go66}). As we have already explained, constructively, property $L \vee L^{\perp} =_{\mathsmaller{S(\C H)}} H$ cannot be accepted for every subspace $L$ of $\C H$. Hence, the constructive study of bounded and partial operators on a Hilbert space $\C H$ is meaningful.

\begin{definition}\label{def: partialproj}Let $\M B(\C H)$ be the totality $\M B(\C H, \C H)$ of bounded partial operators from $\C H$ to $\C H$. A bounded total operator $T$ on $\C H$ is a bounded partial operator $T$ on $\C H$ with $\dom(T) =_{\mathsmaller{S(\C H)}} 1$. A partial projection on $\C H$ is an idempotent and self-adjoint bounded partial operator on $\C H$. We denote their totality by $\M P(\C H)$. A partial projection on $\C H$ is called strict, if there is $x \in \dom(P)$ with $P(x) \neq_H 0$. We denote their totality by $\M P^*(\C H)$.
In order to formulate a dual notion of inequality on $\M B(\C H, \C H{'})$, we use the following definitions of equality and inequality\footnote{In the case of Hilbert spaces the inequality on bounded partial operators is influenced by the geometry of the space. The difference between the inequality in Definition~\ref{def: partialoperator} and that on partial projections is due to Corollary~\ref{cor: corBBcrucial}(ii). We cannot define the inequality on $\C B(\C H)$ quantifying over the whole space $H$, since $0$ is always in $\dom(T) \wedge \dom(U)^{\perp}$, i.e., any two bounded partial operators would be inequal. We could have defined the equality on $\C B(\C H)$ by quantifying over $H_0$, in order to keep the duality between the equality and inequality on $\C B(\C H)$. We prefer though, to keep our weaker equality (the one which uses quantification over $H$) together with the inequality that employs quantification over $H_0$.}
\begin{align*}
	T =_{\mathsmaller{\M B(\C H, \C H{'})}} U & :\TOT \forall_{x \in H}\bigg[\big(x \in \dom(T) \To x \in \dom(U)\big) \wedge \\
	& \ \ \ \ \ \ \ \ \ \ \ \ \ \big(x \in \dom(U) \To x \in \dom(U)\big) \wedge \\
	& \ \ \ \ \ \ \ \ \ \ \ \ \ \big(x \in \dom(T) \wedge \dom(U) \To T(x) =_{H{'}} U(x)\big)\bigg],
\end{align*}
\vspace{-6mm}
\begin{align*}
	T \neq_{\mathsmaller{\M B(\C H, \C H{'})}} U & :\TOT \exists_{x \in H_0}\bigg[\big(x \in \dom(T) \ \& \ x \in \dom(U)^{\perp}\big) \vee \\
	& \ \ \ \ \ \ \ \ \ \ \ \ \ \big(x \in \dom(U) \ \& \ x  \in \dom(T)^{\perp}\big) \vee \\
	& \ \ \ \ \ \ \ \ \ \ \ \ \ \big(x \in \dom(T) \wedge \dom(U) \ \& \ T(x) \neq_{H{'}} U(x)\big)\bigg].
\end{align*} 
If $T \neq_{\mathsmaller{\M B(\C H, \C H{'})}} U$ and $x \in H_0$ witnesses the inequality $T \neq_{\mathsmaller{\M B(\C H, \C H{'})}} U$, then we write $x \colon 	T \neq_{\mathsmaller{\M B(\C H, \C H{'})}} U$.
\end{definition}

Clearly, if $x \in H_0$, then $x \colon 0 \neq_{\mathsmaller{\C B(\C H)}} I_H$. 
Notice the similarity between the definition of equality and inequality on total functions in Definition~\ref{def: canonicalineq} and the definition of equality and inequality on partial bounded operators between Hilbert spaces. 
%Actually, 
%with the excemption of using $H_0$, for technic, instead of $H$, 
%the equality and inequality between total bounded operators is a special case of the equality and inequality between bounded partial operators, respectively.
 As we show in the proof of Theorem~\ref{thm: bijection}(i), the projection $P_{\B L}$ induced in Theorem~\ref{thm: crucial} by some $\B L \leq  \C H$ is a partial projection on $\C H$.
In general, we cannot show that  $T \neq_{\mathsmaller{\M B(\C H, \C H{'})}} U$ is an apartness relation, although there are not many examples of this phenomenon\footnote{Another such inequality is the canonical inequality on the exterior union of a family of sets (see~\cite{Pe20}, p.~43).}.

\begin{proposition}\label{prp: ext1}
\label{prp: ineqext} 
The inequality relation $T \neq_{\mathsmaller{\M B(\C H, \C H{'})}} U$ is extensional.
\end{proposition}

\begin{proof}
If $T =_{\mathsmaller{\M B(\C H, \C H{'})}} T{'}$, $U =_{\mathsmaller{\M B(\C H, \C H{'})}} U{'}$, and $T \neq_{\mathsmaller{\M B(\C H, \C H{'})}} U$, we show that $T{'} \neq_{\mathsmaller{\M B(\C H, \C H{'})}} U{'}$. If $x_0 \colon T \neq_{\mathsmaller{\M B(\C H, \C H{'})}} U$, we show that 
$x_0 \colon T{'} \neq_{\mathsmaller{\M B(\C H, \C H{'})}} U{'}$. Suppose first that $x_0 \in \dom(T) \ \& \ x_0 \in \dom(U))^{\perp}$. As $T =_{\mathsmaller{\M B(\C H, \C H{'})}} T{'}$, we get  $x_0 \in \dom(T{'})$.
Since $x_0 \in \dom(U)^{\perp}$, by the equality $U =_{\mathsmaller{\M B(\C H, \C H{'})}} U{'}$ we 
get $x_0 \in \dom(U{'})^{\perp}$. If $x _0 \in \dom(U) \ \& \ x_0  \in \dom(T)^{\perp}$,
we work similarly. If $x_0 \in \dom(T) \wedge \dom(U) \ \& \ T(x_0) \neq_{H{'}} U(x_0)$, then by the supposed equalities we have that $x_0 \in \dom(T{'}) \wedge \dom(U{'})$, 
and by the extensionality of $\neq_{H{'}}$ we get $T{'}(x_0) =_{H{'}} T(x_0) \neq_{H{'}} U(x_0) =_{H{'}} U{'}(x_0)$.
\end{proof}

\begin{proposition}\label{prp: zeroisse}
The function $\zeta \colon \M B(\C H, \C H{'}) \to \C Z(\M B(\C H, \C H{'}))$, where $T \mapsto 0_T$, is strongly extensional.
\end{proposition}

\begin{proof}
Let $x \in H_0$ with $x \colon 0_T \neq_{\mathsmaller{\M B(\C H, \C H{'})}} 0_U$. Since both $0_T$ and $0_U$ are zero on $\dom(T) \wedge \dom(U)$, the only two possibilities for $x \colon 0_T \neq_{\mathsmaller{\M B(\C H, \C H{'})}} 0_U$ are $x \in \dom(0_U)$ and $x  \in \dom(0_T)^{\perp}$, or $x \in \dom(0_T) $ and $x  \in \dom(0_U)^{\perp}$. We treat only the first case, and for the second we work similarly. Since the first case is rewritten as $x \in \dom(U)$ and $x  \in \dom(T)^{\perp}$, we get immediately that $x \colon T \neq_{\mathsmaller{\M B(\C H, \C H{'})}} U$.
\end{proof}

 The bijection between orthocomplemented subspaces of $\C H$ and partial projections on $\C H$ that we establish next captures the computational content of the classical bijection between $S(\C H)$ and the projections on $\C H$
%  i.e., the self-adjoint and idempotent elements of $\C B(\C H)$ 
  (see~\cite{KR83}, p.~110). The fact that a total projection $P$ on $\C H$ induces constructively a subspace of $\C H$ on which $P$ projects is already described in~\cite{BB85}, p.~371.  Theorem~\ref{thm: bijection} explains why partial projections on $\C H$ correspond to characteristic functions of complemented subsets.
  % as we suggested in the introduction.

\begin{theorem}\label{thm: bijection}
%By Theorem~\ref{thm: crucial}$($ii$)$ 
Let the assignment routine\footnote{See definition~\ref{def: function}.} $i \colon  \B S(\C H) \sto  \M P(\C H)$, defined by the rule
$$\B L \mapsto P_{\B L}^1,$$
and let the assignment routine
$j \colon \M P(\C H) \sto \B S(\C H)$, defined by the rule$:$
$$P \mapsto \B L_P := (L_P^1, L_P^0),$$
$$L_P^1 := \big\{P(x) \mid x \in \dom(P)\} =_{S(\C H)} \{x \in \dom(P) \mid P(x) =_H x\big\} \leq \dom(P) \leq \C H,$$
$$L_P^0 := \big\{x \in \dom(P) \mid P(x) =_H 0\big\}  =: \Ker(P) \leq \dom(P) \leq \C H.$$
If we define the partial order on $\M P(\C H)$ by the rule$:$
$$P  \leq Q :\TOT j(P) \leq j(Q) :\TOT \B L_P \leq \B L_Q ,$$
%:\TOT L_P^1 \leq L_Q^1 \ \& \ L_Q^0 \leq L_P^0,$$
and if we define the inequality on $\B S(\C H)$ by the rule$:$\footnote{This definition is a generalisation of Definition~\ref{def: ineqSloc}.}
$$\B L \neq_{\mathsmaller{\B S(\C H)}} \B M :\TOT i(\B L) \neq_{\mathsmaller{\M P(\C H)}} i(\B M) :\TOT P_{\B L}^1 \neq_{\mathsmaller{\M P(\C H)}} P_{\B M}^1,$$
then the following hold:\\[1mm]
\normalfont (i)
\itshape The routines $i$ and $j$ are well-defined functions that are 
inverse to each other.\\[1mm]
\normalfont (ii)
\itshape $i$ and $j$ are strongly extensional functions and injections.\\[1mm]
\normalfont (iii)
\itshape $i$ and $j$ preserve the corresponding partial orders.\\[1mm]
\normalfont (iv)
\itshape $i$ and $j$ preserve strictness.\\[1mm]
\normalfont (v)
\itshape $i$ and $j$ preserve totality.
%is a bijection between the orthocomplemented subspaces $\B S(\C H)$ and the partial projections $\M P(\C H)$ of a Hilbert space $\C H$ that preserves strictness.
\end{theorem}

\begin{proof}
(i) First, we show that $i$ is well-defined, i.e., $P_{\B L}^1 \in \M P(\C H)$. By Definition~\ref{def: partialoperator} we have that
$$\dom(P_{\B L}^1 \circ P_{\B L}^1) := \big\{x \in L^1 \oplus L^0 \mid P_{\B L}^1(x) \in  L^1 \oplus L^0\big\}
=_{\mathsmaller{S(\C H)}} L^1 \oplus L^0.$$
The remaining defining properties of a partial projection are shown in the proof of Theorem~\ref{thm: crucial}(ii). Next, we show that $i$ is a function , i.e., if $\B L =_{\mathsmaller{\B S(\C H)}} \B M$, then $P_{\B L}^1 =_{\mathsmaller{\M P(\C H)}} P_{\B M}^1$, for every $\B L$, $\B M$ in $\B S(\C H)$.
Clearly, if $\B L =_{\mathsmaller{\B S(\C H)}} \B M$, then $\dom(\B L) =_{\mathsmaller{S(\C H)}} \dom(\B M)$, i.e., $\dom(P_{\B L}^1) =_{\mathsmaller{S(\C H)}} \dom(P_{\B M}^1)$. If $x \in \dom(P_{\B L}^1))$ with $x =_H P_{\B L}^1(x) + P_{\B L}^0(x) =_H P_{\B M}^1(x) + P_{\B M}^0(x)$, then by the equalities $L^1 =_{\mathsmaller{S(\C H)}} M^1$ and $L^0 =_{\mathsmaller{S(\C H)}} M^0$ and the uniqueness of the representation of $0$ in $L^1 \oplus L^0$ we get $P_{\B L}^1(x) =_H P_{\B M}^1(x)$, and since $x \in H$ is arbitrary, we get $P_{\B L}^1 =_{\mathsmaller{\M P(\C H)}} P_{\B M}^1$. Next, we show that $j$ is well-defined. 
We have that $L_P^1 \perp L_P^0$, since by the self-adjointness of $P$, if $x^1 \in L_P^1$ and $x^0 \in L_P^0$, then 
$$\langle x^1, x^0 \rangle =_{\Real} \langle P(x^1), x^0 \rangle =_{\Real} \langle x^1, P(x^0) \rangle =_{\Real} \langle x^1, 0 \rangle =_{\Real} 0.$$
%Moreover, $L_P^1 \wedge L_P^0 =_{\mathsmaller{S(\C H)}} 0$, since if $x \in L_P^1 \cap L_P^0$, then $x =_H P(x) =_H 0$.
Next, we show that $j$ is a function, i.e., if $P =_{\mathsmaller{\M P(\C H)}} Q$, then $\B L_P =_{\mathsmaller{\B S(\C H)}} \B L_Q$, for every $P, Q \in \M P(\C H)$. By Definition~\ref{def: partialproj} we have that
\begin{align*}
	P =_{\mathsmaller{\M P(\C H)}} U & :\TOT \forall_{x \in H}\bigg[\big(x \in \dom(P) \To x \in \dom(Q)\big) \wedge \\
	& \ \ \ \ \ \ \ \ \ \ \ \ \ \big(x \in \dom(Q) \To x \in \dom(P)\big) \wedge \\
	& \ \ \ \ \ \ \ \ \ \ \ \ \ \big(x \in \dom(P) \wedge \dom(Q) \To P(x) =_{H} Q(x)\big)\bigg],
\end{align*}
from which the equalities $L_P^1 =_{\mathsmaller{S(\C H)}}  L_Q^1$ and $L_P^0 =_{\mathsmaller{S(\C H)}}  L_Q^0$ follow immediately. The equality $\B L =_{\mathsmaller{\B S(\C H)}} \B L_{P_{\B L}^1}$ follows by the equalities
$$ L_{P_{\B L}^1}^1 := \big\{P_{\B L}^1(x) \mid x \in L^1 \oplus L^0\big\} =_{\mathsmaller{S(\C H)}} L^1,$$
$$ L_{P_{\B L}^1}^0 := \big\{x \in L^1 \oplus L^0 \mid P_{\B L}^1(x) =_H 0 \big\} =_{\mathsmaller{S(\C H)}} L^0.$$
To show the equality $P =_{\mathsmaller{\M P(\C H)}} P^1_{\B L_P}$, we first show that
$$\dom(P) =_{\mathsmaller{S(\C H)}} \dom(P^1_{\B L_P}) := L_P^1 \oplus L_P^0.$$
Clearly, $L_P^1 \oplus L_P^0 \leq \dom(P)$. To show $\dom(P) \leq L_P^1 \oplus L_P^0$, we observe that if $x \in \dom(P)$, then
$$x =_H P(x) + (x - P(x)), \ \ \ P(x) \in L_P^1 \ \& \ x -P(x) \in L_P^0,$$ 
which also implies that $P^1_{\B L_P}(x) =_H P^1_{\B L_P}(P(x) + (x - P(x))) =_H P(x)$, for every $x \in \dom(P)$.\\
(ii) Injectivity and strong extensionality of $i$ is the equivalence $P_{\B L}^1 \neq_{\mathsmaller{\M P(\C H)}} P_{\B M}^1 \TOT \B L \neq_{\mathsmaller{\B S(\C H)}} \B M$, for every $\B L, \B M \in \B S(\C H)$, which holds by the above definition. By the same definition, case (i), and the extensionality of $\neq_{\mathsmaller{\M P(\C H)}}$ we also get
$$\B L_P \neq_{\mathsmaller{\B S(\C H)}} \B L_Q :\TOT P^1_{\B L_P} \neq_{\mathsmaller{\M P(\C H)}} P^1_{\B L_Q} \TOT P \neq_{\mathsmaller{\M P(\C H)}} Q,$$
and hence the injectivity and strong extensionality of $j$ follow.\\
(iv) If $l^1 \in L^1$ with $l^1 \neq_H 0$, then $P_{\B L}^1(l^1) := l^1 \neq_H 0$.
% i.e., $P_{\B L}^1$ is strict.
Conversely, if $P$ is strict, then $L_P^1$ is strict.\\
(v) Clearly, if $\B L$ is total, then $\dom(P_{\B L}^1) := \dom(\B L) =_{\mathsmaller{S(\C H)}} H$, and if $P$ is total, then $\dom(P) =_{\mathsmaller{S(\C H)}}  H =_{\mathsmaller{S(\C H)}}  L_P^1 \oplus L_P^0 =_{\mathsmaller{S(\C H)}} \dom(\B L_P)$.
\end{proof}

%Clearly, the above proof is the constructive ``partial'' analague of the classical proof

\begin{corollary}\label{cor: normproj}
	If $P$ is a partial projection on $\C H$, then its corresponding subspaces $L_P^1, L_P^0$ are located in $\dom(\B L)$, and if $P$ is strict, then $P$ is normed in $\dom(\B L)$ with $||P|| =_{\Real} 1$.
\end{corollary}

\begin{proof}
The locatedness of $L_P^1$ and $L_P^0 $ follows by Theorem~\ref{thm: crucial}(v). If $P$ is strict, then by the above equality $P =_{\mathsmaller{\M P(\C H)}} P^1_{\B L_P}$ and Theorem~\ref{thm: crucial}(ii)
we get that $P$ is normed with $||P|| =_{\Real} 1$.
\end{proof}

\begin{remark}\label{rem: LP1}
	If $P \in \M P(\C H)$, then $L_P^1 =_{\mathsmaller{S(\C H)}} \{x \in \dom(P) \mid ||P(x)|| =_{\Real} ||x||\}$.
\end{remark}

\begin{proof}
Clearly, $L_P^1 \subseteq \{x \in \dom(P) \mid ||P(x)|| =_{\Real} ||x||\}$. If $x \in \dom(P)$ with $||P(x)|| =_{\Real} ||x||$, then $||x||^2 =_{\Real} ||P(x) + (x - P(x))||^2 =_{\Real} ||P(x)||^2 + ||x - P(x)||^2$, hence $||x - P(x)||^2 =_{\Real} 0$, i.e., $x =_H P(x)$ and $x \in L_P^1$.
\end{proof}

\begin{remark}\label{rem: meet}
	If $\B L, \B M \leq \C H$ with $\B L \leq \B M$, then $\dom(\B L) \wedge \dom(\B M) =_{\mathsmaller{S(\C H)}} L^1 \oplus (L^0 \wedge M^1) \oplus M^0$.
\end{remark}

\begin{proof}
	By Proposition~\ref{prp: closedMod} and Theorem~\ref{thm: crucial}(i) we use modularity twice as follows:
	\begin{align*}
		\dom(\B L) \wedge \dom(\B M)  & := (L^1 \vee L^0) \wedge (M^1 \vee M^0)\\
		& =_{\mathsmaller{S(\C H)}} [(M^1 \vee M^0) \wedge L^0] \vee L^1 \ \ \ \ \ \  [L^1 \leq M^1 \leq M^1 \vee M^0]\\
		& =_{\mathsmaller{S(\C H)}} [(L^0 \wedge M^1) \vee M^0] \vee L^1 \ \ \ \ \ \ [M^0 \leq L^0]\\
		& =_{\mathsmaller{S(\C H)}} L^1 \oplus [(L^0 \wedge M^1) \vee M^0] \ \ \ \ \ \ [L^1 \perp (L^0 \wedge M^1) \ \& \ L^1 \perp M^0]\\
		& =_{\mathsmaller{S(\C H)}} L^1 \oplus (L^0 \wedge M^1) \oplus M^0 \ \ \ \ \ \ \ \  [M^0 \perp L^0 \wedge M^1].\qedhere 
%		& \leq L^1 \oplus L^0. \qedhere
	\end{align*}
\end{proof}

If $\B L \leq \C H$ and $\dom(\B L) =_{\mathsmaller{S(\C H)}} H =_{\mathsmaller{S(\C H)}} \dom(\B M)$, then, clearly, $\dom(\B L) \wedge \dom(\B M) =_{\mathsmaller{S(\C H)}} H$, but also, according to Remark~\ref{rem: meet}, we have that $L^1 \oplus (L^0 \wedge M^1) \oplus M^0 =_{\mathsmaller{S(\C H)}} H$. To show this, let $x \in H$. By the totality of $\B L$ there are $l^1 \in L^1$ and $l^0 \in L^0$ with $x =_H l^1 + l^0$. By the totality of $\B M$ there are $m^1 \in M^1$ and $m^0 \in M^0$ with $l^0 =_H m^1 + m^0$. Hence, $x =_H l^1 + m^1 + m^0$, and since $M^1 \ni m^1 =_H l^0 - m^0 \in L^0$, we get that $m^1 \in L^0 \wedge M^1$.
Next, we show how Proposition~\ref{prp: equivorder} looks like in the complemented framework of orthocomplemented subspaces and partial projections. 
In contrast to the one-dimensional approach, both partial projection $P_{\B L}^1$ and $P_{\B L}^0$ are involved in the formulation of Theorem~\ref{thm: lescomp1}.

\begin{theorem}\label{thm: lescomp1}
If $\B L \leq \C H$ and $\B M \leq \C H$, then the following hold$:$\\[1mm]
\normalfont (i)
\itshape $\B L \leq \B M$ if and only if $P_{\B M}^1 \circ P_{\B L}^1 =_{\mathsmaller{\M P(\C H)}} P_{\B L}^1$ and $P_{\B L}^0\circ P_{\B M}^0 =_{\mathsmaller{\M P(\C H)}} P_{\B L}^0$.\\[1mm]
\normalfont (iia)
\itshape If $\B L \leq \B M$, then 
$$\dom(P_{\B L}^1 \circ P_{\B M}^1) =_{\mathsmaller{S(\C H)}} \dom(P_{\B M}^0 \circ P_{\B L}^0) =_{\mathsmaller{S(\C H)}} \dom(\B L) \wedge \dom(\B M),$$
and $P_{\B L}^1 \circ P_{\B M}^1 = P_{\B L}^1$ and 
$P_{\B M}^0 \circ P_{\B L}^0 = P_{\B M}^0$ on $\dom(\B L) \wedge \dom(\B M)$.\\[1mm]
\normalfont (iib)
\itshape If $L^1 \vee M^0 \leq \dom(P_{\B L}^1 \circ P_{\B M}^1) \wedge \dom(P_{\B M}^0 \circ P_{\B L}^0)$, and if $P_{\B L}^1 \circ P_{\B M}^1 = P_{\B L}^1$ 
%on $L^1 \wedge M^0$, 
and $P_{\B M}^0 \circ P_{\B L}^0 = P_{\B M}^0$ on $L^1 \vee M^0$, then $\B L \leq \B M$.\\[1mm]
\normalfont (iiia)
\itshape If $\B L \leq \B M$ , then, for every $x \in \dom(\B L) \wedge \dom(\B M)$, we have that
$$||P_{\B L}^1(x)|| \leq ||P_{\B M}^1(x)|| \ \ \ \& \ \ \  ||P_{\B M}^0(x)|| \leq ||P_{\B L}^0(x)||.$$
\normalfont (iiib)
\itshape If $\dom(\B L) =_{\mathsmaller{S(\C H)}} \dom(\B M)$ and, for every $x \in \dom(\B L)$, we have that $||P_{\B L}^1(x)|| \leq ||P_{\B M}^1(x)||$ and $||P_{\B M}^0(x)|| \leq ||P_{\B L}^0(x)||$, then $\B L \leq \B M$.\\[1mm]
\normalfont (iva)
\itshape If $\B L \leq \B M$, then, for every $x \in \dom(\B L) \wedge \dom(\B M)$, we have that
$$\langle P_{\B L}^1(x), x \rangle \leq \langle P_{\B M}^1(x), x\rangle \ \ \ \& \ \ \  \langle P_{\B M}^0(x), x \rangle \leq \langle P_{\B L}^0(x), x\rangle.$$
\normalfont (ivb)
\itshape If $\dom(\B L) =_{\mathsmaller{S(\C H)}} \dom(\B M)$, and, for every $x \in \dom(\B L)$, we have that $\langle P_{\B L}^1(x), x \rangle \leq \langle P_{\B M}^1(x), x\rangle$ and $\langle P_{\B M}^0(x), x \rangle \leq \langle P_{\B L}^0(x), x\rangle$, then $\B L \leq \B M$.
\end{theorem}

\begin{proof}
(i) If $L^1 \leq M^1$ and $M^0 \leq L^0$, then by Definition~\ref{def: partialoperator} we have that 
$$\dom(P_{\B M}^1 \circ P_{\B L}^1) := \{x \in L^1 \oplus L^0 \mid P_{\B L}^1(x) \in M^1 \oplus M^0\} =_{\mathsmaller{S(\C H)}} L^1 \oplus L^0 =_{\mathsmaller{S(\C H)}} \dom(P_{\B L}^1),$$
$$\dom(P_{\B L}^0 \circ P_{\B M}^0) := \{x \in M^1 \oplus M^0 \mid P_{\B M}^0(x) \in L^1 \oplus L^0\} =_{\mathsmaller{S(\C H)}} M^1 \oplus M^0 =_{\mathsmaller{S(\C H)}} \dom(P_{\B M}^0).$$
If $x =_H l^1 + l^0 \in L^1 \oplus L^0$, then $P_{\B M}^1 \big(P_{\B L}^1(x)\big) =_H P_{\B M}^1 (l^1) =_H l^1 =_H  P_{\B L}^1(x)$, and if $y =_H m^1 + m^0 \in M^1 \oplus M^0$, then $P_{\B L}^0 \big(P_{\B M}^0(y)\big) =_H P_{\B L}^0 (m^0) =_H m^0 =_H  P_{\B M}^0(y)$. Conversely, by the equality $P_{\B M}^1 \circ P_{\B L}^1 =_{\mathsmaller{\M P(\C H)}} P_{\B L}^1$ we get $\dom(P_{\B M}^1 \circ P_{\B L}^1)
=_{\mathsmaller{S(\C H)}} \dom(P_{\B L}^1) =_{\mathsmaller{S(\C H)}} L^1 \oplus L^0$. Hence, if $l^1 \in L^1$, then $l^1 \in M^1 \oplus M^0$. Let $l^1 =_H m^1 + m^0 \in M^1 \oplus M^0$. By the equality of partial functions we get $P_{\B M}^1\big(P_{\B L}^1(l^1)\big) =_H = m^1 =_H l^1$, hence, $l^1 \in M^1$. Working in a similar way, we show that $M^0 \leq L^0$.\\
(iia) By Definition~\ref{def: partialoperator} we have that
$$\dom(P_{\B L}^1 \circ P_{\B M}^1) := \{x \in M^1 \oplus M^0 \mid P_{\B M}^1(x) \in L^1 \oplus L^0\} \leq \C H,$$
$$\dom(P_{\B M}^0 \circ P_{\B L}^0) := \{x \in L^1 \oplus L^0 \mid P_{\B L}^0(x) \in M^1 \oplus M^0\ \leq \C H.$$
By the hypotheses $L^1 \leq M^1$ and $M^0 \leq L^0$ we have that 
$$L^1, L^0 \wedge M^1, M^0 \leq \dom(P_{\B L}^1 \circ P_{\B M}^1) \ \ \ \& \ \ \ L^1, L^0 \wedge M^1, M^0 \leq \dom(P_{\B M}^0 \circ P_{\B L}^0),$$
hence, by Remark~\ref{rem: meet} we have that
$$\dom(P_{\B L}^1 \circ P_{\B M}^1) \geq L^1 \oplus (L^0 \wedge M^1) \oplus M^0 \ \ \ \& \ \ \ \dom(P_{\B M}^0 \circ P_{\B L}^0) \geq L^1 \oplus (L^0 \wedge M^1) \oplus M^0.$$
For the inequality $\dom(P_{\B L}^1 \circ P_{\B M}^1) \leq L^1 \oplus (L^0 \wedge M^1) \oplus M^0$, let $x \in \dom(P_{\B L}^1 \circ P_{\B M}^1)$ with $m^1 \in M^1$ and $m^0 \in M^0$, such that $x =_H m^1 + m^0$. Since $m^1 \in L^1 \oplus L^0$, there are $l^1 \in L^1$ and $l^0 \in L^0$, such that $m^1 =_H l^1 + l^0$. Hence, $x =_H l^1 + l^0 + m^0$ with $L^0 \ni l^0 = m^1 - l^1 \in M^1$, i.e., $l^0 \in M^1 \wedge L^0$. For the inequality $\dom(P_{\B M}^0 \circ P_{\B L}^0) \leq L^1 \oplus (L^0 \wedge M^1) \oplus M^0$, we work similarly. Let now $x = l^1 + k^{01} + m^0 \in \dom(P_{\B L}^1 \circ P_{\B M}^1)$, where $l^1 \in L^1, k^{01} \in L^0 \wedge M^1$, and $m^0 \in M^0$. Clearly, we have that
$$P_{\B L}^1\big(P_{\B M}^1(x)\big) =_H P_{\B L}^1(l^1 + k^{01}) =_H l^1 =_H P_{\B L}^1(x) \ \ \ \& \ \ \ 
P_{\B M}^0\big(P_{\B L}^0(x)\big) =_H P_{\B L}^0(k^{01} + m^0) =_H m^0 =_H P_{\B M}^0(x).$$
(iib) Let $l^1 \in L^1$. Since $l^1 \in \dom(P_{\B L}^1 \circ P_{\B M}^1)$, there are $m^1 \in M^1$ and $m^0 \in M^0$, such that $l^1 =_H m^1 + m^0$. By the first equality 
%we get 
$P_{\B L}^1(m^1) =_H l^1$. By the second equality 
%we get
$P_{\B M}^0(P_{\B L}^0(l^1)) =_H P_{\B M}^0(0) =_H 0 =_H P_{\B M}^0(l^1)$, i.e., $m^0 =_H 0$, hence, $l^1 =_H m^1 \in M^1$. Next, let $m^0 \in M^0$. Since $m^0 \in \dom(P_{\B M}^0 \circ P_{\B L}^0)$, there are $l^1 \in L^1$ and $l^0 \in L^0$ with $m^0 =_H l^1 + l^0$. By the second equality we get $P_{\B M}^0(l^0) =_H m^0$. By the first equality we get $P_{\B L}^1(P_{\B M}^1(m^0)) =_H P_{\B L}^1(0) =_H 0 =_H P_{\B L}^1(m^0)$, i.e., $l^1 =_H 0$, hence, $m^0 =_H l^0 \in L^0$.\\
(iiia) By Remark~\ref{rem: meet} if $x \in \dom(\B L) \wedge \dom(\B M)$, there are $l^1 \in L^1, k^{01} \in L^0 \wedge M^1$ and $m^0 \in M^0$, such that $x =_H l^1 + k^{01} + m^0$. Hence,
$$||P_{\B M}^1(x)||^2 =_{\Real} ||l^1 + k^{01}||^2 =_{\Real} ||l^1||^2 + ||k^{01}||^2 \geq ||l^1||^2 =_{\Real} ||P_{\B L}^1(x)||^2,$$
$$||P_{\B L}^0(x)||^2 =_{\Real} ||k^{01} + m^0||^2 =_{\Real} ||k^{01}||^2 + ||m^{0}||^2 \geq ||m^0||^2 =_{\Real} ||P_{\B M}^0(x)||^2.$$
(iiib) If $l^1 \in L^1$ and $m^0 \in M^0$, we have that
$$||l^1||^2 =_{\Real} ||P_{\B L}^1(l^1)||^2 \leq  ||P_{\B M}^1(l^1)||^2 \leq ||l^1||^2,$$
$$||m^0||^2 =_{\Real} ||P_{\B M}^0(m^0)||^2 \leq  ||P_{\B L}^0(m^0)||^2 \leq ||m^0||^2,$$
hence $||P_{\B M}^1(l^1)||^2 =_{\Real} ||l^1||^2$ and $||P_{\B L}^0(m^0)||^2 =_{\Real} ||m^0||^2$. Consequently, $||P_{\B M}^1(l^1)|| =_{\Real} ||l^1||$ and $||P_{\B L}^0(m^0)|| =_{\Real} ||m^0||$, hence, by Remark~\ref{rem: LP1} we get $l^1 \in L^1_{P_{\B M}} =_{\mathsmaller{S(\C H)}} M^1$ and $m^0 \in L^0_{P_{\B L}} =_{\mathsmaller{S(\C H)}} L^0$.\\
Cases (iva) and (ivb) are shown trivially by (iiia) and (iiib), respectively.
\end{proof}

Notice, that if $\B L$ and $\B M$ are total, then we recover all equivalences of Proposition~\ref{prp: equivorder}.

\section{Complemented quantum logic}
\label{sec: typeI}
%MAYBE OF TYPE (I)

%Operations of type $(\ti)$ on $\B S(\C H)$

In this section we define the basic  operations 
%of type $(\ti)$ 
on $\B S(\C H)$, and we prove that $\B S(\C H)$ is what we call a complemented quantum lattice (Proposition~\ref{prp: ComplQL}). Due to the definition of negation of an orthocomplemented subspace, complemented quantum logic satisfies all properties of classical negation $L^{\perp}$ in $\ClQL$, and in this way it is a constructive logical system closer to $\ClQL$ than $\CoQL$.

\begin{definition}\label{def: operations2} The following operations are defined on {$\B S(\C H)$}$:$
	%	If $\B L, \B M \leq \C H$, we define the following operations between them$:$
	%orthocomplemented subspaces:
	$$\B 1 := (1, 0) := (H, \{0\}),$$
	$$\B 0 := (0, 1) := (\{0\}, H),$$
	$$\B L \wedge \B M := \big(L^1 \wedge M^1, L^0 \vee M^0\big),$$
	$$\B L \vee \B M := \big(L^1 \vee M^1, L^0 \wedge M^0\big),$$
	$$-\B L := \big(L^0, L^1\big),$$
	$$\B L -\B  M := \B L \wedge (- \B M),$$
	$$\B L \To \B M := (- \B L) \vee \B M,$$
	$$\B L \TOT \B M := (\B L \To \B M) \wedge (\B M \To \B L),$$ 
	$$\neg \B L := \B L \To \B 0.$$
If $(\B L_i)_{i \in I}$ is a family of orthocomplemented subspaces of $\C H$ over the index-set $I$, let
$$\bigwedge_{i \in I} \B L_i := \bigg(\bigwedge_{i \in I}L_i^1, \ \bigvee_{i \in I}L_i^0\bigg),$$
$$\bigvee_{i \in I} \B L_i := \bigg(\bigvee_{i \in I}L_i^1, \ \bigwedge_{i \in I}L_i^0\bigg).$$
The relation of orthogonality on $\B S (\C H)$ is defined by the rule
$$\B L \perp \B M :\TOT \B L \leq (- \B M).$$
If $\B L \leq \C H$ is understood as a proposition, its derivability $\vdash \B L$ is defined as a proof of $\B L =_{\mathsmaller{\B S (\C H)}} \B 1$.

\end{definition}

It is straightforward to show that $\B L \wedge \B M \leq \C H$ , $\B L \vee \B M \leq \C H$, 
% are in $\B {\C S}(\C H)$: 
$\bigwedge_{i \in I} \B L_i \leq \C H$, and $\bigvee_{i \in I} \B L_i \leq \C H$.
Clearly, $\B L := (L, L^{\perp}) \leq \C H$, and $\neg \B L := - \B L \vee \B 0 =_{\mathsmaller{\B S(\C H)}} - \B L$. The equality
$\B 0 \To \B L := \B 1 \vee L =_{\mathsmaller{\B S(\C H)}} \B 1$
expresses that within $\B S(\C H)$, seen as a logical system, $\EFQ$ is derived.
In general, the implication in $\B S(\C H)$ does not satisfy the adjunction of a Heyting algebra.
% ....$$(\B K \wedge \B L) \leq \B M \ \mbox{if and only if}  \ \B K \leq (\B L \To \B M).$$
By Theorem~\ref{thm: crucial}(i) we have that
$$\dom(\B L \wedge \B M) =_{\mathsmaller{S(\C H)}} (L^1 \wedge M^1) \oplus (L^0 \vee M^0),$$
$$\dom(\B L \vee \B M) =_{\mathsmaller{S(\C H)}} (L^1 \vee M^1) \oplus (L^0 \wedge M^0),$$
$$\dom(\B L - \B M) =_{\mathsmaller{S(\C H)}} (L^1 \wedge M^0) \oplus (L^0 \wedge M^1),$$
$$\dom\bigg(\bigwedge_{i \in I} \B L_i\bigg) =_{\mathsmaller{S(\C H)}} \bigg(\bigwedge_{i \in I}L_i^1\bigg) \oplus \bigg(\bigvee_{i \in I}L_i^0\bigg),$$
$$\dom\bigg(\bigvee_{i \in I} \B L_i\bigg) =_{\mathsmaller{S(\C H)}} \bigg(\bigvee_{i \in I}L_i^1\bigg) \oplus \bigg(\bigwedge_{i \in I}L_i^0\bigg).$$

As total orthocomplemented subspaces $\C H$ correspond to located subspaces of $\C H$, by Proposition~\ref{prp: Brouwerian}(i, ii) we get immediately the following fact.

\begin{proposition}\label{prp: Brouwerian2}
	If $\B L, \B M \leq \C H$, it cannot be accepted constructively that $\B L \wedge \B M$  or
	$\B L \vee \B M$ is total.
%	\footnote{In~\cite{BS00}, p.~505, it is mentioned that even if $\B L$ and $\B M$ are both  total, we cannot accept constructively that $\B L \wedge \B M$  or
	%	$\B L \vee \B M$ is total.}.
\end{proposition}

 The orthocomplemented subspaces of a Hilbert space together with its total subspaces form a \textit{constructive, complemented quantum lattice}, or a \textit{swap orthomodular lattice}, since the negation of $\B L$ is formed by swapping its components. 
 %The replacement of locatedness by totality will be also justified in section~\ref{sec: partialproj}, as the definition of the generally partial projection on a Hilbert space is easier to get constructively from the hypothesis of totality. 
 As in the theory of swap algebras (see~\cite{MWP24, MWP25, MWP25b}), where the classical properties of a Boolean algebra hold for the total elements of a swap algebra, the classical properties of $S(\C H)$ hold for its total (and located by Remark~\ref{rem: orthocs1}) elements. Next, follows the orthocomplemented, two-dimensional analague to Proposition~\ref{prp: CoQL}. Notice that the inequality on $ \B S_{\tota}(\C H)$ is the one inherited by the inequality on $\B S(\C H)$, defined in Theorem~\ref{thm: bijection}.

\begin{proposition}\label{prp: ComplQL}
	Let the structure $(\B S(\C H), \B S_{\tota}(\C H), \leq, -, \B 0, \B 1)$. If $\C H$ is strict, if $\B L, \B M \in \B S(\C H)$, and if $(\B L_i)_{i \in I}$ is a family of orthocomplemented subspaces of $\C H$ over the index-set $I$, then the following hold$:$\\[1mm]
	$(\ComplQL_0)$ $\B 0, \B 1 \in \B S_{\tota}(\C H)$ and $\B 0 \neq_{\mathsmaller{\B S_{\tota}(\C H)}} \B 1$.\\[1mm]
	$(\ComplQL_1)$ $\B L \wedge \B M$ $(\B L \vee \B M)$ is the greatest lower bound $($least upper bound$)$ of $\B L, \B M$ relative to $\leq$.\\[1mm]
	$(\ComplQL_{1{'}})$ $\bigwedge_{i \in I} \B L_i$ \hspace{-2.5mm} $\big(\bigvee_{i \in I} \B L_i\big)$ is the greatest lower bound $($least upper bound$)$ of $(\B L_i)_{\mathsmaller{i \in I}}$ relative to $\leq$.\\[1mm]
%	$(\ClQL_2)$ $0 \leq L \leq 1$.\\[1mm]
	$(\ComplQL_2)$ $\B 0 \leq \B L \leq \B 1$.\\[1mm]
	$(\ComplQL_3)$ If $\B L \leq \B M$, then $- \B M \leq - \B L$.\\[1mm]
	$(\ComplQL_4)$ $\B L =_{\mathsmaller{\B S(\C H)}} -(- \B L)$.\\[1mm]
	$(\ComplQL_5)$ $\B L \wedge (- \B L) := (0, \dom(\B L)) =: 0_{\B L}$ and $\B L \vee (- \B L) := (\dom(\B L), 0) =: 1_{\B L}$.\\[1mm]
	$(\ComplQL_6)$ If $\B L \in \B S_{\tota}(\C H)$, then $- \B L \in \B S_{\tota}(\C H)$.\\[1mm]
	$(\ComplQL_7)$ If $\B L, \B M \in \B S_{\tota}(\C H)$ and $\B L \leq (- \B M)$, then $\B L \vee \B M \in \B S_{\tota}(\C H)$.\\[1mm]
	$(\ComplQL_8)$ If $\B L \in \B S_{\tota}(\C H)$ and $\B L \leq \B M$, then $\B M =_{\mathsmaller{\B S(\C H)}} \B L \vee (\B M - \B L)$.
\end{proposition}

\begin{proof}
$(\ComplQL_0)$: By definition $\B 0 \neq_{\mathsmaller{\B S_{\tota}(\C H)}} \B 1 :\TOT P_{\B 0} \neq_{\mathsmaller{\M P(\C H)}} P_{\B 1}$, which holds by the strictness of $\C H$.\\
$(\ComplQL_1)-(\ComplQL_6)$ are straightforward to show. \\
$(\ComplQL_7)$: If $\B L, B M \in \B S_{\tota}(\C H)$, then by Corollary~\ref{cor: corthm} we have that $\B L =_{\B S(\C H)} (L, L^{\perp})$ and $\B M =_{\B S(\C H)} (M, M^{\perp})$. By the hypothesis $\B L \leq (- \B M)$ we get $L \leq M^{\perp}$ and $M \leq L^{\perp}$. Since $(L^{\perp} \wedge M^{\perp}) \perp (L \vee M)$ by Theorem~\ref{thm: crucial}(i) we get 
$$(L \vee M) \vee (L^{\perp} \wedge M^{\perp}) =_{\mathsmaller{S(\C H)}} (L \vee M) \oplus (L^{\perp} \wedge M^{\perp}) =_{\mathsmaller{S(\C H)}} \overline{L + M} \oplus (L^{\perp} \wedge M^{\perp}) \leq H.$$
Let $x \in H$. Since $H =_{\mathsmaller{S (\C H)}} L + L^{\perp} =_{\mathsmaller{S (\C H)}} M + M^{\perp}$, there are $l \in L, l^{\mathsmaller{\perp}} \in L^{\perp}$ and $m \in M, m^{\mathsmaller{\perp}} \in M^{\perp}$, such that 
$x =_H l + l^{\mathsmaller{\perp}} =_H m + m^{\mathsmaller{\perp}}$. Hence,
$L^{\perp} \ni l^{\mathsmaller{\perp}} - m =_H m^{\mathsmaller{\perp}} - l \in M^{\perp},$
i.e., $$l^{\mathsmaller{\perp}} - m =_H m^{\mathsmaller{\perp}} - l \in L^{\perp} \wedge M^{\perp}.$$
Moroever, we have that
$$x  =_H m + m^{\mathsmaller{\perp}} =_H (m + l) + (m^{\mathsmaller{\perp}} - l) =_H (l + m) + (m^{\mathsmaller{\perp}} - l) \in (L + M) + (L^{\perp} \wedge M^{\perp}),$$
this,
$$H \leq (L + M) + (L^{\perp} \wedge M^{\perp}) \leq \overline{L + M} \oplus (L^{\perp} \wedge M^{\perp}) \leq H,$$
and hence, $H =_{\mathsmaller{S (\C H)}} (L + M) + (L^{\perp} \wedge M^{\perp})$. Consequently, 
$\B L \vee \B M \in \B S_{\tota}(\C H)$.\\
$(\ComplQL_8)$: Let $\B L =_{\B S (\C H)} (L, L^{\perp})$ with $H =_{S (\C H)} L \oplus L^{\perp}$, $L \leq M^1$, and $M^0 \leq L^{\perp}$. We show that 
\begin{align*}
\B M & =_{\mathsmaller{\B S(\C H)}}  \B L \vee (\B M - \B L)\\
& =_{\mathsmaller{\B S(\C H)}}  (L, L^{\perp}) \vee [(M^1, M^0) \wedge (L^{\perp}, L)]\\
& =_{\mathsmaller{\B S(\C H)}}  (L, L^{\perp}) \vee (M^1 \wedge L^{\perp}, \overline{M^0 + L})\\
& =_{\mathsmaller{\B S(\C H)}}  (L \vee (M^1 \wedge L^{\perp}), L^{\perp} \wedge \overline{M^0 + L}).
\end{align*}
First, we show that $M^1 =_{\mathsmaller{S (\C H)}} \overline{L + (M^1 \wedge L^{\perp})}$. Clearly, $\overline{L + (M^1 \wedge L^{\perp})} \leq M^1$. For the converse inequality, if $m^1 \in M^1$, there are $l \in L$ and $l^{\mathsmaller{\perp}} \in L^{\perp}$, such that $m^1 =_H l + l^{\mathsmaller{\perp}}$.
Hence, $m^1 - l =_H l^{\mathsmaller{\perp}}$, i.e., $l^{\mathsmaller{\perp}} \in M^1 \wedge L^{\perp}$.
Consequently, $m^1 \in L + (M^1 \wedge L^{\perp})$, and hence $M^1 \leq  L + (M^1 \wedge L^{\perp})$.
Overall, we get
$$M^1 =_{\mathsmaller{S (\C H)}} L + (M^1 \wedge L^{\perp}) =_{\mathsmaller{S (\C H)}} \overline{L + (M^1 \wedge L^{\perp})}.$$
Next, we show that $M^0 =_{\mathsmaller{S (\C H)}} L^{\perp} \wedge \overline{M^0 + L}$. Clearly, $M^0 \leq L^{\perp} \wedge \overline{M^0 + L}$. For the converse inequality, if $x \in L^{\perp} \wedge \overline{M^0 + L}$, then $x \in L^{\perp}$ and $x \in \overline{M^0 + L}$. Let $(x_n)_{n \in \Nat} \subseteq M^0 + L$, such that $(x_n) \stackrel{n} \longrightarrow x$. Since $M^0 \leq L^{\perp}$, there are unique $m_n^0 \in M^0$ and $l_n \in L$, for every $n \in \Nat$, such that $x_n =_H m_n^0 + l_n$.
Since $L \leq M^1$, we get $L \perp M^0$, and working as in the proof of Theorem~\ref{thm: crucial}(i), the induced by $\MANC$ sequences $(l_n)_{n \in \Nat} \subseteq L$ and $(m_n^0)_{n \in \Nat} \subseteq M^0$ are Cauchy sequences, hence, there are $l \in L$ and $m^0 \in M^0$, such that 
$(l_n) \stackrel{n} \longrightarrow l$ and $(m_n^0) \stackrel{n} \longrightarrow m^0$. By the uniqueness of the limit we have that $x =_H m^0 + l$, i.e., $L^{\perp} \ni x - m^0 =_H l \in L$, and hence, $x - m^0 =_H l =_H 0$. As $x =_H m^0 \in M^0$, we get $x \in M^0$. 
\end{proof}

%Remark close to be but there is the problem with sup inf to take classical quantum logic as a limiting case...
%NOT EXACTLY!!!!!!!
%$\ComplQL$ has $\ClQL$ as a limiting case exactly as a Boolean algebra is a limiting case of a swap algebra. Namely, the total orhocomplemented subspaces form a amodel of $\ClQL$, exactly as the total elements of a swap algebra is a Boolean algebra (see~\cite{MWP24}, Proposition 7.10).
%This cannot be shown for the located subspaces of $\C H$, as the negation $L^{\perp}$ of $L$ does not behave classically, as the negation $-\B L$ of $\B L$ does.
%
%\begin{corollary}\label{cor: toalclass}
%$\B S_{\tot}(\C H)$ satisfies the axioms of the classical quantum logic $\ClQL$.
%\end{corollary}
%
%
%\begin{proof}
%$\CoQL_1$-$\CoQL_4$ correspond exactly to $\ClQL_1$-$\ClQL_4$.
%If $\B L$ is total, then $0_{\B L} =_{\mathsmaller{S(\C H)}} \B 0$ and $1_{\B L} =_{\mathsmaller{S(\C H)}} \B 1$, and hence the abstarct version of $\ClQL_5$ and $\ClQL_6$ are satisfied.
%
%\end{proof}
Next, we show an important corollary of $(\ComplQL_7)$ by inspection of its
proof. The equality in the following cases (ii) and (iii) is the equality of the set of functions from $H$ to $H$.

\begin{corollary}\label{cor: cor7}
If $\B L, \B M \in \B S_{\tota}(\C H)$ and $\B L \leq (- \B M)$, then the following hold$:$\\[1mm]
\normalfont (i)
\itshape $\overline{L + M} =_{\mathsmaller{S (\C H)}} L + M$.\\[1mm]
\normalfont (ii)
\itshape $P_{\B L} \circ P_{\B M} =_{\mathsmaller{\M P(\C H)}} 0$.\\[1mm]
\normalfont (iii)
\itshape $P_{\B L \vee \B M} =_{\mathsmaller{\M P(\C H)}} P_{\B L} + P_{\B M}$.
\end{corollary}

\begin{proof}
(i) By the proof of $(\ComplQL_7)$ we get the additional information $H =_{\mathsmaller{S (\C H)}} (L + M) + (L^{\perp} \wedge M^{\perp}) =_{\mathsmaller{S (\C H)}} \overline{L + M} \oplus (L^{\perp} \wedge M^{\perp})$, and hence $\overline{L + M} =_{\mathsmaller{S (\C H)}} L + M$.
To show this, let $x \in \overline{L + M}$. By the above decomposition of $H$ there are $l \in L, m \in M$ and $k^{\mathsmaller{\perp}} \in L^{\perp} \wedge M^{\perp}$, such that $x =_H (l + m) + k^{\mathsmaller{\perp}}$. Hence $ \overline{L + M} \ni x - (l + m) =_H k^{\mathsmaller{\perp}} \in L^{\perp} \wedge M^{\perp}$. Consequently, $x =_H l + m$ and $k^{\mathsmaller{\perp}} =_H 0$, i.e., 
$\overline{L + M} \leq L + M$. \\
(ii) By hypothesis $\dom(P_{\B L}) =_{\mathsmaller{S (\C H)}} \dom(\B M) =_{\mathsmaller{S (\C H)}} 1$, and hence $\dom(P_{\B L} \circ P_{\B M}) =_{\mathsmaller{S (\C H)}} 1$ too. If $x =_H l + l^{\mathsmaller{\perp}} =_H m + m^{\mathsmaller{\perp}}$, then $P_{\B L}(x) =_H l$ and $P_{\B M}(x) =_H m \in M \leq L^{\perp}$. Hence, $P_{\B L}(P_{\B M}(x)) =_H 0$.\\
(iii) By $(\ComplQL_7)$ $\dom(P_{\B L \vee \B M}) =_{\mathsmaller{S (\C H)}} 1 =_{\mathsmaller{S (\C H)}} 1 \wedge 1 =_{\mathsmaller{S (\C H)}} \dom(P_{\B L} + P_{\B M})$. 
By the equalities $H =_{\mathsmaller{S (\C H)}} (L + M) + (L^{\perp} \wedge M^{\perp})$ and 
$x =_H (l + m) + (m^{\mathsmaller{\perp}} - l)$ we get $P_{\B L \vee \B M}(x) =_H l + m =_H P_{\B L}(x) + P_{\B M}(x)$.
\end{proof}

Our constructive proof of Corollary~\ref{cor: cor7}(iii) does not depend on the classical equality $(L \wedge M)^{\perp} =_{\mathsmaller{S (\C H)}} L^{\perp} \vee M^{\perp}$, as the classical proof of Proposition 2.5.3 in~\cite{KR83}, p.~111, which is used in the proof of the corresponding Corollary 2.5.4 in~\cite{KR83}, p.~112. Because of this, our proof is interesting also from a classical point of view.
The proof of $(\ComplQL_8)$, which also requires $\MANC$, provides the additional information $L + (M^1 \wedge L^{\perp}) =_{\mathsmaller{S (\C H)}} \overline{L + (M^1 \wedge L^{\perp})}$. 

The following properties of the complemented quantum lattice are derived as the orthocomplemented, two dimensional analague to Proposition~\ref{prp: corCoQL}, and are straightforward to show. 
%As in the proof of Proposition\ref{prp: ComplQL}, the proofs of these properties within our two-dimensinal framework provide similar additional information for the related subspaces.

\begin{proposition}\label{prp: corComplQL}
	If $\B L, \B M \in \B S(\C H)$, then the following hold$:$\\[1mm]
	\normalfont (i)
	\itshape If $\B L \in \B S_{\tota}(\C H)$, $\B L \leq \B M$, and  $\B M - \B L =_{\mathsmaller{\B S(\C H)}} \B 0$, then $\B L =_{\mathsmaller{\B S(\C H)}} \B M$.\\[1mm]
	\normalfont (ii)
	\itshape $\B S(\C H)$ satisfies $(\ComplQL_1)$-$(\ComplQL_7)$, but we cannot accept in $\BISH$ that it satisfies $(\ComplQL_8)$.\\[1mm]
	\normalfont (iii)
	\itshape $\B 0 =_{\mathsmaller{\B S(\C H)}} - \B 1$ and $\B 1 =_{\mathsmaller{\B S(\C H)}} - \B 0$.\\[1mm] 
	\normalfont (iv)
%	\itshape $L^{\perp \perp \perp} =_{\mathsmaller{S(\C H)}} L^{\perp}$.\\[1mm]
%	\normalfont (v)
	\itshape $- (\B L \vee \B M) =_{\mathsmaller{\B S(\C H)}} (- \B L) \wedge (- \B M)$.\\[1mm]
	\normalfont (v)
	\itshape $- (\B L \wedge \B M) =_{\mathsmaller{\B S(\C H)}} (- \B L) \vee (- \B M)$.\\[1mm]
	\normalfont (vi)
%	\itshape If $\B L \in \B S_{\tota}(\C H)$, then $\B L =_{\mathsmaller{\B S(\C H)}} L^{\perp \perp}$.\\[1mm]
%	\normalfont (vii)
	\itshape If $\B L, \B M \in \B S_{\tota}(\C H)$ and $\B L \leq \B M$, then $\B M - \B L \in S_{\tota}(\C H)$.\\[1mm]
	\normalfont (vii)
	\itshape If $\B L, \B M \in S_{\tota}(\C H)$ and $\B L \leq (- \B M)$, then $\B M =_{\mathsmaller{\B S(\C H)}} (\B L \vee \B M) - \B L$.\\[1mm]
%	\normalfont (viii)
%	\itshape $\B 1, \B 0 \in \B S_{\tota}(\C H)$.\\[1mm]
	\normalfont (viii)
	\itshape If $\B L, \B M \in \B S_{\tota}(\C H)$ and $(- \B L) =_{\mathsmaller{\B S(\C H)}} (- \B M)$, then $\B L =_{\mathsmaller{\B S(\C H)}} \B M$.
\end{proposition}

%\begin{proof}
%(i)-(vi) and (vii)-(viii) are straightforward to show.
%(vi) Let $\B L := (L, L^{\perp})$ and $\B M := (M, M^{\perp})$, such that $L \leq M$ and $M^{\perp} \leq L^{\perp}$. Since $\B M - \B L := (M \wedge L^{\perp}, \overline{M^{\perp} + L})$, it suffices to show that 
%$$H \leq (M \wedge L^{\perp}) \vee \overline{M^{\perp} + L} =_{\mathsmaller{S(\C H)}} \overline{M \wedge L^{\perp}) \vee \overline{M^{\perp} + L}}.$$
%If $x \in H$, by the totality of $\B L$ and $\B M$ there are $l \in L, l^{\mathsmaller{\perp}} \in L^{\perp}, m \in M$, and $m^{\mathsmaller{\perp}} \in M^{\perp}$, such that $x =_H l + l^{\mathsmaller{\perp}} =_H m + m^{\mathsmaller{\perp}}$. Hence, we have that
%$$M \ni l - m =_H m^{\mathsmaller{\perp}} - l^{\mathsmaller{\perp}} \in L^{\perp},$$
%and hence $m^{\mathsmaller{\perp}} - l^{\mathsmaller{\perp}} \in M \wedge L^{\perp}$. Consequently, we get
%\begin{align*}
%x & =_H l + l^{\mathsmaller{\perp}} + m^{\mathsmaller{\perp}} - m^{\mathsmaller{\perp}}\\
%& =_H 
%\end{align*}
%\end{proof}

The totality of the meet $\B L \wedge \B M$ is established as follows.

\begin{corollary}\label{cor: wedgetotal}
	$(\ComplQL_{7'})$ If $\B L, \B M \in \B S_{\tota}(\C H)$ and $(- \B M) \leq \B L$, then $\B L \wedge \B M \in \B S_{\tota}(\C H)$.
\end{corollary}

\begin{proof}
By $(\ComplQL_6)$ and Proposition~\ref{prp: corComplQL}(v) it suffices to show that $- (\B L \wedge \B M) =_{\mathsmaller{S(\C H)}} (-\B L) \vee (-\B M) \in \B S_{\tota}(\C H)$. By hypothesis $(- \B M) \leq \B L
\TOT (-\B L) \leq -(-\B M) =_{\mathsmaller{S(\C H)}} \B M$, and by $(\ComplQL_7)$ we get $(-\B L) \vee (-\B M) \in \B S_{\tota}(\C H)$.
\end{proof}

The following properties are straightforward to show.

\begin{proposition}\label{prp: swapalg1}
	If $\B L \leq \C H$, then the following hold:\\[1mm]	
	%	$(1)$ \ $a \wedge a =_{\mathsmaller{A}} a$.\\[1mm]
	%	$(2)$ \ $a \wedge b =_{\mathsmaller{A}} b \wedge a$.\\[1mm]
	%	$(3)$ \ $a \wedge (b \wedge c) =_{\mathsmaller{A}} (a \wedge b) \wedge c$.\\[1mm]
	%	$(4)$ \ $a \wedge (b  \vee c) =_{\mathsmaller{A}} (a \wedge b) \vee (a \wedge c)$.\\[1mm]
	\normalfont (i)
	\itshape $-0_{\B L} =_{\mathsmaller{\B S(\C H)}} 1_{\B L}$.\\[1mm]
	%$(6)$ \ $a \wedge (-a) =_{\mathsmaller{A}} 0_a$. \\[1mm]
	\normalfont (ii)
	\itshape  $0_{\B L} =_{\mathsmaller{\B S(\C H)}} 0_{\B L}$.\\[1mm]
	\normalfont (iii)
	\itshape $0_{\B 0} =_{\mathsmaller{\B S(\C H)}} \B 0$ and $1_{\B 1} =_{\mathsmaller{\B S(\C H)}} \B 1$.\\[1mm]
	\normalfont (iv)
	\itshape $0_{0_{\B L}} =_{\mathsmaller{\B S(\C H)}} 0_{\B L} =_{\mathsmaller{\B S(\C H)}} 0_{1_{\B L}}$. \\[1mm]
	\normalfont (v)
	\itshape $\B 0 \vee \B L =_{\mathsmaller{\B S(\C H)}} {\B L}$. \\[1mm]
	\normalfont (vi)
	\itshape $0_{\B L} \wedge \B L =_{\mathsmaller{\B S(\C H)}} 0_{\B L}$. \\[1mm]
	\normalfont (vii)
	\itshape $\B 1 \wedge {\B L}  =_{\mathsmaller{\B S(\C H)}} {\B L}  =_{\mathsmaller{\B S(\C H)}} 1_{\B L} \wedge {\B L}$.\\[1mm]
	\normalfont (viii)
	\itshape $1_{\B L} \vee {\B L}=_{\mathsmaller{\B S(\C H)}} 1_{\B L}$.
\end{proposition}

Distributivity and modularity in $\B S(\C H)$ depends on the distributivity and modularity in $S(\C H)$, and hence they do not hold, in general.
As we remark next, condition $\B L \leq (- \B M)$ in $\ComplQL_7$ and in Corollary~\ref{cor: cor7} is the two-dimensional version of the one-dimensional relation $L \perp M$ on $S(\C H)$ and  the Hilbert 
space-analogue to the disjointness condition between complemented subsets
$$\B A \Disj \B B :\TOT \B A \subseteq (- \B B).$$
The proof of the following remark is straightforward.
%,  which implies $L \leq M^{\perp}$ and $M \leq L^{\perp}$.

\begin{remark}\label{rem: perp2}If $\B L, \B M, \B N \leq \C H$ and $(\B L_i)_{i \in I}$ is an $I$-family
	in $\B S(\C H)$, then the following hold$:$\\[1mm]
	\normalfont (i)
	\itshape If $\B L =_{\mathsmaller{\B S(\C H)}} (L, L^{\perp})$ and $\B M =_{\mathsmaller{\B S(\C H)}} (M, M^{\perp})$, then $\B L \perp \B M \TOT L \perp M$.\\[1mm]
	\normalfont (ii)
	\itshape If $\B L \perp \B M$, then $\B M \perp \B L$.\\[1mm]
	\normalfont (iii)
	\itshape If $\B L \perp \B M$ and $\B N \leq \B L$, then $\B N \perp \B M$.\\[1mm]
	\normalfont (iv) $\B L \perp \bigvee_{i \in I} \B L_i \TOT \forall_{i \in I}\big(\B L^1 \perp \B L_i)$.
\end{remark}

%\begin{proof}
%We show only case (iv).
%\end{proof}

Complemented quantum logic can be formulated constructively as follows.

\begin{definition}\label{def: ComplQL}
	A complemented quantum lattice is a structure $(\B {\C S}, \leq, -, 0, 1, (0_p)_{p \in \B {\C S}}, (1_p)_{p \in \B {\C S}})$, such that $(\B {\C S}, =_{\mathsmaller{\B {\C S}}}, \neq_{\mathsmaller{\B {\C S}}}, leq)$ is a poset, $0, 1 \in \B {\C S}$, $0_p, 1_p \in \B {\C S}$, for every $p \in \B {\C S}$, and $- \colon \B {\C S} \to \B {\C S}$. If $\B {\C S}_{\tota}$ is the subset of the total elements of $\B {\C S}$, i.e., $p \vee (-p) =_{\B {\C S}} 1$, for every $p \in \B {\C S}_{\tota}$, the following conditions hold$:$\\[1mm]
	$(\ComplQL_0)$ $0, 1 \in  \B {\C S}_{\tota}$ and $0 \neq_{\mathsmaller{\B {\C S}}} 1$.\\[1mm]
	$(\ComplQL_1)$ $p \wedge q$ $(p \vee q)$ is the least upper bound $($greatest lower bound$)$ of $p, q$ relative to $\leq$.\\[1mm]
	$(\ComplQL_2)$ $0 \leq p \leq 1$.\\[1mm]
	$(\ComplQL_3)$ If $p \leq q$, then\footnote{As in classical quantum logic, the converse implication follows by $(\ComplQL_4)$.} $-q \leq -p$.\\[1mm]
	$(\ComplQL_4)$ $p  =_{\mathsmaller{\B {\C S}}} -(-p)$.\\[1mm]
	$(\ComplQL_5)$ $p \wedge (-p) =_{\mathsmaller{\B {\C S}}} 0_p$ and $p \vee (-p) =_{\mathsmaller{\B {\C S}}} 1_p$.\\[1mm]
	$(\ComplQL_6)$ If $t \in \B {\C S}_{\tota}$, then $-t \in \B {\C S}_{\tota}$.\\[1mm]
	$(\ComplQL_7)$ If $t \in \B {\C S}_{\tota}$ and $t \leq (-p)$, then $t \vee p \in \B {\C S}_{\tota}$.\\[1mm]
	$(\ComplQL_8)$ If $t \in \B {\C S}_{\tota}$ and $t \leq p$, then $p =_{\mathsmaller{\B {\C S}}} t \vee (p-t)$, where $p-t := p \wedge (-t)$.\\[1mm]
	A complete complemented quantum lattice is defined in the obvious way.
\end{definition}

%If we require that $ \B {\C S}$ is equipped with a given inequality $\neq_{\mathsmaller{ \B {\C S}}}$, then one can also demand that $0 \neq_{\mathsmaller{ \B {\C S}}} 1$. 
Since $1 \in  \B {\C S}_{\tota}$,  we get immediately that $ \B {\C S}_{\tota}$ is inhabited. Clearly, we have the following.

\begin{remark}\label{remL comparison} A complemented quantum lattice satisfies all properties of an intuitionistic quantum lattice, except $(\CoQL_5)$. 
\end{remark}

As swap algebras are the constructive and complemented analogue to Boolean algebras, in analogy to Example~\ref{ex: BA}, we have the following example of a complemented quantum logic.

\begin{example}\label{ex: swapI}
	\normalfont
	Every swap algebra of type $(\ti)$ is a distributive(see~\cite{MWP24}), complemented quantum logic. Since a Boolean algebra is a swap algebra of type $(\ti)$, a Boolean algebra is also a distributive, complemented quantum logic.
\end{example}

Using the bijection between $\B S(\C H)$ and $\M P(\C H)$ in Theorem~\ref{thm: bijection}, the operations on $\B S(\C H)$ induce the corresponding operations on $\M P(\C H)$.

\begin{definition}\label{def: operationsproj} Let the functions $i \colon \B S(\C H) \to \M P(\C H)$ and $j \colon  \M P(\C H) \to \B S(\C H)$, defined in Theorem~\ref{thm: bijection}. The following operations are defined on {$\M P(\C H)$}$:$
	$$1 := I_H,$$
	$$0 := 0_H,$$
	$$P \wedge Q := i\big(j(P) \wedge j(Q)\big),$$
	$$P \vee Q := i\big(j(P) \vee j(Q)\big),$$
	$$\sim{\hspace{-1mm}}P := 1 - P,$$
	$$P \sim Q := P \wedge (\sim{\hspace{-1mm}}Q),$$
	$$P \To Q := (\sim{\hspace{-1mm}}P) \vee Q,$$
	$$P \TOT Q := (P \To Q) \wedge (Q \To P),$$ 
	$$\neg P := P \To 0.$$
	If $(P_k)_{k \in K}$ is a family of partial projections of $\C H$ over the index-set $K$, let
	$$\bigwedge_{k \in K} P_k := i\bigg(\bigwedge_{k \in k}j(P_k)\bigg),$$
	$$\bigvee_{k \in K} P_k := i\bigg(\bigvee_{k \in k}j(P_k)\bigg).$$
	The relation of orthogonality on $\M P (\C H)$ is defined by the rule
	$$P \perp Q :\TOT P \leq (\sim{\hspace{-1mm}} Q) :\TOT j(P) \leq j(\sim{\hspace{-1mm}} Q).$$
	If $L \leq \C H$ is understood as a proposition, its derivability $\vdash P$ is defined as a proof of $P =_{\mathsmaller{\M P(\C H)}} 1$.
\end{definition}

\begin{proposition}\label{prp: projiscomplql}
The structure $(\M P(\C H), \M P_{\tota}(\C H), \wedge, \vee, \sim, 0, 1)$ is a complemented quantum logic.
\end{proposition}

\begin{proof}
All properties of a complemented quantum lattice follow by the definition of the operations on $\M P(\C H)$, the corresponding properties of $\B S(\C H)$, and the properties of the functions $i$ and $j$ in Theorem~\ref{thm: bijection}. For example, we show that $P \wedge Q$ is the greatest lower bound of $P$ and $Q$. First, we show that $P \wedge Q \leq P$. By definition of the inequality on $\M P(\C H)$ we have that
\begin{align*}
P \wedge Q \leq P & :\TOT j(P \wedge Q) \leq j(P) \\
& :\TOT j\big(i(j(P) \wedge j(Q))\big) \leq j(P) \\
& \TOT j(P) \wedge j(Q) \leq j(P).
\end{align*}
Similarly, 
we have that 
$P \wedge Q \leq Q$. If $R \leq P$ and $R \leq Q$, then we show that $R \leq P \wedge Q$. By definition we have that $j(R) \leq j(P)$ and $j(R) \leq j(Q)$, hence 
$j(R) \leq j(P) \wedge j(Q)$ Since $i$ is monotone, we get $i(j(R)) =_{\mathsmaller{\M P(\C H)}} R \leq i\big(j(P) \wedge j(Q)\big) =: P \wedge Q$. For the remaining properties we work similarly.
\end{proof}

\begin{proposition}\label{prp: ij}
Let the functions $i \colon \B S(\C H) \to \M P(\C H)$ and $j \colon  \M P(\C H) \to \B S(\C H)$, defined in Theorem~\ref{thm: bijection}. The following hold$:$\\[1mm]
\normalfont (i)
\itshape $i(\B 1) =_{\mathsmaller{\M P(\C H)}} 1$ and $j(1) =_{\mathsmaller{\B S(\C H)}} \B 1$.\\[1mm]
\normalfont (ii)
\itshape $i(\B 0) =_{\mathsmaller{\M P(\C H)}} 0$ and $j(0) =_{\mathsmaller{\B S(\C H)}} \B 0$.\\[1mm]
\normalfont (iii)
\itshape $i(- \B L) =_{\mathsmaller{\M P(\C H)}} \sim{\hspace{-1mm}}i(\B L)$ and $j(\sim{\hspace{-1mm}}P) =_{\mathsmaller{\B S(\C H)}} -j(P)$.\\[1mm]
\normalfont (iv)
\itshape $i(\B L \wedge \B M) =_{\mathsmaller{\M P(\C H)}} i(\B L) \wedge i(\B M)$ and $j(P \wedge Q) =_{\mathsmaller{\B S(\C H)}} j(P) \wedge j(Q)$.\\[1mm]
\normalfont (v)
\itshape $i(\B L \vee \B M) =_{\mathsmaller{\M P(\C H)}} i(\B L) \vee i(\B M)$ and $j(P \vee Q) =_{\mathsmaller{\B S(\C H)}} j(P) \vee j(Q)$.\\[1mm]
\normalfont (vi)
\itshape $i(\B L - \B M) =_{\mathsmaller{\M P(\C H)}} i(\B L) \sim i(\B M)$ and $j(P \sim Q) =_{\mathsmaller{\B S(\C H)}} j(P) - j(Q)$.\\[1mm]
\normalfont (vii)
\itshape $i(\B L \To \B M) =_{\mathsmaller{\M P(\C H)}} i(\B L) \To i(\B M)$ and $j(P \To Q) =_{\mathsmaller{\B S(\C H)}} j(P) \To j(Q)$.\\[1mm]
\normalfont (viii)
\itshape $i( \neg \B L) =_{\mathsmaller{\M P(\C H)}} \neg i(\B L)$ and $j(\neg P) =_{\mathsmaller{\B S(\C H)}} \neg j(P)$.\\[1mm]
\normalfont (ix)
\itshape $i\bigg(\bigwedge_{k \in K}\B L_k\bigg) =_{\mathsmaller{\M P(\C H)}} \bigwedge_{k \in K}i(\B L_k)$ and $j\bigg(\bigwedge_{k \in K}P_k\bigg) =_{\mathsmaller{\B S(\C H)}} \bigwedge_{k \in K}j(P_k)$.\\[1mm]
\normalfont (ix)
\itshape $i\bigg(\bigvee_{k \in K}\B L_k\bigg) =_{\mathsmaller{\M P(\C H)}} \bigvee_{k \in K}i(\B L_k)$ and $j\bigg(\bigvee_{k \in K}P_k\bigg) =_{\mathsmaller{\B S(\C H)}} \bigvee_{k \in K}j(P_k)$.
\end{proposition}

\begin{proof}
We prove only case (iv), and for the remaining cases we work similarly. As $i$ and $j$ are inverse to each other, we have that
$$i(\B L) \wedge i(\B M) := i\big(j(i(\B L)) \wedge j(i(\B M))\big)=_{\mathsmaller{\M P(\C H)}} i(\B L \wedge \B M),$$
$$\ \ \ \ \ \ \ \ \ \ \ \ \ \ \ \ \ \ \ \ \ \ \ \ \ \ \ \ \ \ \ \ \ j(P \wedge Q) := j\big(i\big(j(P) \wedge j(Q)\big)\big) =_{\mathsmaller{\B S(\C H)}} j(P) \wedge j(Q). \ \ \ \ \ \ \ \ \ \ \ \ \ \ \ \ \ \ \ \ \ \ \ \ \ \ \ 
\qedhere$$
\end{proof}

\section{Commuting partial projections}
\label{sec: commuting}

If $P \colon L_P^1 \oplus L_P^0 \to L_P^1$ and $Q \colon L_Q^1 \oplus L_Q^0 \to L_Q^1$ are in $\M P(\C H)$, the domains of their compositions are 
\begin{align*}
\dom(Q \circ P) &:= \big\{x \in L_P^1 \oplus L_P^0 \mid P(x) \in L_Q^1 \oplus L_Q^0\big\}\\
& =_{\mathsmaller{S(\C H)}} \big\{l_P^1 + l_P^0 \in  L_P^1 \oplus L_P^0 \mid \exists_{l_Q^1 \in L_Q^1}\exists_{l_Q^0 \in L_Q^0}(l_P^1 =_H l_Q^1 + l_Q^0)\big\}\\
& \leq \dom(P),
\end{align*}
%and
\begin{align*}
	\dom(P \circ Q) &:= \big\{x \in L_Q^1 \oplus L_Q^0 \mid Q(x) \in L_P^1 \oplus L_P^0\big\}\\
	& =_{\mathsmaller{S(\C H)}} \big\{l_Q^1 + l_Q^0 \in  L_Q^1 \oplus L_Q^0 \mid \exists_{l_P^1 \in L_P^1}\exists_{l_P^0 \in L_P^0}(l_Q^1 =_H l_P^1 + l_P^0)\big\}\\
	& \leq \dom(Q).
\end{align*}
If $P$ and $Q$ commute, i.e., $Q \circ P =_{\mathsmaller{\M P(\C H)}} P \circ Q$, then by definition $\dom(Q \circ P) =_{\mathsmaller{S(\C H)}} \dom(P \circ Q)$ and for every $x \in \dom(Q \circ P)$ we have that $Q(P(x)) =_H P(Q(x))$. It is straightforward to show that in this case $Q \circ P$ is a partial projection on $\C H$, and 
$$\dom(Q \circ P) \leq \dom(P) \wedge \dom(Q).$$
If $x \in \dom(Q \circ P) =_{\mathsmaller{S(\C H)}} \dom(P \circ Q)$, then 
\begin{align*}
x & =_H l_P^1 + l_P^0 \in \dom(Q \circ P)\\
& =_H (l_Q^1 + l_Q^0) + l_P^0\\
& =_H \lambda_Q^1 + \lambda_Q^0 \in \dom(P \circ Q)\\
& =_H (\lambda_P^1 + \lambda_P^0) + \lambda_Q^0,
\end{align*}
with
$$Q(P(x)) =_H l_Q^1, \ \ \ P(Q(x)) =_H \lambda_P^1,$$
and the equality $Q(P(x)) =_H P(Q(x))$ implies that 
$$l_Q^1 =_H \lambda_P^1 \in L_Q^1 \wedge L_P^1.$$
Next, we show the constructive and partial version of Proposition 2.5.3 in~\cite{KR83}, pp.~111-112. The classical proof of Proposition~\ref{prp: comm1}(iii) is based on the classical property $(\BC_3)$ of Proposition~\ref{prp: Brouwerian}, which is avoided here. The proof of Proposition~\ref{prp: Brouwerian}(iv) in~\cite{KR83} is adopted here accordingly, as complemented negation behaves ``classically''.

\begin{proposition}\label{prp: comm1}
If $P, Q$ are commuting partial projections on $\C H$, then the following hold$:$\\[1mm]
\normalfont(i)
\itshape $\dom(P \circ Q) =_{\mathsmaller{S(\C H)}} (L_P^1 \wedge L_Q^1) \oplus (L_P^0 \vee L_Q^0)$.\\[1mm]
\normalfont(ii)
\itshape $L_P^0 \vee L_Q^0 =_{\mathsmaller{S(\C H)}} L_P^0 + L_Q^0$.\\[1mm]
\normalfont(iii)
\itshape $P \wedge Q =_{\mathsmaller{\M P(\C H)}} Q \circ P$.\\[1mm]
\normalfont(iv)
\itshape If $\dom[(\sim{\hspace{-1mm}}Q) \circ (\sim{\hspace{-1mm}}P)] =_{\mathsmaller{S(\C H)}} \dom[(\sim{\hspace{-1mm}}P) \circ (\sim{\hspace{-1mm}}Q)]$, then 
$$P \vee Q =_{\mathsmaller{\M P(\C H)}} P + Q - (Q \circ P) \ \ \ \& \ \ \ L_P^1 \vee L_Q^1 =_{\mathsmaller{S(\C H)}} L_P^1 + L_Q^1.$$
\end{proposition}

\begin{proof}
(i) By definition $j(P) \wedge j(Q)  := \B L_P \wedge \B L_Q := \big(L_P^1 \wedge L_Q^1, L_P^0 \vee L_Q^0\big)$, where
\begin{align*}
L_P^1 \wedge L_Q^1 & := \{x \in \dom(P) \mid P(x) =_H x\} \wedge \{y \in \dom(Q) \mid Q(y) =_H y\} \\
& =_{\mathsmaller{S(\C H)}} \{z \in \dom(P) \wedge \dom(Q) \mid P(z) =_H z =_H Q(z)\}, 
\end{align*}
and $L_P^0 \vee L_Q^0 := \overline{L_P^0 + L_Q^0}$. Next, we show that 
$$(L_P^1 \wedge L_Q^1) \oplus (L_P^0 \vee L_Q^0) \leq \dom(P \circ Q).$$
If $z \in L_P^1 \wedge L_Q^1$, then $z \in \dom(P) \wedge \dom(Q)$ with $P(z) =_H z \in \dom(Q)$, hence $z \in \dom(Q \circ P)$. Consequently, $L_P^1 \wedge L_Q^1 \leq \dom(P \circ Q)$. Since
$$L_P^0 := \{x \in \dom(P) \mid P(x) =_H 0 \in \dom(Q)\} \leq \dom(Q \circ Q),$$
$$L_Q^0 := \{x \in \dom(Q) \mid Q(x) =_H 0 \in \dom(P)\} \leq \dom(P \circ Q),$$
and as $\dom(P \circ Q)$ is closed, we get $L_P^0 \vee L_Q^0 \leq \dom(P \circ Q)$. To show the inequality
$$\dom(P \circ Q) \leq (L_P^1 \wedge L_Q^1) \oplus (L_P^0 \vee L_Q^0),$$
let $x \in \dom(P \circ Q)$. Since $x =_H \lambda_P^1 + (\lambda_P^0 + \lambda_Q^0)$, for some $ \lambda_P^1 \in  L_P^1 \wedge L_Q^1$, $ \lambda_P^0 \in L_P^0$, and $ \lambda_Q^0 \in L_Q^0$, we have that
\begin{align*}
(L_P^1 \wedge L_Q^1) \oplus (L_P^0 + L_Q^0) & \leq (L_P^1 \wedge L_Q^1) \oplus (L_P^0 \vee L_Q^0)\\
& =_{\mathsmaller{S(\C H)}} \dom(j(P) \wedge j(Q))\\
& \leq \dom(P \circ Q)\\
& \leq (L_P^1 \wedge L_Q^1) \oplus (L_P^0 + L_Q^0).
\end{align*}
Hence, 
$$(L_P^1 \wedge L_Q^1) \oplus (L_P^0 + L_Q^0) =_{\mathsmaller{S(\C H)}} (L_P^1 \wedge L_Q^1) \oplus (L_P^0 \vee L_Q^0) =_{\mathsmaller{S(\C H)}} \dom(P \circ Q).$$
(ii) By the previous equality we get easily\footnote{If $K, L$ are linear subspaces of $\C H$, such that $K \oplus L =_{\mathsmaller{S(\C H)}} K \oplus \overline{L}$, then it is immediate to show that $L =_{\mathsmaller{S(\C H)}} \overline{L}$.} that $L_P^0 + L_Q^0 =_{\mathsmaller{S(\C H)}} L_P^0 \vee L_Q^0$.\\
(iii) By the definition of $i(\B L)$, the domain of $P \wedge Q$ is the domain of $j(P) \wedge j(Q)$, hence, by case (i) we get $\dom(P \wedge Q) =_{\mathsmaller{S(\C H)}} \dom(Q \circ P)$. Since
$$P \wedge Q \colon (L_P^1 \wedge L_Q^1) \oplus (L_P^0 + L_Q^0) \to L_P^1 \wedge L_Q^1,$$
$$l_{PQ}^1 + (l_P^0 + l_Q^0) \mapsto l_{PQ}^1,$$
we also have that
\begin{align*}
(Q \circ P)(l_{PQ}^1 + (l_P^0 + l_Q^0)) & =_H Q\big(l_{PQ}^1 + 0 + P(l_Q^0)\big)\\
& =_H Q\big(l_{PQ}^1\big) + 0 + Q\big(P(l_Q^0)\big)\\
& =_H l_{PQ}^1 + P\big(Q(l_Q^0)\big)\\
& =_H  l_{PQ}^1 + P(0)\\
& =_H l_{PQ}^1.
\end{align*}
(iv) The equality $P(Q(x)) =_H Q(P(x))$, for every $x \in \dom(Q \circ P)$ implies that $(\sim{\hspace{-1mm}}P)((\sim{\hspace{-1mm}}Q)(x)) =_H (\sim{\hspace{-1mm}}Q)((\sim{\hspace{-1mm}}P)(x))$, for every $x \in \dom((\sim{\hspace{-1mm}}Q) \circ (\sim{\hspace{-1mm}}P))$, and since the corresponding domains are by hypothesis equal, we get 
that $\sim{\hspace{-1mm}}P$ and $\sim{\hspace{-1mm}}Q$ commute. By
 %the property of 
 Proposition~\ref{prp: corComplQL}(iv) for partial projections we have that $P \vee Q =_{\mathsmaller{\M P(\C H)}} \sim{\hspace{-1mm}}[(\sim{\hspace{-1mm}}P) \wedge (\sim{\hspace{-1mm}}Q)]$. Using case (iii) we get $P \vee Q =_{\mathsmaller{\M P(\C H)}} \sim{\hspace{-1mm}}[(I-P) \circ (I - Q)] =_{\mathsmaller{\M P(\C H)}} P + Q - (Q \circ P)$. Since $L_{\sim{\hspace{-0.4mm}}P}^1 =_{\mathsmaller{S(\C H)}} L_P^0$ and $L_{\sim{\hspace{-0.4mm}}P}^0 =_{\mathsmaller{S(\C H)}} L_P^1$, the equality $L_P^1 \vee L_Q^1 =_{\mathsmaller{S(\C H)}} L_P^1 + L_Q^1$ follows immediately from case (ii).
\end{proof}

\end{document}